\documentclass[11pt]{article}

\usepackage[margin=1in]{geometry}

\usepackage{amsfonts}
\usepackage{amsmath}
\usepackage{amsthm}
\usepackage{amssymb}
\usepackage{mathtools}
\usepackage{bm}
\usepackage{bbm}
\usepackage{mathrsfs}
\usepackage{dsfont}
\usepackage{upgreek}
\usepackage{xparse} 
\usepackage{graphicx}
\usepackage{epsfig}
\usepackage{float}
\usepackage{rotating}
\usepackage{wrapfig}
\usepackage[ruled,vlined]{algorithm2e}

\usepackage{indentfirst}
\usepackage{verbatim}
\usepackage[makeroom]{cancel}
\usepackage{framed}
\usepackage{xcolor}
\usepackage{enumitem}
\usepackage[normalem]{ulem}
\usepackage{soul}
\usepackage{lineno}

\usepackage[font=footnotesize,labelfont=bf]{caption}
\usepackage[colorinlistoftodos,prependcaption,textsize=small]{todonotes}

\usepackage{abstract}
\usepackage[backref=page]{hyperref}
\usepackage[capitalize]{cleveref}

\definecolor{corlinks}{RGB}{0,0,180}
\hypersetup{
  colorlinks=true,
  linkcolor=corlinks,
  citecolor=corlinks,
  urlcolor=corlinks,
  pdfborder={0 0 0}
}

\newcommand{\ket}[1]{|#1\rangle}
\newcommand{\bra}[1]{\langle#1|}

\newcommand{\ketbra}[2]{|#1\rangle\! \langle #2|}
\newcommand{\proj}[1]{|#1\rangle\!\langle #1|}

\newcommand{\Tr}{\operatorname{Tr}}
\newcommand{\tr}[1]{\operatorname{Tr}\left(#1\right)}

\newcommand{\polylog}{\operatorname{polylog}}

\def\01{\{0,1\}}

\newcommand{\diag}{\operatorname{diag}}

\newcommand{\A}{\mathcal A}

\newcommand{\F}{\mathcal{F}}

\newcommand{\D}{\mathrm{D}}

\newcommand{\id}{\mathrm{id}}

\newcommand{\regret}{\mathrm{Regret}}

\DeclarePairedDelimiterX{\norm}[1]{\lVert}{\rVert}{#1}

\newtheorem{theorem}{Theorem}
\numberwithin{theorem}{section}

\newtheorem{definition}[theorem]{Definition}
\newtheorem{lemma}[theorem]{Lemma}
\newtheorem{corollary}[theorem]{Corollary}
\newtheorem{proposition}[theorem]{Proposition}

\newtheorem{assumption}[theorem]{Assumption}

\newtheorem*{proposition*}{Proposition}

\theoremstyle{definition}
\newtheorem{remark}[theorem]{Remark}

\def\01{\{0,1\}}

\DeclareDocumentCommand{\dist}{o}{%
  \IfNoValueTF{#1}{d}{d_{\mathrm{#1}}}%
}

\NewDocumentCommand{\Prob}{e{_} m}{%
  \IfNoValueTF{#1}{%
    \Pr \set*{#2}
  }{%
    \Pr_{#1} \set*{#2}
  }%
}

\newcommand{\dd}{\mathrm{d}}

\newcommand{\beq}{\begin{equation}}
\newcommand{\beql}[1]{\begin{equation}\label{#1}}
\newcommand{\eeq}{\end{equation}}
\newcommand{\eeqp}{\,\,\,.\end{equation}}
\newcommand{\eeqc}{\,\,\,,\end{equation}}

\DeclareMathOperator{\poly}{poly}

\DeclareMathOperator{\Exp}{\mathbb{E}}

\usepackage{tikz,tikz-qtree}
\usetikzlibrary{arrows}

\begin{document}

\title{Reinforcement learning for quantum processes with memory}

\renewcommand{\thefootnote}{\fnsymbol{footnote}}

\author{
Josep Lumbreras\textsuperscript{* 1,2}\\
\href{mailto:josep.lz@ntu.edu.sg}{josep.lz@ntu.edu.sg} 
\and
Ruo Cheng Huang\textsuperscript{* 1,2}\\
\href{mailto:ruocheng001@e.ntu.edu.sg}{ruocheng001@e.ntu.edu.sg}
\and
Yanglin Hu\textsuperscript{* 3}\\
\href{mailto:yanglin.hu@u.nus.edu}{yanglin.hu@u.nus.edu}
\and
Marco Fanizza\textsuperscript{4} \\
\href{mailto:marco.fanizza@inria.fr}{marco.fanizza@inria.fr}
\and
Mile Gu\textsuperscript{1,2,5} \\
\href{mailto:mgu@quantumcomplexity.org}{mgu@quantumcomplexity.org}
}

\date{\today}

\maketitle
\thispagestyle{empty}
\begingroup
\renewcommand\thefootnote{}
\footnotetext{$^*$ J. Lumbreras, R. C. Huang, and Y. Hu contributed equally and share first authorship.}
\footnotetext{$^1$ Nanyang Quantum Hub, School of Physical and Mathematical Sciences, Nanyang Technological University, Singapore}
\footnotetext{$^2$ Centre for Quantum Technologies, Nanyang Technological University, Singapore}
\footnotetext{$^3$ QICI Quantum Information and Computation Initiative, School of Computing and Data Science, The University of Hong Kong, Pokfulam Road, Hong Kong SAR, China}
\footnotetext{$^4$ Inria, Télécom Paris - LTCI, Institut Polytechnique de Paris}
\footnotetext{$^5$ MajuLab, CNRS-UNS-NUS-NTU International Joint Research Unit, UMI 3654, Singapore 117543, Singapore}
\endgroup

\begin{abstract}
In reinforcement learning, an agent interacts sequentially with an environment to maximize a reward, receiving only partial, probabilistic feedback. This creates a fundamental exploration-exploitation trade-off: the agent must explore to learn the hidden dynamics while exploiting this knowledge to maximize its target objective. While extensively studied classically, applying this framework to quantum systems requires dealing with hidden quantum states that evolve via unknown dynamics. We formalize this problem via a framework where the environment maintains a hidden quantum memory evolving via unknown quantum channels, and the agent intervenes sequentially using quantum instruments. For this setting, we adapt an optimistic maximum-likelihood estimation algorithm. We extend the analysis to continuous action spaces, allowing us to model general positive operator-valued measures (POVMs). By controlling the propagation of estimation errors through quantum channels and instruments, we prove that the cumulative regret of our strategy scales as $\widetilde{\mathcal{O}}(\sqrt{K})$ over $K$ episodes. Furthermore, via a reduction to the multi-armed quantum bandit problem, we establish information-theoretic lower bounds demonstrating that this sublinear scaling is strictly optimal up to polylogarithmic factors. As a physical application, we consider state-agnostic work extraction. When extracting free energy from a sequence of non-i.i.d. quantum states correlated by a hidden memory, any lack of knowledge about the source leads to thermodynamic dissipation. In our setting, the mathematical regret exactly quantifies this cumulative dissipation. Using our adaptive algorithm, the agent uses past energy outcomes to improve its extraction protocol on the fly, achieving sublinear cumulative dissipation, and, consequently, an asymptotically zero dissipation rate.
\end{abstract}

\renewcommand{\thefootnote}{\arabic{footnote}}
\setcounter{footnote}{0}

\setlength{\parskip}{4pt}
\setlength{\parindent}{7mm}

\newpage

\tableofcontents

\section{Introduction}
\label{sec:intro}

In reinforcement learning (RL), an agent interacts sequentially with an unknown environment to learn optimal decision-making~\cite{sutton1998reinforcement}. The primary mechanism driving this learning process is the maximization of a cumulative reward signal, a versatile formulation capable of encoding complex tasks such as the control of autonomous robots~\cite{kober2013reinforcement}, autonomous navigation~\cite{levine2016end}, and industrial resource optimization~\cite{mirhoseini2021graph,mao2016resource}. However, in most realistic physical settings, such as a robot navigating a cluttered, partially mapped physical space, the agent lacks direct access to the true underlying state of the system. Such scenarios are traditionally modeled as Partially Observable Markov Decision Processes (POMDPs)~\cite{kaelbling1998planning}. This partial observability severely complicates learning, as the agent must maintain a history of past interactions to infer the latent dynamics while simultaneously balancing the fundamental \emph{exploration-exploitation trade-off}. While the theoretical understanding of exploration in classical POMDPs has seen significant recent breakthroughs~\cite{jin2020sample,liu2022partially}, these classical frameworks are fundamentally incompatible with quantum physical systems.

As quantum technologies mature, there is a critical need for agents capable of learning and controlling quantum processes. Reinforcement learning has emerged as a highly effective tool for this task, demonstrating significant success across several fundamental domains. In quantum control and state preparation, RL agents have been deployed to discover optimal control pulses and stabilize complex target states, frequently outperforming traditional gradient-based optimal control methods by learning robust policies directly from environmental interaction~\cite{Andreasson2019quantumerror,zhang2019does}. In quantum circuit compilation, RL is utilized to autonomously generate hardware-efficient gate sequences and routing strategies that minimize circuit depth and limit decoherence~\cite{quetschlich2023compiler}. Most notably, RL has driven massive breakthroughs in quantum error correction. While foundational works established the ability of RL to autonomously discover novel error-correcting codes and decoding strategies~\cite{Andreasson2019quantumerror}, recent landmark results by Google Quantum AI have repurposed the error correction cycle itself as an RL environment~\cite{sivak2026reinforcement}. By utilizing binary error detection events as learning signals, these agents continuously tune thousands of physical control parameters on the fly, actively stabilizing logical qubits against environmental drift.

Despite these impressive application-specific successes, learning and controlling general processes in these settings remains intrinsically difficult: unlike classical probability distributions, hidden quantum states evolve via completely positive trace-preserving (CPTP) channels, and the agent's very act of gathering information inevitably disturbs the latent state. Consequently, the canonical reinforcement learning problem, where an agent sequentially interacts with an unknown evolving environment to maximize a cumulative reward, lacks a rigorous formulation when the environment itself is quantum. It is fundamentally natural to consider a classical agent probing an evolving quantum process to dynamically extract a continuous resource (Figure~\ref{fig:cartoon}). State-agnostic work extraction is a key example, which has been successfully explored in the i.i.d. setting~\cite{lumbreras2025quantum,watanabe2026universal}. However, realistic physical sources rarely emit uncorrelated states. Instead, they exhibit temporal correlations driven by an inaccessible latent memory. This necessitates a shift to the non-i.i.d. regime, where an agent seeks to harvest thermodynamic free energy from a sequence of correlated quantum states. In such settings, the agent must continually balance learning the underlying emission dynamics with maximizing its physical energy yield. This physical scenario exemplifies the broader challenge of sequential reward maximization under unknown quantum dynamics, motivating our central question:

\begin{center}
    \textit{Can a classical agent efficiently resolve the fundamental exploration-exploitation trade-off over an uncharacterized quantum stochastic process?}
\end{center}

To rigorously formalize the problem of maximizing rewards over quantum processes, we introduce the input-output \emph{Quantum Hidden Markov Model} (QHMM). This framework serves as a true POMDP-style generalization for quantum systems, treating the latent transition dynamics as strictly unknown and learned purely through interaction. Furthermore, it acts as the natural, stateful extension of the Multi-Armed Quantum Bandit (MAQB) framework~\cite{lumbreras2022multi}. In our episodic QHMM setting, the environment carries a finite-dimensional quantum memory that evolves in time. At each step, the agent probes the system by selecting an action corresponding to a quantum instrument, producing a classical outcome (reward) and a corresponding post-measurement memory state.

While the QHMM faithfully captures the memory of the quantum process, learning and controlling its dynamics directly poses a fundamental challenge: the underlying quantum state driving this memory is inherently unobservable. To overcome this limitation, we adopt the framework of Observable Operator Models (OOMs). Rather than attempting to track the inaccessible internal quantum dynamics, the OOM formalism allows us to express the system's evolution entirely in terms of observable trajectory probabilities. By translating the unobservable physical dynamics of the QHMM strictly into these observable statistics, we establish a rigorous foundation not just for learning the underlying process, but for designing algorithmic protocols that systematically resolve the exploration-exploitation trade-off.

By establishing this general mathematical framework, we can then directly resolve the physical work extraction problem motivated above. When an agent attempts to extract free energy from the non-i.i.d.\ sequence generated by a QHMM, any lack of knowledge regarding the hidden memory leads to a sub-optimal extraction protocol. By the laws of thermodynamics, this informational deficit inevitably results in the irreversible dissipation of energy~\cite{kolchinsky2017dependence,riechers2021initial}. Crucially, in our sequential learning framework, the mathematical regret incurred by the agent's policy exactly quantifies this cumulative thermodynamic dissipation. Consequently, designing a regret-minimizing RL algorithm is physically equivalent to engineering an adaptive Maxwell's demon that learns to minimize energy dissipation on the fly.

\paragraph{Organization of the paper.}
The remainder of this manuscript is organized as follows. Section~\ref{sec:results} outlines our main theoretical results and thermodynamic applications, followed by a broader discussion of future directions in Section~\ref{sec:discussion} and related literature in Section~\ref{sec:related_work}. We formally define the input-output QHMM framework and the sequential learning objective in Section~\ref{sec:model}. Section~\ref{sec:oom} introduces the undercomplete assumption and constructs the parameter-free Observable Operator Model (OOM). Section~\ref{sec:omle_discrete} details the Optimistic Maximum Likelihood Estimation (OMLE) algorithm for this quantum setting. We establish the end-to-end $\widetilde{\mathcal{O}}(\sqrt{K})$ regret bounds for discrete action spaces in Section~\ref{sec:regret_discrete}, and extend this mathematical machinery to continuous measurement spaces in Section~\ref{sec:cont-actions}. In Section~\ref{sec:lower_bounds}, we prove the statistical optimality of these bounds via information-theoretic lower bounds. Finally, Section~\ref{sec:work_extraction} maps this framework to state-agnostic work extraction, and Section~\ref{sec:numerics} corroborates our analytical bounds with numerical simulations.

\begin{figure}
    \centering
    \includegraphics[width=0.85\linewidth]{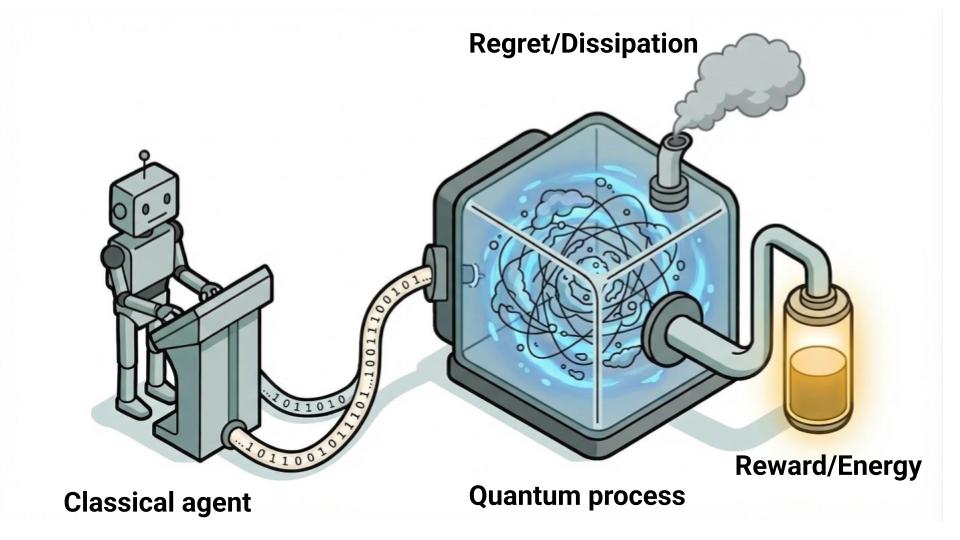}
    \caption{\textbf{The Classical-Quantum RL Interface}. Scheme of the reinforcement learning framework explored in this work. A classical agent (left) interacts with an unknown quantum process (center) via classical control bitstrings. The agent’s goal is to maximize cumulative rewards—represented here as extracted energy—while minimizing ``waste''. This waste can be viewed through two lenses: as regret in the context of learning theory, or as dissipation in the context of quantum thermodynamics.}
    \label{fig:cartoon}
\end{figure}

\subsection{Results}
\label{sec:results}

In this section, we present the formal framework of our learning problem and outline our main results. We first introduce the input-output Quantum Hidden Markov Model (QHMM) and establish the fundamental observability assumptions required for statistical identifiability. We then present our learning algorithm and establish sublinear regret bounds for both discrete and continuous action spaces, highlighting the distinct mathematical innovations necessitated by the quantum setting. Furthermore, we prove the information-theoretic optimality of these bounds. Finally, we demonstrate the physical significance of our framework through a direct application to quantum thermodynamics.

\subsubsection{The Input-Output QHMM Framework}

We consider a sequential decision-making scenario in which a classical agent interacts with a finite-dimensional quantum memory. We formalize this interaction through the input-output Quantum Hidden Markov Model (QHMM), strictly defined in Section~\ref{sec:model}. The environment maintains a hidden quantum state $\rho_1 \in \D(\mathcal{H}_S)$ of dimension $S$. Between interaction rounds, this latent memory evolves via unknown, completely positive trace-preserving (CPTP) channels $\mathbb{E}_l: \mathbb{M}_{S} \to \mathbb{M}_{S}$ for $l \in [L]$.

At each step $l \in [L]$, the agent selects an action $a_l \in \mathcal{A}$ from an available action set. This action dictates the specific quantum instrument $\mathcal{P}^{(a_l)}$ applied to the latent memory. The instrument comprises a collection of completely positive trace-non-increasing (CPTNI) maps $\{\Phi_o^{(a_l)}\}_{o \in [O]}$ that sum to a CPTP map. Applying this instrument yields a classical outcome $o_l \in [O]$ and correspondingly updates the post-measurement memory state. Based on the chosen action and the observed outcome, the agent receives a scalar reward $r_l(a_l, o_l)$.

This sequential interaction produces a trajectory $\tau_l = (a_1, o_1, \dots, a_l, o_l)$. The agent leverages this history to determine its subsequent action $a_{l+1} \in \mathcal{A}$ using a policy $\pi(\cdot | \tau_l)$, which defines a conditional probability distribution over the action set. The agent's performance over a single interaction episode of length $L$ is quantified by the value function, defined as the expected cumulative reward $V^\pi \coloneqq \mathbb{E}_\pi\big[\sum_{l=1}^L r_l(a_l, o_l)\big]$. 

The ultimate goal of the agent is to maximize this expected reward over $K$ independent episodes, which is mathematically equivalent to minimizing the cumulative regret:
\begin{equation}
    \regret(K) := \sum_{k=1}^K (V^* - V^{\pi^k}),
\end{equation}
where $V^*$ denotes the value of the optimal policy for the true environment. We provide a schematic representation of this classical-quantum interface in Figure~\ref{fig:scheme_interaction}.

\begin{figure}[htbp]
    \centering
    \includegraphics[width=\linewidth]{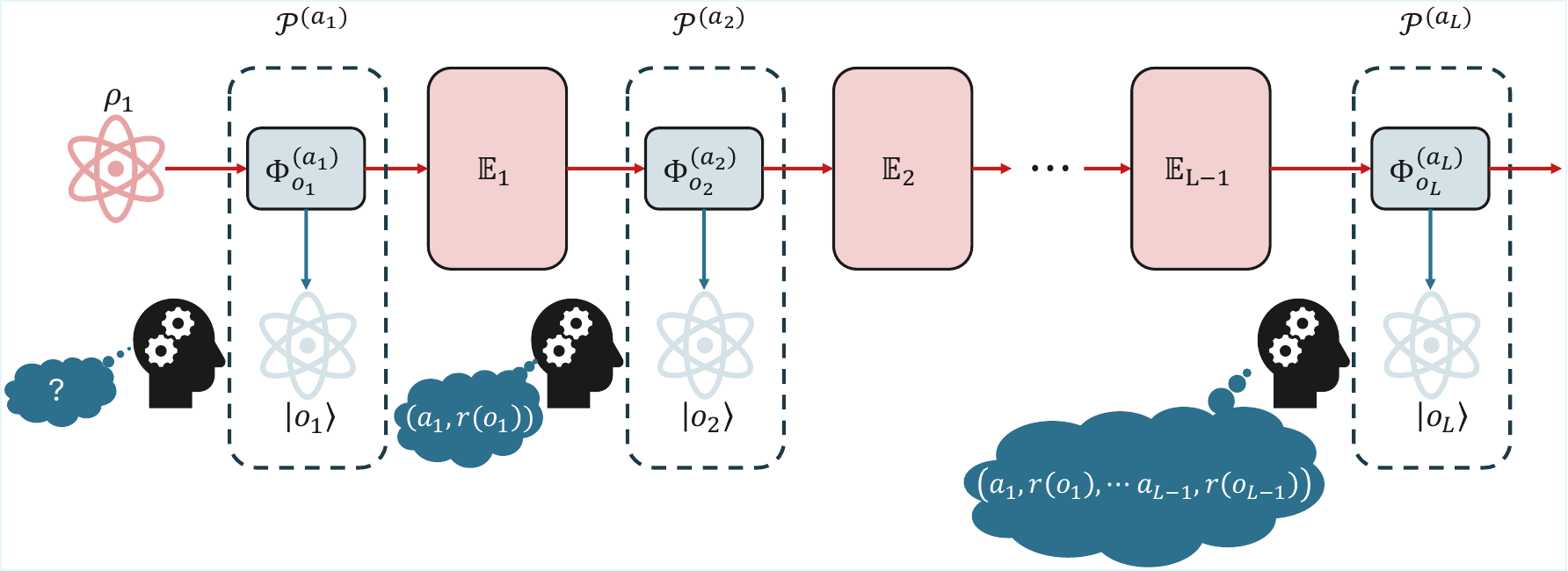}
    \caption{Schematic representation of the $L$-step input-output QHMM environment. The environment's initial hidden memory state, $\rho_1$, evolves sequentially through unknown completely positive trace-preserving (CPTP) channels $\mathbb{E}_l$. At each round $l$, a classical agent utilizes the past trajectory to select an action $a_l$, which dictates the quantum instrument $\mathcal{P}^{(a_l)}$ applied to the latent memory. This physical interaction, governed by the completely positive map $\Phi_{o_l}^{(a_l)}$, yields a classical outcome $o_l$. As depicted in the thought bubbles, the agent internally registers a scalar reward $r(o_l)$ based on this outcome and updates its accumulated trajectory history to optimize future interventions. Simultaneously, the post-measurement memory state is updated for the subsequent channel evolution.}
    \label{fig:scheme_interaction}
\end{figure}

\subsubsection{Identifiability and the Observable Operator Model}
\label{subsec:results_identifiability}

A fundamental challenge in learning hidden Markov models---both classical and quantum---is identifiability. In general, recovering the exact physical latent parameters (such as the specific internal CPTP channels) strictly from classical data is fundamentally ill-posed. The mapping from the internal quantum state space to the observable sequence is non-injective; infinitely many distinct physical environments can produce the exact same observable statistics. Consequently, any learning algorithm attempting to estimate the true physical parameters cannot distinguish between these identical-looking environments. Without a unique target to converge on, direct parameter estimation becomes statistically inefficient and algorithmically intractable.

To bypass this, one can instead construct an observable realization (historically termed a quasi-realization) that perfectly captures the outcome distributions. While such a realization is not strictly unique, it is identifiable up to a similarity transformation: any two minimal (regular) realizations predicting the same trajectory statistics are related by an invertible change of basis~\cite{vidyasagar2011complete}. This ensures that the equivalence class of realizations is unique, making the learning problem mathematically well-posed. This foundational principle has been successfully leveraged in the quantum domain to bound the sample complexity of learning matrix product operators and finitely correlated states~\cite{fanizza2023learning,baumgratz2013scalable, monras2016quantum}.

To ensure that such an observable realization exists and can be efficiently learned in our interactive setting, we require that the observable classical data contains sufficient information to distinguish between different latent memory states. In classical POMDPs, this is typically resolved by requiring the stochastic emission matrices to possess full column rank~\cite{jin2020sample,liu2022partially}, an assumption widely referred to as the \emph{undercomplete} condition.

In our quantum framework, the direct mathematical analogue is the existence of a linear recovery map that perfectly reconstructs the post-instrument quantum state strictly from the classical marginals. We formalize this through the \emph{undercomplete assumption}: for every action $a \in \mathcal{A}$, there exists a linear map $R^{(a)} : \mathbb{M}_{O} \rightarrow \mathbb{M}_{S} \otimes \mathbb{M}_{O}$ such that
\begin{equation}
    R^{(a)} \circ \Tr_S \circ \mathcal{P}^{(a)} = \mathcal{P}^{(a)}.
\end{equation}
To ensure sample-efficient learning, this inversion must be robust. We define the robustness constant $\kappa_{\mathrm{uc}}$ as the uniform upper bound on the induced $1$-norm of these recovery maps (see Section~\ref{sec:oom_uc} for the formal definition). This generalized condition perfectly recovers the classical undercomplete assumption studied in recent literature~\cite{hsu2012spectral,jin2020sample,liu2022partially} when the memory is restricted to diagonal states and the instruments correspond to purely classical operations.

Crucially, the existence of this recovery map allows us to evaluate the statistics of our process using only linear operations within the observable outcome space. Specifically, we construct an Observable Operator Model\footnote{Historically pioneered as the \emph{quasi-realization} problem in control theory; see Section~\ref{sec:oom_uc} for a detailed discussion.} (OOM)~\cite{jaeger2000observable}, which elegantly bypasses the latent quantum space entirely. For any step $l$, outcomes $o$, and actions $a, a'$, the OOM superoperator acting strictly on the classical register is defined as:
\begin{equation}\label{eq:oom_intro}
    \mathbf{A}_l(o, a, a') := \Tr_S \circ \mathcal{P}^{(a')} \circ (\mathbb{E}_l \otimes \mathcal{T}_o) \circ R^{(a)},
\end{equation}
where $\mathcal{T}_o(X) = \Tr_O(\proj{o}X)$. Using these operators, we can compute the sequence of conditional outcome probabilities of the form $\Pr(o_{l'}| \tau_{l'-1}, a_{l'})$ simply by sequentially composing the corresponding operators and tracing the observable space. We detail the construction of these trajectory probabilities in Section~\ref{sec:trajectory_probab}. This parameter-free formulation forms the rigorous backbone of the estimation phase of our learning algorithm.

\subsubsection{The OMLE Algorithm for quantum processes}
\label{subsec:results_algorithm}

To minimize regret, we adapt the Optimistic Maximum Likelihood Estimation (OMLE) algorithm~\cite{liu2022partially} to our quantum setting, fully detailed in Section~\ref{sec:omle_discrete} and Algorithm~\ref{alg:omle-qhmm}. Operating under the principle of optimism in the face of uncertainty, the algorithm proceeds in three main steps at each episode $k \in [K]$:
\begin{enumerate}
    \item \textbf{Estimation (Confidence Set Construction):} Using the dataset of past trajectories, the algorithm computes the Maximum Likelihood Estimator (MLE) for the QHMM parameters. These probabilities are evaluated via the OOM superoperators~\eqref{eq:oom_intro} for plausible quantum processes satisfying the undercomplete assumption. The algorithm then constructs a confidence region $\mathcal{C}_k$ of environments $\omega'$ whose log-likelihood falls within a specified radius $\beta_k$ of the MLE, guaranteeing that the true quantum process lies within this region with high probability.
    \item \textbf{Exploration-Exploitation (Optimistic Planning):} The algorithm selects a candidate model $\omega_k \in \mathcal{C}_k$ and a policy $\pi^k$ that jointly maximize the expected cumulative reward. This step naturally balances the exploration of uncertain dynamics with the exploitation of known rewarding transitions.
    \item \textbf{Execution:} The agent executes $\pi^k$ in the true environment for $L$ steps, records the new trajectory $\tau^k$, and updates the dataset. Intuitively, if the confidence region is sufficiently narrow due to extensive data collection, the selected optimistic policy will perform near-optimally when executed in the true environment.
\end{enumerate}

\subsubsection{End-to-End Learning Guarantees and Technical Contributions}
\label{subsec:results_bounds}

We establish end-to-end regret bounds for the OMLE algorithm. We separate our results into two regimes: discrete action spaces (general instruments) and continuous action spaces. For continuous action spaces we will restrict to quantum instruments that model POVMs. 

\paragraph{Discrete Actions and Error Propagation.} For a finite set of actions corresponding to general quantum instruments, we provide the following guarantee:

\begin{theorem}[Informal: Regret for Discrete Actions]
\label{thm:informal_discrete}
Let the action space $\mathcal{A}$ be discrete with cardinality $A$. Assuming that the true QHMM environment satisfies the undercomplete assumption with robustness $\kappa_{\mathrm{uc}}$, the OMLE algorithm achieves a cumulative regret of $\widetilde{\mathcal{O}}(\poly(A,O,S,L) \times \kappa_{\mathrm{uc}}^2 \sqrt{K})$. 
\end{theorem}

The formal statement is Theorem~\ref{thm:discrete_regret}. The proof of this result necessitates a major technical innovation: controlling the forward propagation of the estimated OOM error. By telescoping the sequential OOM products, bounding the cumulative regret fundamentally reduces to controlling terms of the form:
\begin{equation}
    \big\| \mathbf{A}_{L-1:l+1} \big(\mathbf{A}^k_{l}-\mathbf{A}_{l}\big) \mathbf{a}_l\big\|_1.
\end{equation}
Here, $\big(\mathbf{A}^k_{l}-\mathbf{A}_{l}\big)$ represents the local OOM estimation error at step $l$, and $\mathbf{A}_{L-1:l+1}$ represents the sequence of future operators acting upon this error. Because each OOM operator internally applies a recovery map $R$, naively bounding the trace norm of this future sequence would compound the robustness penalty at every step, leading to an exponential blow-up of $\mathcal{O}(\kappa_{\mathrm{uc}}^{L-l-1})$.

 In classical POMDPs~\cite{liu2022partially}[Lemma 27 and 31], this forward propagation is bounded by unpacking the sequential contractions of the OOM linear maps in a way that reveals a sequence of substochastic maps followed by a pseudo-inverse of the stochastic emission matrix. Therefore, the exponential blow-up of the error is avoided thanks to the contractivity of sub-stochastic maps. This decomposition does not apply in our QHMM framework, the operators are constructed from non-commutative superoperators (CPTP channels and quantum instruments), but we can arrive at an analogous result.

We overcome this by explicitly unpacking the composition of the future OOM sequence $\mathbf{A}_{L-1:l+1}$. By applying our undercompleteness assumption ($R^{(a)} \circ \Tr_S \circ \mathcal{P}^{(a)} = \mathcal{P}^{(a)}$), we demonstrate that the recovery maps perfectly annihilate the subsequent partial traces and instruments across the entire sequence. This structural cancellation isolates the very first recovery map as the \emph{only} potentially non-contractive component. As formally established in Lemma~\ref{lem:one_norm_op}, the remainder of the sequence elegantly collapses into a purely physical succession of true quantum channels and instruments. Because these quantum operations are trace-non-increasing, they act as a uniform contraction on the state space. Thus, it is the robustness constant $\kappa_{\mathrm{uc}}$ that stops to propagate; the operator norm of the entire future sequence is strictly bounded by a single factor of $\kappa_{\mathrm{uc}}$, completely bypassing any exponential dependence on the length of the sequence $L$.

\paragraph{Continuous Actions (POVMs) and OOM Error Decomposition.}
Restricting our agents to finite discrete actions is physically artificial, as quantum instruments naturally possess a much richer, continuous structure. In particular, it is natural to consider continuous action sets generated by Positive Operator-Valued Measures (POVMs). However, because the discrete regret bound scales polynomially with the number of actions $A$, naively applying the discrete approach to the continuous regime will inevitably diverge.

To ensure statistical tractability for infinite action sets, we define the \emph{spanning dimension} $d_{\mathcal{A}}$ of the action class (see Section~\ref{sec:embedding-dim}), which counts the maximal number of linearly independent centered POVM effects. While this establishes a finite embedding space, the core technical challenge lies in bounding the trajectory estimation error without paying a statistical penalty for the continuous action selection at each step.

We solve this by rigorously decoupling the environment's estimation error from the agent's continuous action choices. The critical insight is that the quantum instruments (the agent's actions) are \emph{known}; the only source of OOM estimation error $\Delta \mathbf{A}_l^k=\mathbf{A}_l^k-\mathbf{A}_l$ stems from the unknown internal memory transition channels $\Delta\mathbb{E}_l^k = \mathbb{E}_l^k - \mathbb{E}_l$.

To exploit this, in Lemmas~\ref{lem:linear-first-action} and~\ref{lem:quantum_forward_factorization} we introduce an exact multi-linear decomposition of the OOM prediction error. By expanding the operators in a fixed, finite-dimensional Hilbert-Schmidt basis, we factorize the prediction error for any outcome into two distinct components: an \emph{action-independent} environment error tensor (capturing the channel mismatch $\Delta\mathbb{E}_l^k$) and \emph{action-dependent} data vectors (capturing the known POVM effects and filtered states). 

This multi-linear factorization elegantly projects the infinite-dimensional action space into a finite $S^4$-dimensional tensor space. Consequently, the statistical complexity of learning the continuous quantum process is bottlenecked entirely by the finite dimension of the latent quantum memory $S$, rather than the cardinality of the action set. Building upon this structural decomposition, we adapt the OMLE algorithm to the continuous setting and prove the following bound:

\begin{theorem}[Informal: Regret for Continuous Actions]
\label{thm:informal_continuous}
Let the action space $\mathcal{A}$ be a continuous set of POVMs with spanning dimension $d_{\mathcal{A}} \le (O-1)S^2$. The OMLE algorithm achieves a cumulative regret of $\widetilde{\mathcal{O}}(\poly(d_\mathcal{A},S,L)\times\kappa_{\mathrm{uc}}^2 \sqrt{K})$. This scaling is completely independent of the cardinality of the action set.
\end{theorem}

The formal statement is Theorem~\ref{thm:continuous_regret}. This result establishes that learning unknown quantum dynamics over continuous measurement spaces is highly efficient, provided the spanning dimension is bounded.

We note that while recent classical literature has explored regret bounds for POMDPs in continuous settings~\cite{liu2023optimistic}, such works exclusively focus on continuous \emph{observation} spaces while restricting the agent to a finite, discrete set of actions. Consequently, our mathematical machinery for handling continuous action spaces---governed physically by quantum instruments and POVMs---provides novel analytical tools that may prove broadly applicable within the classical reinforcement learning community. Furthermore, in the simpler regime of bandits, our continuous POVM formulation naturally parallels the classical linear bandit framework~\cite{Dani2008stochastic, rusmevichientong2010linearly, abbasiyadkori20111improved}. Just as linear bandits transcend discrete action sets by embedding actions into a finite-dimensional real vector space and exploiting the linearity of the expected reward, our approach embeds continuous quantum measurements into a finite-dimensional operator spanning space, fully exploiting the linear dependence of the outcome probabilities on the latent quantum state.

\paragraph{Information-Theoretic Lower Bounds.}
To establish that our $\widetilde{\mathcal{O}}(\sqrt{K})$ scaling is statistically optimal, we derive information-theoretic lower bounds. A trivial $\Omega(\sqrt{K})$ lower bound could be obtained by simply restricting the latent memory to diagonal states and embedding a classical Multi-Armed Bandit. However, such a reduction merely proves that learning classical probability distributions is hard, bypassing the structural requirements of our fully quantum model class.

Specifically, for a fully quantum latent memory ($S \ge 2$), satisfying the undercomplete assumption requires the measurement instruments to be informationally complete to guarantee the existence of a valid recovery map. To demonstrate that the sublinear regret scaling remains strictly optimal even under these stringent quantum observability requirements, we introduce a highly non-trivial reduction to the Multi-Armed Quantum Bandit (MAQB) problem~\cite{lumbreras2022multi}. In Proposition~\ref{prop:sic_povm_bandit}, we construct a hard environment utilizing symmetric informationally complete measurements (SIC-POVMs) and completely depolarizing channels. The SIC-POVMs guarantee that a valid recovery map exists with optimal robustness ($\kappa_{\mathrm{uc}}=1$). By doing so, we explicitly demonstrate that the $\Omega(\sqrt{K})$ hardness is intrinsic to the geometry of quantum state discrimination, rather than an artifact of classical probability limits.

Furthermore, while the MAQB reduction establishes the asymptotic temporal scaling, it relies on perfectly informative measurements. To show that the polynomial dependence on the robustness constant $\kappa_{\mathrm{uc}}$ is also strictly necessary, we demonstrate that our input-output QHMM framework fully subsumes classical POMDPs, allowing us to embed ``combinatorial lock'' hard instances~\cite{liu2022partially} directly into the quantum setting. Together, these two reductions yield the following unified lower bound.

\begin{theorem}[Informal: Information-Theoretic Lower Bounds]
\label{thm:informal_lower_bounds}
For any horizon $L \ge 1$, there exists a valid QHMM environment satisfying the undercomplete assumption with optimal robustness ($\kappa_{\mathrm{uc}}=1$) for which any learning algorithm incurs an expected cumulative regret of $\Omega(\sqrt{K L})$. Furthermore, the polynomial dependence on the robustness constant is unavoidable: for any $L,A \ge 2$, there exists an undercomplete QHMM environment such that any algorithm suffers a linear expected regret of $\Omega(K)$ during the initial phase of $K \le \mathcal{O}(\kappa_{\mathrm{uc}} / L)$ episodes.
\end{theorem}
The above result is properly formalized in Theorem~\ref{thm:qhmm_lower_bound} and Corollary~\ref{cor:qhmm_lock}.

\subsubsection{Application: State-Agnostic Work Extraction}
\label{sec:intro_work_extraction}

Beyond abstract learning, our framework has direct implications for quantum thermodynamics. We consider the problem of state-agnostic work extraction, where an agent seeks to harvest non-equilibrium free energy from a sequence of non-i.i.d.\ quantum states and store it in a battery. If the agent does not know the underlying emission source, it cannot perfectly align its extraction protocol with the incoming states. By the laws of thermodynamics, this informational deficit inevitably results in the irreversible dissipation of energy~\cite{kolchinsky2017dependence,riechers2021initial}.

\begin{figure}[h]
    \centering
    \includegraphics[width=0.95\linewidth]{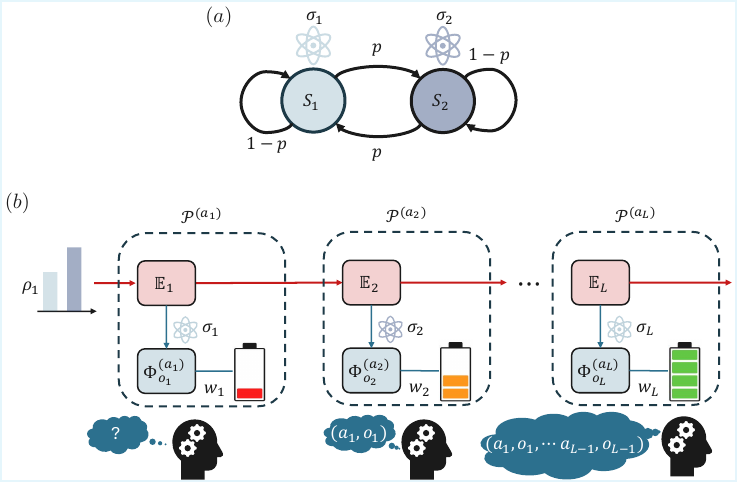}
    \caption{Schematic of the work extraction. Panel (a) illustrates an example of classical Hidden Markov Model which dictates the dynamics of the process. $\{S_i\}_{i=1}^2$ are the latent classical states, $\{\sigma_i\}_{i=1}^2$ are emitted based on which transition happens at each time step, dictated by the probability $p$. These can be encoded into transition matrix $\mathbb{E}$. Panel (b) illustrates how an agent interact sequentially with the emitted quantum state via CPTNI maps $\{\Phi_{o_i}^{(a_i)}\}_{i=1}^L$. The memory is assumed to be a classical distribution over latent state shown in panel (a) and it evolves according to the dynamic transcribed in $\{\mathbb{E}_l\}_{l=1}^L$. After each interaction, the agent receives outcome $o_l$ in the form of the work value $w_l$ that is stored in the battery. This sequential emission model can be described by input-output QHMM framework, this is shown in Sec.~\ref{sec:sequential_HMM}.}
    \label{fig:extraction_sketch}
\end{figure}

We can formalize this thermodynamic process as sequential-emission HMM, an instance of our input-output QHMM. As illustrated in Figure~\ref{fig:extraction_sketch}, we model a scenario where a classical hidden memory governs the sequential emission of quantum states. We emphasize that this restriction to a classical memory is not a limitation of our general framework, but a consequence of the physical work extraction protocol itself. In the physical models we use for work extraction, measuring a qubit yields only two discrete energy outcomes. Consequently, our undercomplete assumption dictates that we can only robustly reconstruct a classical (diagonal) memory from the observable marginals, rather than a full quantum memory. At each time step, the agent probes the emitted system to extract work. Building upon established work extraction models~\cite{skrzypczyk2014work, huang2023engines}, the agent's action $a_l \in \mathcal{A}$ configures the physical extraction protocol by choosing a measurement basis and a target purity parameter (see Section~\ref{sec:work_extraction} for full physical details). Mathematically, we model this entire physical interaction as a quantum instrument. 

The protocol gives a discrete classical outcome $o_l$. Crucially, the probability distribution of this outcome follows Born's rule strictly according to the chosen measurement basis. The actual energy exchanged with the battery then depends on both this observed discrete outcome and the continuous target purity chosen by the agent. Because this physical interaction acts simultaneously as the data collection mechanism and the process for energy extraction, we formally define our RL reward function directly as the extracted work, $r(a_l, o_l) = w_l$, aiming to maximize the harvested free energy.

By setting the reward to be the extracted work, the optimal value function $V^*(\omega)$ achieved by the optimal policy represents the maximum extractable free energy a sequential agent is able to extract from the temporal correlations of the source. Thus, the suboptimality of any learning policy at episode $k$---the gap between the theoretical maximum extractable work and the actually achieved work---corresponds to the unrecoverable thermodynamic heat dissipated during that episode. This establishes a direct equivalence between the cumulative mathematical regret and the \emph{cumulative dissipation} as
\begin{equation}
    \regret(K) = W_{\mathrm{diss}}(K)~,
\end{equation}
where $W_{\mathrm{diss}}(K)$ is the \emph{total dissipation across all episodes}. This dissipation is a fundamental cost of learning; it cannot be avoided even if the physical work-extraction operations are executed perfectly~\cite{proesmans2020finite,riechers2021initial,van2022finite}. Ultimately, it quantifies the minimum unavoidable waste of free energy resulting from the agent's incomplete knowledge of the underlying process.

Designing a regret-minimizing RL algorithm is then physically equivalent to engineering an adaptive Maxwell's demon that learns to minimize energy waste on the fly. By using the OMLE strategy, the agent uses past battery energy outcomes to learn the latent memory dynamics, continuously updating its extraction basis to achieve optimal performance.

We demonstrate this through numerical simulations. As shown in Figure~\ref{fig:extraction_sketch} (see Section~\ref{sec:numerics} for implementation details), we apply the algorithm to a work extraction protocol over different horizons $L$. This demonstrates that our theoretical algorithm can be practically implemented for this specific classical-memory model. The plot displays a strict sublinear accumulation of work dissipation, confirming that the agent effectively uses past battery energy outcomes to learn the latent dynamics on the fly.

From our theoretical bounds we have that our $\widetilde{\mathcal{O}}(\sqrt{K})$ regret scaling establishes rigorous worst-case guarantees. However the numerical behavior reflects the expected instance-dependent sublinear scaling which is much better than this worst-case bound. It is worth noting that the theoretical $\tilde{\Theta}(\sqrt{K})$ bound (Theorem~\ref{thm:informal_discrete}) applies to the worst-case environment for a fixed, finite number of episodes $K$. In contrast, the $\polylog(K)$ behavior observed numerically in Figure~\ref{fig:intro_cum_dissipation_compare} is typical of instance-dependent bounds for a specific environment as the number of episodes asymptotically increases. For a more detailed discussion on the distinction between these worst-case and instance-dependent regimes, see~\cite[Chapters 15 and 16]{lattimore2020bandit}. Crucially, in both scenarios, the regret scales sublinearly, which ensures that the \emph{rate} of average cumulative dissipation or average cumulative dissipation per episode vanishes asymptotically. More formally, we have that
\begin{equation}
    \lim_{K\to\infty} \frac{1}{K} W_{\mathrm{diss}}(K) = 
    \lim_{K\to\infty}\frac{\widetilde{\mathcal{O}}(\sqrt{K})}{K} = \lim_{K\to\infty} \widetilde{\mathcal{O}}\left(\frac{1}{\sqrt{K}}\right)=0~.
\end{equation}
This shows that the thermodynamic cost of learning such a process, over many periods, can become negligible. The same argument applies to the instance-dependent $\polylog(K)$ behavior of cumulative dissipation of Figure~\ref{fig:intro_cum_dissipation_compare}.

\begin{figure}
    \centering
    \includegraphics[width=0.90\linewidth]{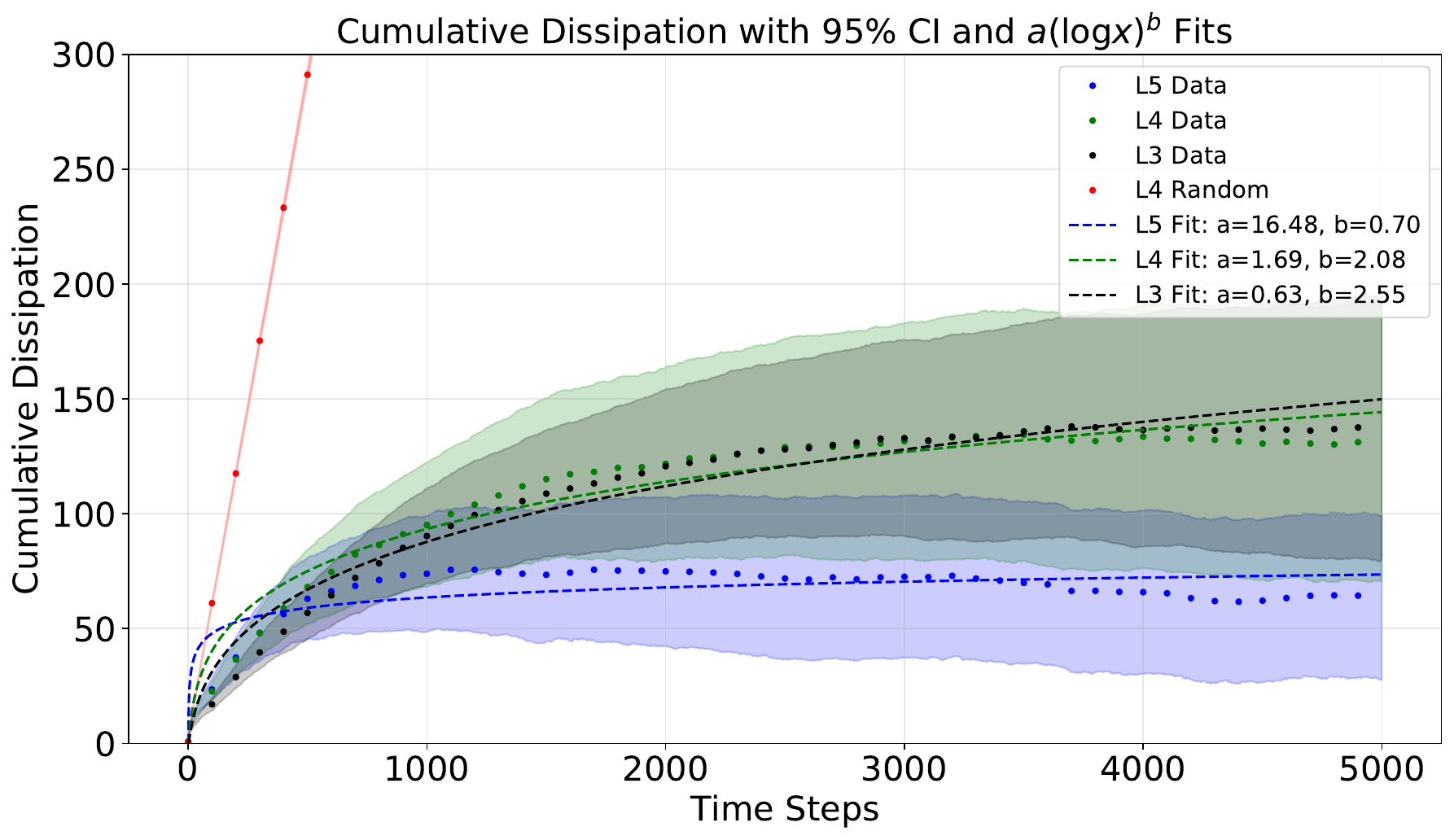}
    \caption{Plot showing cumulative work dissipation against the number of episodes. The black line represents $L=3$, green line being $L=4$ and blue being $L=5$. Their 95\% confidence interval is filled with a translucent area. The red line represent a comparison if the agent were to use random policy, i.e. at each time step the agent choose randomly from given actions. }
    \label{fig:intro_cum_dissipation_compare}
\end{figure}

\subsection{Discussion}
\label{sec:discussion}

In this work, we have bridged partially observable reinforcement learning and quantum information theory by introducing a comprehensive framework for sequential decision-making over unknown quantum processes with memory.

\subsubsection{Summary of Theoretical Contributions}
Rather than relying on classical emission probabilities, our input-output QHMM environment models an agent interacting with a latent quantum memory evolving via unknown CPTP channels (Definition~\ref{def:qrl_model}). By adapting the OMLE algorithm (Algorithm~\ref{alg:omle-qhmm}), we established that an agent can learn a QHMM environment with a cumulative regret $\widetilde{\mathcal{O}}(\sqrt{K})$ over $K$ episodes (Theorem~\ref{thm:discrete_regret}).

Crucially, we bypassed the standard pseudo-inverse bounds used in classical POMDPs by isolating the recovery map and leveraging the uniform contraction of subsequent quantum channels (Lemma~\ref{lem:one_norm_op}). This ensures our error propagation remains strictly bounded by the robustness constant $\kappa_{\mathrm{uc}}$, avoiding exponential blow-up with the horizon $L$ (Lemma~\ref{lem:regret_step1}). By utilizing the $\ell_1$-eluder dimension, we establish a rigorous connection between historical estimation error and future performance (Lemma~\ref{lem:oom_first_action_dicrete} and~\ref{lem:oom_error_discrete}), which then ensures that the cumulative regret for the OMLE algorithm is strictly bounded. Furthermore, by utilizing a finite-dimensional spanning embedding (Section~\ref{sec:embedding-dim}), we extended these $\widetilde{\mathcal{O}}(\sqrt{K})$ guarantees to continuous action spaces without dependence on the action set cardinality (Theorem~\ref{thm:continuous_regret}). Via a reduction to the MAQB using SIC-POVMs, we proved this scaling is statistically optimal (Section~\ref{sec:lower_bounds}). 

\subsubsection{Physical Insights and Thermodynamic Application}
Beyond the theoretical regret bounds, our results offer deeper insights into the physical nature of learning quantum dynamics. The undercomplete assumption, which guarantees a linear recovery map, serves as a quantum analogue to the classical full-rank emission matrix. Physically, it implies that the chosen quantum instruments extract sufficient information to infer the post-measurement quantum state strictly from the classical marginal. Our framework demonstrates that as long as this inversion is robust ($\kappa_{\mathrm{uc}} < \infty$), the latent quantum dimension does not inherently bottleneck the macroscopic learning process. 

Furthermore, we demonstrated the physical significance of our framework by applying it to state-agnostic quantum work extraction (Section~\ref{sec:work_extraction}). By defining the RL reward strictly as the extracted physical work, we established a fundamental link between information-theoretic regret and thermodynamic dissipation. Our numerics in Figure~\ref{fig:intro_cum_dissipation_compare} confirm that an adaptive Maxwell's demon can learn the underlying hidden Markovian behavior, achieving sublinear cumulative dissipation on the fly.

\subsubsection{Open Problems and Future Directions}
Our framework opens several promising directions for future exploration at the intersection of quantum information theory, reinforcement learning and quantum resources extraction. 

\paragraph{Going beyond the undercomplete assumption.} The most direct extension from our work is to move beyond the undercomplete case where the agent is able to extract sufficient information to perfectly recover the post-measurement state strictly from the classical data of a single action. In the more complex \emph{overcomplete} regime, the observable space is fundamentally smaller than the dimension of the latent quantum memory. Consequently, the statistics of observables from any single quantum instrument are insufficient to recover the post-measurement state. Instead, the agent is only able to extract sufficient information to reconstruct the delayed post-measurement state by aggregating data across multiple sequential actions ~\cite{jin2020sample, liu2022partially}. It is an important next step to adapt our OOM learning framework for this deficiency.

\paragraph{Agnostic Extraction of General Quantum Resources.} 
While Section~\ref{sec:work_extraction} successfully applies our framework to the extraction of thermodynamic work, free energy is merely one specific operational resource within the broader framework of quantum resource theories \cite{chitambar2019quantum}. A highly compelling open direction is the application of sequential learning to the agnostic extraction and distillation of other fundamental quantum resources, such as quantum coherence \cite{streltsov2017colloquium} or entanglement, from uncharacterized, memoryful sources. In realistic physical scenarios, a sequence of emitted states carrying these resources will naturally exhibit temporal correlations. Any informational deficit regarding these latent dynamics will inevitably lead to a suboptimal extraction. By reformulating the RL reward function to quantify the yield of distilled coherence or entanglement, an agent could autonomously adapt its operations to maximize resource harvesting on the fly. Extending our OMLE algorithm to this generalized setting requires rigorously mapping specific resource distillation protocols into valid quantum instruments, and determining whether our $\tilde{O}(\sqrt{K})$ regret bounds successfully capture the fundamental, operational limits of general quantum resource theories.

\paragraph{Beyond Markovian Latent Dynamics: Full Quantum Combs and Process Tensors.} 
While our input-output QHMM faithfully captures temporal correlations through a finite-dimensional quantum memory, its internal evolution remains fundamentally Markovian in the latent space; the subsequent memory state depends strictly on the current state and the agent's immediate instrument. A natural and highly motivated extension is to generalize this learning framework to fully non-Markovian quantum stochastic processes, where the environment may retain arbitrarily complex memory and entanglement across multiple time steps. Such multi-time quantum processes are rigorously formalized by the framework of quantum combs \cite{chiribella2008quantum, chiribella2009theoretical} or, equivalently, process tensors \cite{pollock2018non}. Transitioning from a QHMM to a general quantum comb presents a significant theoretical challenge: the sufficient statistic for future observations is no longer a simple finite-dimensional filtered state. Consequently, our undercomplete assumption and OOM construction would need to be significantly generalized to handle non-Markovian temporal correlation structures. Establishing regret-minimizing algorithms for learning optimal control over process tensors would provide a universal framework for reinforcement learning in arbitrarily structured open quantum systems.

\paragraph{Learning Beyond Fixed Causal Structures.} Our current QHMM environment assumes a strict, sequential application of CPTP channels and instruments. A fascinating open problem is whether these regret bounds can be generalized to learn quantum processes with indefinite causal order (ICO)~\cite{oreshkov2012quantum, chiribella2013quantum}. In frameworks exhibiting ICO, quantum operations are not applied in a fixed sequence but can be routed in a superposition of causal orders, such as via a quantum switch. While recent advances have shown advantages of indefinite causality in quantum metrology~\cite{zhao2020quantum,kong2026noncommutativity}, integrating these indefinite causal structures into a sequential, reward-maximizing reinforcement-learning framework is completely unexplored. It is unclear how the OOM must be adapted and whether the uniform contraction properties preventing error propagation would still hold. 

\paragraph{Adversarial Settings and Online Learning.} We assumed that the quantum memory evolves under unknown CPTP channels fixed by a stationary environment. If the dynamics are instead chosen by an adversary, the regret minimization problem shifts toward adversarial online learning~\cite{auer2002nonstochastic,lattimore2020bandit}. In the field of quantum learning theory, adversarial online learning has also been prominently explored in the context of predicting measurement outcomes on unknown quantum states that are adaptively presented by an adversary~\cite{aaronson2019online, chen2024adaptiveonline, meyer2025online}. Extending our sequential QHMM environment to this regime would generalize these online learning protocols from states to quantum processes with memory, yielding robust learning algorithms with minimal regret even in highly uncharacterized environments.

\paragraph{Continuous Time and Markovian Master Equations.} Finally, our results are strictly defined for discrete-time steps $l \in [L]$. Extending this framework to continuous-time open quantum systems, where the latent state evolves via a Lindbladian master equation and the agent applies continuous measurement and feedback, remains a significant open challenge~\cite{barchielli1991measurements,wiseman2009quantum}. In continuous quantum control theory, the discrete CPTP maps and instruments are replaced by stochastic master equations driven by weak continuous measurements~\cite{sivak2022modelfree, song2025fast,ma2025machine}. Determining the continuous analogue of the undercomplete assumption and the QHMM learning problem would bridge our discrete reinforcement learning bounds with the rich, highly active field of adaptive quantum control.

\subsection{Related Works}
\label{sec:related_work}

\paragraph{Quantum POMDPs.} The most immediate related work to ours is the ``quantum POMDPs'' (QOMDPs) introduced more than a decade ago in~~\cite{barry2014quantum}. Their model is almost the same as ours with the difference that they assume that the underlying channels i.e. the dynamics are known and their focus was not regret minimization, it was on finding policies that can reach a certain value and their complexity in comparison with the classical model. Our input--output QHMM instead treats the latent dynamics as unknown and learned from interaction, in an episodic fixed-horizon setting. In this sense, our model is closer to a true POMDP-style generalization—partial observability with unknown transition dynamics—tailored to regret-minimizing reinforcement learning rather than known-model control.

\paragraph{Classical POMDPs.} We acknowledge that our framework is deeply inspired by recent, elegant breakthroughs in learning for classical POMDPs~\cite{jin2020sample,liu2022partially,liu2023optimistic}. We highly recommend these insightful papers to any reader seeking a comprehensive foundation in the mechanics of regret minimization under partial observability. While we adapt the high-level optimistic principles introduced in these works, the transition to the quantum regime required the development of substantially new mathematical machinery. As detailed in our earlier results our new ideas are, controlling error propagation through non-commutative quantum channels, embedding continuous measurement spaces, and mapping regret to thermodynamic dissipation.

\paragraph{Multi-armed quantum bandits.} 
The multi-armed quantum bandit framework introduced in~\cite{lumbreras2022multi,lumbreras2025banditsroaminghilbertspace} is essentially the little brother of our current model. It corresponds exactly to the special case where the horizon is $L=1$. This setting was later extended to the tomography problem of learning pure states without disturbance~\cite{pmlr-v247-lumbreras24a,lumbreras24pure}, as well as state-agnostic work extraction from i.i.d.\ unknown states~\cite{lumbreras2025quantum}. The framework has also been adapted into contextual bandits for quantum recommender systems~\cite{brahmachari2024quantum}.

\paragraph{Quantum hidden Markov models.}
Quantum hidden Markov models can be viewed as our model without an action set. They operate strictly without input. These models were introduced as a natural generalization of classical HMMs~\cite{Monras2010}. Such models required less historical data than optimal classical counterparts to generate equally accurate predictions~\cite{gu2012quantum}. Since then, the representational power of QHMMs, or both sequence generation and prediction, has been extensively studied~\cite{monras2016quantum,binder2018practical,elliott2020extreme,sundar2026quantum}, demonstrating that quantum models can replicate certain stochastic behavior with various memory constraints. Moreover, this advantage can scale without bound. More recently, their fundamental relationship to generalized probabilistic theories (GPTs) was analyzed in~\cite{fanizza2024quantum}.

\paragraph{Quantum Agents.}
Our work here considered classical agents operating in a quantum environment. This is mirrored by a body of work on advanced quantum agents operating in classical environments. Such agents generalize classical information transducers and exhibit target input-output behavior with less memory than their classical counterparts~\cite{thompson2017using}. Moreover, this advantage scales without bound~\cite{elliott2022quantum} and can directly translate into energetic advantages for online decision-making~\cite{thompson2025energetic}.

\section{Quantum Stochastic Process Framework}
\label{sec:model}

In this section, we introduce an input-output reinforcement-learning framework for quantum stochastic processes, modeled as a quantum hidden Markov model (QHMM). The environment carries a hidden finite-dimensional quantum memory that evolves in time, while the agent can only probe it indirectly: at each step the agent selects an action in the form of a quantum instrument applied to the memory, producing a classical outcome and a corresponding (generally disturbed) post-instrument memory state. By choosing instruments adaptively as a function of the past action-outcome history, the agent controls both information acquisition and performance, with the goal of maximizing a cumulative reward assigned to the observed outcomes over a finite interaction horizon.

\subsection{Notation}

For an integer $n\in\mathbb{N}$ we write $[n]\coloneqq\{1,2,\dots,n\}$. We denote by $\mathcal{H}_S\simeq\mathbb{C}^S$ the (finite-dimensional) Hilbert space of the hidden memory, and by $\mathcal{H}_O\simeq\mathbb{C}^O$ the classical outcome register, identified with the span of the computational basis $\{\ket{o}\}_{o\in[O]}$ (so operators on the classical register are diagonal in this basis). We write $\mathbb{M}_d$ for the space of complex $d\times d$ matrices, and use $\mathsf{D}(\mathcal{H})$ for the set of density operators on a Hilbert space $\mathcal{H}$. Norms are $\|\cdot\|_1$ (trace norm), $\|\cdot\|_\infty$ (operator norm), and $\|\cdot\|_2$ (Hilbert-Schmidt norm). We write $X\succeq 0$ for positive semidefinite operators. We use $\mathbb{I}$ for the identity operator and $\mathrm{id}$ for the identity map (with subsystem clear from context). Linear maps between operator spaces are denoted $\mathsf{L}(\mathbb{M}_d,\mathbb{M}_{d'})$, with $\mathrm{CP}(\mathrm{TP})$ standing for completely positive (trace-preserving) maps.

A POVM on $\mathcal{H}_S$ with $O$ outcomes is a collection $\mathsf{M}=\{M_x\}_{x\in[O]}\subset\mathbb{M}_S$ such that $M_x\succeq 0$ for all $x\in[O]$ and $\sum_{x\in[O]} M_x=\mathbb{I}_S$. For any given $\rho\in\mathsf{D}(\mathcal{H}_S)$, the probability of observing outcome $x$ follows the Born rule $\Pr(x)=\Tr(M_x\rho)$. While a POVM fully characterizes the statistics of the measurement outcome, it does not dictate the post-measurement state. We use the instrument to capture this full physical interaction between the measurement device and the quantum system. Formally, the instrument $\mathcal{P}:\mathbb{M}_{S}\to \mathbb{M}_{S}\otimes \mathbb{M}_{O}$ associated with the POVM $\lbrace M_x\rbrace_{x\in[O]}$ is defined as 
\begin{align}
    \mathcal{P}(X)\coloneqq \sum_{x\in[O]} \Phi_x (X) \otimes \proj{x},
\end{align}
where $\{\ket{x}\}$ forms the computational basis of a classical register recording the outcome and the set of maps $\lbrace \Phi_x: \mathbb{M}_{S}\to \mathbb{M}_{S}\rbrace_{x\in[O]}$ are completely positive trace non-increasing (CPTNI) as well as 
\begin{align}
    \Tr[\Phi_x(X)]=\Tr[M_x X] \qquad \forall X\succeq 0. 
\end{align}
A particular choice of the instrument is the \emph{L\"uders} instrument
\begin{align}
    \label{eq}
    \Phi_x(X) \coloneqq \sqrt{M_x}X\sqrt{M_x}. 
\end{align}

For multipartite systems we use $\mathrm{Tr}$ for the full trace and $\mathrm{Tr}_S,\mathrm{Tr}_O$ for partial traces over $\mathcal{H}_S$ and $\mathcal{H}_O$, respectively. Furthermore, to simplify the notation when tracing out specific classical outcomes, we define the superoperator:
\begin{align}
    \mathcal{T}_o(X) := \Tr_{O}(\proj{o} X ). 
\end{align}

\subsection{Input-output QHMM environment}
\label{sec:QHMM}
To formalize the \emph{environment}, let the initial hidden quantum memory state be $\rho_1\in \mathsf{D}(\mathcal{H}_S)$, where $S\in\mathbb{N}$ denotes the dimension of the quantum memory. A finite-dimensional hidden quantum memory state is a direct analogue to the hidden state chosen from a finite set in hidden Markov models. We fix the length of the interaction to be $L\in\mathbb{N}$. Between rounds $l\in[L]$, the memory evolves via CPTP channels $\mathbb{E}_l:\mathbb{M}_{S}\to \mathbb{M}_{S}$ for all $l\in[L]$. 
These channels govern the internal dynamics of the memory which, together with the initial state $\rho_1$, remain strictly \emph{unknown} to the agent.

We define the \emph{action set} $\mathcal{A}$ as the collection of all instruments $\mathcal{A} = \lbrace  \mathcal{P}^{(a)}:\mathbb{M}_S\to\mathbb{M}_{S\times O}\rbrace_a$ that the agent can apply to the hidden memory, where $a$ is the classical control label specifying the quantum instrument. With abuse of notation, when it is clear from the context, we will also denote the action set $\mathcal{A}$ as the set of indices $a$ labeling all instruments. Also for finite discrete sets we will use $A := |\mathcal{A} |$ to denote the number of actions in the set. The interaction of the agent is sequential which means that at each round $l \in [L]$, the agent selects an action $a_l \in \mathcal{A}$ and measures the classical register to obtain a classical outcome $o_l \in [O]$. The sequential interaction naturally defines a trajectory of length $l$, comprising the history of chosen actions and observed outcomes which is defined as $\tau_l \coloneqq (a_1 , o_1,...,a_l,o_l)$. 

We evaluate an important quantity, the probability of observing outcome $o_l$ at round $l$ given the chosen action $a_l$ and the past trajectory $\tau_{l-1}=(a_1,o_1,\dots,a_{l-1},o_{l-1})$ as follows. We propagate the (subnormalized) prior hidden memory state at round $l$ conditioned on the past trajectory $\tau_{l-1}$, which captures the cumulative effect of instruments and memory channels associated with $\tau_l$: 
\begin{equation}\label{eq:unnormalized_filter}
    \tilde{\rho}_{1} := \rho_{1}, \qquad \tilde{\rho}_{l}(\tau_{l-1}) := \big(\mathbb{E}_{l-1}\circ\Phi_{o_{l-1}}^{(a_{l-1})}\circ\dots\circ\mathbb{E}_{1}\circ\Phi_{o_{1}}^{(a_{1})}\big)(\rho_{1}), \qquad l \geq 2.
\end{equation}

Given this subnormalized prior state, the conditional probability of observing outcome $o_l$ upon selecting action $a_l$ follows directly from Born's rule:
\begin{equation}
    \mathrm{Pr}(o_l|\tau_{l-1},a_l) = \frac{\Tr \!\big[\Phi^{(a_l)}_{o_l}\!\left(\tilde{\rho}_{l}(\tau_{l-1})\right)\big]}{\Tr(\tilde{\rho}_{l}(\tau_{l-1}))}.
    \label{eq:conditional_prob_compact}
\end{equation}

Upon choosing $a_l$ and observing $o_l$, the subnormalized post-measurement state is $\Phi_{o_l}^{(a_l)}(\tilde{\rho}_l(\tau_{l-1}))$, determined by the update rule for quantum instruments. This post-measurement state then evolves via the memory channel to the prior state for the subsequent round $(\mathbb{E}_l\circ\Phi_{o_l}^{(a_l)})( \tilde{\rho}_l(\tau_{l-1}))$, justifying our definition for $\tilde{\rho}_l(\tau_{l-1})$ for all $l\in[L]$. 

We use the following compact notation to collect all parameters that define the QHMM introduced above (initial memory state, inter-round memory evolution, and the action-dependent instruments)
\begin{align}
    \omega \coloneqq \Big(\rho_{1}, \mathbb{E}_{1},\ldots,\mathbb{E}_{L-1}, \big\lbrace\Phi^{(a)}_{o}\big\rbrace_{a\in\mathcal{A},\,o\in[O]} \Big),
\end{align}
where $\rho_{1}\in\mathbb{M}_{S}$ is a density operator, each $\mathbb{E}_{l}:\mathbb{M}_{S}\to \mathbb{M}_{S}$ is a CPTP map, and for every action $a\in\mathcal{A}$ the collection $\big\lbrace\Phi^{(a)}_{o}\big\rbrace_{o\in[O]}$ consists of CPTNI maps satisfying the instrument constraint that $\sum_{o\in[O]}\Phi^{(a)}_{o}$ is trace preserving. We let $\Omega$ denote the set of all such admissible parameter ensembles, i.e.,
$\omega\in\Omega$ if and only if the above CPTNI and instrument constraints hold. 

After defining all the objects, we provide a compact definition of our model.

\begin{definition}[Input-output QHMM environment]
\label{def:qrl_model}
Fix integers $S,O,L\in\mathbb{N}$ and an action set $\mathcal{A}$. An \emph{$L$-step QHMM environment} $\omega$ is a parameter ensemble
\begin{align}
    \omega \coloneqq \Big( \rho_{1}, \mathbb{E}_{1},\ldots,\mathbb{E}_{L-1}, \big\lbrace\Phi^{(a)}_{o}\big\rbrace_{a\in\mathcal{A},\,o\in[O]} \Big)\in\Omega,
    \label{eq:omega_def}
\end{align}
where $\rho_{1}\in\mathsf{D}(\mathcal{H}_S)$ is the initial memory state, each $\mathbb{E}_{l}:\mathbb{M}_{S}\to \mathbb{M}_{S}$ is a CPTP map, and for every $a\in\mathcal{A}$ the collection $\big\lbrace\Phi^{(a)}_{o}\big\rbrace_{o\in[O]}$ consists of CPTNI maps satisfying the instrument condition that $\sum_{o\in[O]}\Phi^{(a)}_{o}$ is trace preserving. The agent interacts with $\omega_L$ sequentially over $L$ rounds: at each round $l\in[L]$ it selects an action $a_l\in\mathcal{A}$ (possibly adaptively as a function of the past trajectory $\tau_{l-1}$) and observes an outcome $o_l\in[O]$. Given any past trajectory $\tau_{l-1}$, the conditional outcome distribution is
\begin{align}
    \Pr(o_l|\tau_{l-1},a_l)= \frac{\Tr\!\Big[\Phi^{(a_l)}_{o_l} (\tilde{\rho}_l(\tau_{l-1}))\Big]}{\Tr(\tilde{\rho}_l(\tau_{l-1}))},
\end{align}
where the (subnormalized) prior hidden memory state $\tilde{\rho}_l(\tau_{l-1})$ is described in~\eqref{eq:unnormalized_filter}. The subnormalized post-measurement state is $\Phi^{(a_l)}_{o_l}(\tilde{\rho}_l(\tau_{l-1}))$ and the subnormalized prior hidden state in the subsequent round is $(\mathbb{E}_l\circ \Phi_{o_l}^{(a_l)})(\tilde{\rho}_l(\tau_{l-1}))$.
\end{definition}

\subsection{Sequential-emission HMM environment}
\label{sec:sequential_HMM}
While our framework assumes direct interaction with the memory, it also captures the essence of another practical setup, the sequential-emission hidden Markov model (HMM) environment, where the hidden memory evolves untouched and at each step it emits auxiliary quantum probe system from which we can gain information about the unknown process. In this model, at each round $l$, the hidden memory $S$ is mapped by a CPTP map
$\mathbb{E}_l:\mathbb{M}_{S}\to \mathbb{M}_{S}\otimes \mathbb{M}_{S_O}$, producing an auxiliary output system $S_O$, and the action $a_l$ measures the POVM $\{M_{o_l}^{(a_l)}\}_{o_l\in[O]}$ on $S_O$. This model reduces to the instrument-on-memory formulation by pushing the measurement back onto $S$: define, for each outcome $o_l\in[O]$,
\begin{align}
    \Phi^{(a_l)}_{o_l}(X) \coloneqq \Tr_{S_O}\!\big[(\mathbb{I}_S \otimes M_{o_l}^{(a_l)})\,\mathbb{E}_l(X)\big].
\end{align}
Each map $\Phi^{(a_l)}_{o_l}$ is CPTNI, and $\sum_{o_l\in[O]}\Phi^{(a_l)}_{o_l}$ is CPTP. Hence $\big\lbrace\Phi^{(a_l)}_{o_l}\big\rbrace_{o_l\in[O]}$ defines a valid instrument on the memory, as in our main model.

\subsection{Policies}

The agent's strategy is governed by a policy $\pi = \{\pi_l\}_{l=1}^L$. Each $\pi_l(\cdot|\tau_{l-1})$ is a conditional probability distribution over the action set $\mathcal{A}$, strictly dependent on the past trajectory $\tau_{l-1}$.
We incorporate an additional reward layer: at each step $l$, upon selecting action $a_l \in \mathcal{A}$ and observing outcome $o_l \in [O]$, the agent receives a scalar reward given by a known, bounded reward function $r_l: \mathcal{A} \times [O] \to \mathbb{R}$ for all $l\in[L]$.

\begin{remark}
    While standard POMDP literature often assumes $r_l$ depends only on $o_l$, our framework and regret bounds hold identically for any known action-dependent reward function as long as it is uniformly bounded by some constant $R$, i.e., $|r_l(a, o)| \le R$. We exploit this flexibility in our thermodynamic application in Section~\ref{sec:work_extraction}.
\end{remark}

The interaction between the policy $\pi$ and the QHMM environment $\omega$ induces a joint probability distribution over the past trajectory $\tau_l$. Fixing the policy $\pi$ and a QHMM environment $\omega$, we denote the probability of a specific trajectory as
\begin{align}\label{eq:full_probabilities}
    \mathbb{P}^\pi_\omega(\tau_l) = \Pr(a_1, o_1, \dots, a_l, o_l).
\end{align}
By the chain rule of probability, this joint distribution naturally factorizes into two distinct components: the action sequence governed by the agent, and the observation sequence governed by the environment. Concretely, we can write:
\begin{align}
    \mathbb{P}^\pi_\omega(\tau_l) = \pi(\tau_l) A(\tau_l),
\end{align}
where we have introduced the compact shorthands for the \emph{cumulative policy distribution product} and the \emph{cumulative observation likelihood product}, respectively:
\begin{align}\label{eq:policy_shorthand}
   \pi(\tau_l) \coloneqq \prod_{l'=1}^l \pi(a_{l'}|\tau_{l'-1}), \qquad  
   A(\tau_l) \coloneqq \prod_{l'=1}^l \Pr(o_{l'} |\tau_{l'-1}, a_{l'}).
\end{align}

In our QHMM model, the conditional probabilities $\Pr(o_{l'}| \tau_{l'-1},a_{l'})$ are given by the Born-rule expression induced by the CPTNI maps $\big\lbrace\Phi^{(a_{l'})}_{o_{l'}}\big\rbrace$ and the prior hidden memory state $\tilde{\rho}_l(\tau_{l-1})$ via~\eqref{eq:conditional_prob_compact}. Indeed, the factorization \eqref{eq:full_probabilities} follows from the chain rule:
\begin{align}
    \mathbb{P}^\pi_\omega(\tau_l) = \Pr(a_1,\ldots,a_l,o_1,\ldots,o_l)
    = \prod_{l'=1}^l \Pr(a_{l'},o_{l'}|\tau_{l'-1}) = \prod_{l'=1}^l \pi(a_{l'} | \tau_{l'-1})\Pr(o_{l'} | \tau_{l'-1},a_{l'}).
\end{align}
The reward $V^\pi(\omega)$ of a policy $\pi$ for a QHMM $\omega$ (when interacting for $L$ steps) is now defined in terms of rewards as
\begin{align}
\label{eq:value_def}
    V^\pi(\omega) = 
    \Exp_{\pi,\omega} \!\left[\sum_{l=1}^L r_l(a_l,o_l)\right] = \sum_{\tau_L} \mathbb{P}^\pi_\omega(\tau_L) \sum_{l=1}^L r_l(a_l,o_l) ,
\end{align}
where the expectation is taken over the randomness of the policy and the observation process.

\subsection{Objectives} 

We formalize the learning procedure as an episodic interaction. The agent is granted access to $K$ independent episodes of an identical, unknown QHMM environment $\omega$, where each episode consists of a fixed horizon of $L$ sequential rounds. At the start of any episode $k \in [K]$, the agent sticks to a policy $\pi^k$. Throughout the episode, for each step $l = 1, \dots, L$, the agent samples an action $a_l$ according to $\pi^k_l(\cdot|\tau_{l-1})$ and applies the corresponding quantum instrument $\mathcal{P}^{(a_l)}$ to the hidden memory. This physical interaction yields a classical outcome $o_l \in [O]$ and an associated scalar reward $r_l(a_l,o_l)$. Upon concluding the $L$-step interaction, the agent collects the total empirical reward $\sum_{l=1}^L r_l(a_l,o_l)$. Finally, the agent utilizes the accumulated trajectory data to update its policy to $\pi^{k+1}$. The QHMM environment is then discarded to conclude the current episode and initialized to start the next episode.

The ultimate goal of the agent is to maximize the expected cumulative reward interacting with $K$ independent episodes of the same underlying QHMM environment $\omega$. Equivalently, this corresponds to minimizing the cumulative regret defined as
\begin{align}\label{eq:regret}
    \regret\Big(K,\omega,\lbrace \pi^k \rbrace_{k=1}^K\Big)
    =
    \sum_{k=1}^K \Big( V^*(\omega) - V^{\pi^k}(\omega)\Big),
\end{align}
where $V^*(\omega)$ is the reward with the optimal policy $\pi$ for the environment i.e
\begin{align}
    V^*(\omega) \coloneqq \sup_{\pi'} V^{\pi'}(\omega).
\end{align}
To simplify notations, we will use $\regret(K)$ when the environment $\omega$ and the policies $\lbrace \pi^k \rbrace_{k=1}^K$ are clear from the context.

\section{Observable Operator Model}\label{sec:oom}

\subsection{Undercomplete assumption}
In order to exclude pathological instances and ensure the identifiability of the latent memory from observable data, we must impose an undercompleteness assumption on the quantum instruments. We begin by recalling the classical precedents. In the seminal work on learning HMM environments via spectral algorithms~\cite{hsu2012spectral}, identifiability relies heavily on a rank condition: the observation matrix must possess full column rank. This requirement effectively rules out problematic cases where the output distribution of one state is merely a convex combination of others. This spectral separation concept was subsequently extended to reinforcement learning in classical POMDPs through the $\alpha$-weakly revealing condition~\cite{jin2020sample, liu2022partially, liu2023optimistic}. By robustly lower-bounding the $S$-th singular value of the emission matrices ($\min_h \sigma_S(\mathbb{O}_h) \ge \alpha > 0$), the condition guarantees that the observable data contains sufficient information to distinguish between different mixtures of latent states.

An analogous phenomenon arises in the quantum setting. In~\cite{fanizza2023learning}, the authors study the spectral reconstruction of \emph{finitely correlated states}, which can be viewed as one-dimensional matrix product
density operator (equivalently, sequential quantum channel) representations with a finite memory dimension. Their reconstruction and its stability analysis also rely on an \emph{invertibility} condition ensuring that suitable observable marginals determine the underlying hidden memory state dynamics. Similarly, to characterize Virtual Quantum Markov Chains, Wang et al.~\cite{wang_quantummarkovchain} introduce a  \emph{virtual recovery map} that perfectly reconstructs a global multipartite state directly from its local subsystem marginal. This condition guarantees that all global measurement statistics can be exactly recovered via local operations and classical post-processing.

Motivated by these classical and quantum foundations, we generalize this identifiability requirement to our input-output QHMM environment. The direct quantum analogue of a full-rank classical emission matrix is the existence of a valid linear recovery map $R^{(a)}$ that perfectly reconstructs the post-measurement state from the classical marginals. We formalize this in the following assumption.

\begin{assumption}[Undercompleteness]\label{ass:undercomplete}
For every $a\in\mathcal{A}$, there exists a linear map
$R^{(a)}:\mathbb{M}_{O}\to \mathbb{M}_{S}\otimes \mathbb{M}_{O}$
such that for all $X\in\mathbb{M}_{S}$,
\begin{equation}
    \big(R^{(a)} \circ \Tr_{S}\circ \mathcal{P}^{(a)}\big)\!(X)=\mathcal{P}^{(a)}(X).
\end{equation}
\end{assumption}

However, the mere existence of a recovery map is insufficient for sample-efficient learning; to ensure statistical tractability, we must bound the sensitivity of this inversion. Thus, the classical $\frac{1}{\alpha}$ singular value bound is naturally replaced in our framework by an upper bound on the induced trace norm of the recovery map, which we denote as $\kappa_{\mathrm{uc}}$. In the observable quasi-realization model of~\cite{fanizza2023learning}, the analogue of this condition was expressed in terms of the minimum singular value of a certain linear map constructed from the state marginals. That formulation was conceived to be fully-general and allow learning of linear models beyond those described by quantum memory. In this paper we focus the attention on purely quantum models, and this also allows us to be more explicit about $\kappa_{\mathrm{uc}}$: it only depends on the (possibly known) instruments, therefore one can choose the instruments in advance to have a controlled $\kappa_{\mathrm{uc}}$. Imposing this robustness constraint leads to the formal definition of our $\kappa_{\mathrm{uc}}$-robust QHHM environment class.

\begin{definition}[$\kappa_{\mathrm{uc}}$-robust QHHM environment class]\label{def:Omega_uc}
Fix a constant $\kappa_{\mathrm{uc}}\in(0,\infty)$. Let $\Omega_{\mathrm{uc}}\subseteq\Omega$ be the set of
$L$-step QHMM environments $\omega$ such that for every $a\in\mathcal A$ there exists a linear map
$R^{(a)}:\mathbb M_{O}\to \mathbb M_{S}\otimes \mathbb M_{O}$ satisfying
\begin{equation}
    R^{(a)}\circ \Tr_S\circ \mathcal P^{(a)}=\mathcal P^{(a)},
    \qquad
    \|R^{(a)}\|_{1\to 1,\diag}\le \kappa_{\mathrm{uc}},
\end{equation}
where
\begin{equation}
    \|R^{(a)}\|_{1\to 1,\diag}\coloneqq\sup\{\|R^{(a)}(X)\|_1:X=\diag(X), \|X\|_1=1\}.
\end{equation}
\end{definition}

\subsection{Observable operators}\label{sec:oom_uc}

In standard formulations, computing the probability of an observation sequence requires explicitly tracking the evolution of the hidden memory state. An elegant alternative is to bypass the latent space entirely by expressing the probabilities of future observations as sequential products of linear operators acting strictly on the classical register. 

While this formulation has become a standard tool for sample-efficient learning in classical HMMs and POMDPs under the ``Observable Operator Model'' (OOM) moniker introduced by Jaeger~\cite{jaeger2000observable}, the underlying mathematical framework was extensively developed in the control theory literature as the \emph{quasi-realization} problem~\cite{vidyasagar2005realization, vidyasagar2011complete}. We adopt the term OOM throughout this work to maintain seamless alignment with the modern reinforcement learning literature~\cite{hsu2012spectral, jin2020sample, liu2022partially}, while explicitly acknowledging its rich control-theoretic origins.

Crucially, the undercompleteness assumption and the $\kappa_{\mathrm{uc}}$-robust QHMM environment class established above guarantee that we can construct a faithful OOM (or quasi-realization) representation for our quantum setting. We formalize these superoperators as follows.

\begin{definition}[OOM superoperators on the classical register]\label{def:oom_qhmm}
Fix an instance $\omega\in\Omega_{\mathrm{uc}}$.
For each step $l\in[L-1]$, outcome $o\in[O]$, and actions $a,a'\in\mathcal{A}$, define the \emph{OOM operator} $\mathbf{A}_l(o,a,a'):\mathbb{M}_{O}\to \mathbb{M}_{O}$ as the linear \emph{superoperator}
\begin{align}
    \mathbf{A}_l(o,a,a') \coloneqq \Tr_{S}\circ \mathcal{P}^{(a')}\circ\big(\mathbb{E}_l \otimes \mathcal{T}_o\big)\circ R^{(a)}.
\end{align}
where $\mathcal{T}_{o}(X)=\Tr_O(\proj{o} X)$. Moreover, define the initial observable vector $\mathbf{a}:\mathcal{A}\to \mathbb{M}_{O},$ as
\begin{align}\label{eq:initial_observable}
    \mathbf{a}(a)\coloneqq \big(\Tr_{S}\circ \mathcal{P}^{(a)}\big)\!(\rho_1).
\end{align}
As $\mathcal{P}^{(a')}$ and $\mathcal{P}^{(a)}$ are instruments, both $\mathbf{A}_l(o,a,a')$ and $\mathbf{a}(a)$ lie in the diagonal subspace of $\mathbb{M}_{O}$: for every $X\in\mathbb{M}_{O}$, $\mathbf{A}_l(o,a,a')[X]$ is diagonal in the computational basis; $\mathbf{a}(a)$ is also diagonal.
\end{definition}

\subsection{Trajectory probabilities}\label{sec:trajectory_probab}

The primary utility of the OOM framework is that it allows us to compute the probability of any sequence of observations without referencing the latent quantum state. In particular, for an instance $\omega\in\Omega$, the prefix-likelihood $A_\omega(\tau_l) = \prod_{l'=1}^l \Pr(o_{l'}|\tau_{l'-1}, a_{l'})$ can be expressed purely in terms of the operators $\mathbf{A}_l$ and the initial vector $\mathbf{a}$, as formalized in the following lemma.

\begin{lemma}\label{lem:probabilities_operator_multiquantum}
For any $l\in[L]$ and trajectory $\tau_l=(a_1,o_1,\dots,a_l,o_l)$,
\begin{equation}
\label{eq:oom_probabilitiesquantum_clean}
    A_\omega(\tau_l) = \mathcal{T}_{o_l} \mathbf{A}_{l-1}(o_{l-1},a_{l-1},a_l)\cdots\mathbf{A}_{1}(o_{1},a_{1},a_{2})\,\mathbf{a}(a_1).
\end{equation}
In particular, for $l=1$,
\begin{equation}
    \Pr(o_1|a_1) = \mathcal{T}_{o_1}\,\mathbf{a}(a_1).
\end{equation}
\end{lemma}

\begin{proof}
Given a trajectory prefix $\tau_{l-1}$ and a candidate next action $a_l$, define the (subnormalized) \emph{observable state}
\begin{equation}
\label{eq:observable_state_def}
    \mathbf{x}_l(\tau_{l-1},a_l) \coloneqq \Tr_S \big[ \mathcal{P}^{(a_l)}(\tilde{\rho}_l(\tau_{l-1}))\big] \in \mathbb{M}_{O},
\end{equation}
where $\tilde{\rho}_l(\tau_{l-1})$ is the subnormalized prior hidden memory state from~\eqref{eq:unnormalized_filter}.
By construction,
\begin{equation}
    \label{eq:born_from_x}
    \Pr(o_l|\tau_{l-1},a_l) = \frac{\mathcal{T}_{o_l}\!\big[\mathbf{x}_l(\tau_{l-1},a_l)\big]}
    {\Tr\!\big[\tilde{\rho}_l(\tau_{l-1})\big]}.
\end{equation}
We start by showing the one-step recursion. Noting $\tau_l=(\tau_{l-1},a_l,o_l)$, for all $l\in[L-1]$: 
\begin{equation}
    \label{eq:x_recursion}
    \mathbf{x}_{l+1}(\tau_l,a_{l+1}) = \mathbf{A}_l(o_l,a_l,a_{l+1})\,\mathbf{x}_l(\tau_{l-1},a_l).
\end{equation}
Fix $l$ and abbreviate $\mathbf{x}_l\coloneqq \mathbf{x}_l(\tau_{l-1},a_l)$ and $\tilde{\rho}_l \coloneqq \tilde{\rho}_l(\tau_{l-1})$.
By~\eqref{eq:observable_state_def}, $\mathbf{x}_l=\Tr_S\big[\mathcal{P}^{(a_l)}(\tilde{\rho}_l)\big]$. 
By Assumption~\ref{ass:undercomplete},
\begin{equation}
    \label{eq:recover_joint}
    R^{(a_l)}(\mathbf{x}_l)\;=\;\mathcal{P}^{(a_l)}(\tilde{\rho}_l) \in\mathbb{M}_{S}\otimes\mathbb{M}_{O}.
\end{equation}
Conditioning on outcome $o_l$ and propagating the memory via $\mathbb{E}_l$, we get
\begin{align}
    \tilde{\rho}_{l+1}(\tau_l)
    =\big(\mathbb{E}_l\circ \Phi^{(a_l)}_{o_l}\big)\!(\tilde{\rho}_l)
    =
    \mathbb{E}_l \!\Big((\id_S\otimes \mathcal{T}_{o_l})\!\big(\mathcal{P}^{(a_l)}(\tilde{\rho}_l)\big)\Big) =\mathbb{E}_l\!\Big((\id_S\otimes \mathcal{T}_{o_l})\!\big(R^{(a_l)}(\mathbf{x}_l)\big)\Big),
\end{align}
where we used the identity $ \Phi^{(a_l)}_{o_l} (X) = (\id_S\otimes \mathcal{T}_{o_l})\,\mathcal{P}^{(a_l)} (X) $ in the first step and \eqref{eq:recover_joint} in the last step.
Now apply the next action $a_{l+1}$ and trace out $S$ to get
\begin{align}
    \mathbf{x}_{l+1}  =\Tr_S\!\Big[\mathcal{P}^{(a_{l+1})}\big(\tilde{\rho}_{l+1}(\tau_l)\big)\Big] =\Tr_S\!\Big[\mathcal{P}^{(a_{l+1})}\Big(\mathbb{E}_l\big((\id_S\otimes\mathcal{T}_{o_l})R^{(a_l)}(\mathbf{x}_l)\big)\Big)\Big].
\end{align}
Recognizing the composition in Definition~\ref{def:oom_qhmm} yields the recursion in~\eqref{eq:x_recursion}.

To establish the full probability sequence, we first note that the conditional probability of the outcome is given by:
\begin{equation}
    \mathrm{Pr}(o_l |\tau_{l-1}, a_l) = \frac{\mathcal{T}_{o_l}(\mathbf{x}_l)}{\Tr[\mathbf{x}_l]}.
\end{equation}
Crucially, we must relate the trace of the observable state at step $l+1$ to the state at step $l$. The update of the subnormalized hidden memory state implies:
\begin{equation}
    \Tr[\tilde{\rho}_{l+1}(\tau_l)] = \Tr\Big[\big(\mathbb{E}_l \circ \Phi^{(a_l)}_{o_l}\big)\!(\tilde{\rho}_l(\tau_{l-1}))\Big] = \Tr\big[\Phi^{(a_l)}_{o_l}(\tilde{\rho}_l(\tau_{l-1}))\big] = \mathcal{T}_{o_l}\big(\mathbf{x}_l\big),
\end{equation}
where the second equality utilizes the trace-preserving property of the channel $\mathbb{E}_l$. Because the instrument $\mathcal{P}^{(a_{l+1})}$ is completely positive and trace-preserving, we have:
\begin{equation}
    \Tr [\mathbf{x}_{l+1}] = \Tr\big[\mathcal{P}^{(a_{l+1})}(\tilde{\rho}_{l+1}(\tau_l))\big] = \Tr[\tilde{\rho}_{l+1}(\tau_l)] = \mathcal{T}_{o_l}\big(\mathbf{x}_l\big).
\end{equation}
This exact relation allows us to evaluate the prefix-likelihood as a telescoping product:
\begin{equation}
    A_\omega(\tau_l) = \prod_{l'=1}^l \Pr(o_{l'} | \tau_{l'-1}, a_{l'}) = \prod_{l'=1}^l \frac{\mathcal{T}_{o_{l'}}(\mathbf{x}_{l'})}{\Tr[\mathbf{x}_{l'}]} = \frac{\mathcal{T}_{o_l}(\mathbf{x}_l)}{\Tr[\mathbf{x}_1]}.
\end{equation}
Finally, since $\Tr[\mathbf{x}_1] = \Tr[\tilde{\rho}_1] = \Tr[\rho_1] = 1$, and by unfolding the recursion $\mathbf{x}_l = \mathbf{A}_{l-1} \dots \mathbf{A}_1 \mathbf{a}(a_1)$, we obtain the desired result
\begin{equation}
    A_\omega(\tau_l) = \mathcal{T}_{o_l}(\mathbf{x}_l) = \mathcal{T}_{o_l} \mathbf{A}_{l-1}(o_{l-1}, a_{l-1}, a_l) \dots \mathbf{A}_1(o_1, a_1, a_2) \mathbf{a}(a_1).
\end{equation}
The base case for $l=1$ follows immediately from $\mathbf{a}(a_1) = \Tr_S[\mathcal{P}^{(a_1)}(\rho_1)]$. 
\end{proof}

\section{Optimistic Maximum Likelihood Estimation for QHMMs}
\label{sec:omle_discrete}

We follow closely the optimistic learning framework for partially observable models developed in~\cite{liu2022partially}. The overall structure is that of an upper-confidence-bound-type algorithm: at each episode, the agent maintains a confidence region over the model class, selects an optimistic model within this region, and computes a policy that is optimal for this optimistic instance. The executed policy generates new trajectories, which are then used to refine the of the underlying parameters and to update the confidence set.

A distinctive feature of the approach in~\cite{liu2022partially} is that the confidence region is constructed around a Maximum Likelihood Estimator (MLE) of the model, rather than around moment-based or least-squares estimates. We adopt this MLE-based confidence construction in our setting and adapt it to the input-output QHMM model described above. The precise procedure is detailed in Algorithm~\ref{alg:omle-qhmm}. In the following, we sketch the main components of this strategy.

The Optimistic Maximum Likelihood Estimator (OMLE) maintains a confidence set $\mathcal{C}_k\subseteq\Omega$ of plausible QHMM models.
Let $D_{k}=\{(\pi^i,\tau^i)\}_{i=1}^{k}$ denote the dataset of trajectories and policies collected at the episode $k$. The main parts of the algorithm are the following.

\paragraph{Exploration-Exploitation.}(Optimistic planning): At the beginning of each episode $k\in[K]$, we apply the optimistic principle to select a model and policy pair that maximizes the expected reward within our confidence region
\begin{align}
\label{eq:argmax}
    (\omega_k,\pi^k)\in \arg\max_{\omega'\in \mathcal{C}_{k},\ \pi'} V^{\pi'}(\omega').
\end{align}

\paragraph{Data collection and update.}
Execute $\pi^k$ in the true environment $\omega$ for $L$ steps, observe $\tau^k$, add $(\pi^k,\tau^k)$ to the dataset,
and update $D_{k}$. 

\paragraph{Estimation.}(Confidence set construction): For a candidate model $\omega' \in \Omega$, we define the log-likelihood of the data set 
\begin{align}
    \mathcal{L}_k(\omega')\coloneqq\sum_{i=1}^{k}\log \mathbb{P}_{\omega'}^{\pi^i}(\tau^i) = \sum_{i=1}^{k}\log \pi^i(\tau^i) + \sum_{i=1}^{k}\log A_{\omega'}(\tau^i).
\end{align}
Using a radius $\beta>0$, we construct a confidence region that restricts the parameters to the undercomplete class $\Omega_{\mathrm{uc}}$ defined as:
\begin{align}\label{eq:confidence_mle}
\mathcal{C}_1\coloneqq \Omega_{\mathrm{uc}}, \qquad \mathcal{C}_{k+1}
\coloneqq  \Big\{\omega'\in\Omega: \mathcal{L}_k(\omega')\geq \max_{\omega''\in\Omega}\mathcal{L}_k(\omega'')-\beta \Big\} \cap \mathcal{C}_1.
\end{align}

\begin{remark}
Not that since $\pi^i(\tau^i)$ does not depend on $\omega'$, maximizing $\mathcal{L}_k(\omega')$ is equivalent to maximizing
$\sum_{i=1}^{k}\log A_{\omega'}(\tau^i)$. Also note that when $\beta=0$, this reduces to the  MLE solution set.
\end{remark}

\begin{remark}
The reason why we enforce the intersection with $\mathcal{C}_1 = \Omega_{\mathrm{uc}}$ is that the model parameters we optimize must satisfy the undercomplete assumption (Assumption~\ref{ass:undercomplete}). We will use the notation $\mathbf{A}^k_l(o,a,a')$ and $\mathbf{a}^k(a)$ for the OOM of $\omega_k$ given by Definition~\ref{def:oom_qhmm}.
\end{remark}

\begin{algorithm}[t]
    \caption{OMLE for the input-output QHMM model}
    \label{alg:omle-qhmm}
    \DontPrintSemicolon
    \SetKwInOut{KwRequire}{Require}
    \KwRequire {Horizon $L$, action set $\mathcal A$, outcomes $[O]$, rewards $\{r_\ell\}_{\ell=1}^L$, model class $\Omega$, radius $\beta > 0 $.\;}
    \BlankLine

    Initialize dataset $D_0\gets\emptyset$ and confidence set $\mathcal{C}_1\gets \Omega_{\mathrm{uc}}$.\;

    \For{$k=1,2,\dots,K$}{

    \tcp{\textbf{1. Exploration-Exploitation (Optimistic Planning)}}
    Choose $(\omega_k,\pi^k)\in\arg\max\limits_{\omega'\in \mathcal{C}_k ,\ \pi'} V^{\pi'}(\omega')$.\;

    \tcp{\textbf{2. Rollout on the QHMM}}
    Execute $\pi^k$ for $L$ rounds and observe $\tau^k=(a_1,o_1,\dots,a_L,o_L)$:\;
    \For{$\ell=1,2,\dots,L$}{
        Sample $a_\ell\sim \pi^k_\ell(\cdot|\tau_{\ell-1})$.\;
        Apply the instrument $\mathcal P^{(a_\ell)}=\{\Phi^{(a_\ell)}_{o}\otimes \proj{o}\}_{o\in[O]}$ to the hidden memory.\;
        Observe $o_\ell\in[O]$ and update the memory state accordingly.\;
        Receive reward $r_\ell(a_\ell,o_\ell)$.\;
    }

    \tcp{\textbf{3. Update dataset}}
    $D_{k}\gets D_{k-1}\cup\{(\pi^k,\tau^k)\}$.\;

    \tcp{\textbf{4. Estimation (MLE / Confidence set update)}}
    $\hat\omega_k \in \arg\max\limits_{\omega'\in \Omega}\ \sum_{(\pi^i,\tau^i)\in D_k}\log A_{\omega'}(\tau^i)$.\;
    $\mathcal{C}_{k+1} \gets \Bigl\{\omega'\in\mathcal{C}_k:\ \sum_{(\pi^i,\tau^i)\in D_k}\log A_{\omega'}(\tau^i) \ge \sum_{(\pi^i,\tau^i)\in D_k}\log A_{\hat\omega_k}(\tau^i)-\beta\Bigr\} \cap \Omega_{\mathrm{uc}}$.\;
    }
\end{algorithm}

\subsection{Embedding dimension of the QHMM model class (discrete actions)}
\label{subsec:embedding_dimension_qhmm}

To analyze the regret of Algorithm~\ref{alg:omle-qhmm}, we require MLE concentration bounds of the type established in~\cite{liu2022partially}. These bounds depend on the model class $\Omega$ only through the dimension of a finite-dimensional real embedding. We therefore upper bound the number of free real parameters needed to specify any $\omega\in\Omega$ by vectorizing the Choi representations of its constituent channels/instruments.

\begin{itemize}
    \item \emph{Initial state.} $\rho_1$ is Hermitian with $\Tr(\rho_1)=1$, hence it has $S^2-1$ free real parameters.
    \item \emph{Dynamics.} Each CPTP map $\mathbb{E}_l:\mathbb M_{S}\to\mathbb M_{S}$ is represented by its Choi matrix\footnote{For a linear map $\Phi: \mathbb{M}_S \to \mathbb{M}_S$, its Choi matrix is defined as $J(\Phi) \coloneqq (\Phi \otimes \text{id}_S)(\proj{\Omega})$, where $|\Omega\rangle = \sum_{i=1}^S |i\rangle \otimes |i\rangle$ is the unnormalized maximally entangled state.} $J(E_l)\in\mathbb{M}_{S^2}$, which is Hermitian PSD and satisfies $\Tr_{\mathrm{out}}(J(\mathbb{E}_l))=\mathbb{I}_S$.
    For an explicit upper bound we count at most $S^4$ real parameters per $E_l$.
    \item \emph{Instruments.} For each action $a\in\mathcal{A}$ and outcome $o\in[O]$, the CP map $\Phi^{(a)}_o$ has a Choi matrix
    $J\big(\Phi^{(a)}_o\big)\in\mathbb{M}_{S^2}$; again we count at most $S^4$ real parameters per $(a,o)$.
\end{itemize}

Altogether, this yields the explicit embedding-dimension bound
\begin{align}
    d_{\Omega}\leq  (L + AO)S^4+S^2,
\end{align}
where $A = |\mathcal{A}|$ is the cardinality of the action set.

The original proofs in~\cite[Propositions 13, 14]{liu2022partially} establish concentration bounds by upper-bounding the $\epsilon$-covering number by $\mathcal{O}((1/\epsilon)^{d})$, where $d$ is the number of free parameters in the classical tabular model. Because our model class $\Omega$ can be embedded into a bounded subset of $\mathbb{R}^{d_{\Omega}}$ via the Choi matrices above, our $\epsilon$-covering number is similarly upper-bounded by $\mathcal{O}((1/\epsilon)^{d_\Omega})$. Consequently, their likelihood-class arguments apply identically to our QHMM model, requiring only a single modification: replacing the classical parameter dimension $d$ with our quantum embedding dimension $d_\Omega$. Thus, the concentration bounds below retain the exact form as in~\cite{liu2022partially}, with the effective dimension safely bounded by $d_\Omega$. In case the reader is interested in these specific proofs and constructions we adapted the original proofs to handle the continuous action spaces that we will study in the next section. The arguments follow the original ones that can be found in~\cite[Appendix B]{liu2022partially}. The continuous adaptations can be found in our Appendix~\ref{sec:appendix_mle_continuous}, and these proofs are also almost identical to the discrete case (just by replacing the corresponding continuous measures to discrete ones).

\begin{proposition}[Likelihood ratio concentration~{\cite[Proposition 13]{liu2022partially}}]
\label{prop:concentration_mle}
There exists an absolute constant $c>0$ such that for any $\delta\in(0,1]$, with probability at least $1-\delta$, for all $k\in[K]$ and all $\omega'\in\Omega$,
\begin{equation}
    \sum_{i=1}^k
    \log\!\left(
        \frac{\mathbb{P}_{\omega'}^{\pi^i}(\tau^i)}
             {\mathbb{P}_{\omega}^{\pi^i}(\tau^i)}
    \right)
    \leq 
    c\left((L + A O) S^4
        \log(KSAOL)
        +
        \log(K/\delta)
    \right).
\end{equation}
\end{proposition}

\begin{proposition}[Trajectory $\ell_1$ concentration~{\cite[Proposition 14]{liu2022partially}}]
\label{prop:trajectory_mle}
There exists a universal constant $c>0$ such that for any $\delta\in(0,1]$, with probability at least $1-\delta$, for all $k \in[K]$ and all $\omega'\in\Omega$,
\begin{equation}
    \sum_{i=1}^k 
    \left(
    \sum_{\tau_L}
        \left|
            \mathbb{P}_{\omega'}^{\pi^i}(\tau_L)
            - 
            \mathbb{P}_{\omega}^{\pi^i}(\tau_L)
        \right|
    \right)^2
    \leq 
    c\left(
        \sum_{i=1}^k
        \log\!\left(
            \frac{\mathbb{P}_{\omega}^{\pi^i}(\tau^i)}
                 {\mathbb{P}_{\omega'}^{\pi^i}(\tau^i)}
        \right)
        + (L + A O) S^4
        \log(KSAOL)
        +
        \log(K/\delta)
    \right).
\end{equation}
\end{proposition}
Thus, from the above propositions, we will choose the radius of our confidence region as
\begin{align}\label{eq:beta_def}
    \beta \coloneqq  c\left((L + |\mathcal A|\,O)\,S^4
        \log(KSAOL)
        +
        \log(K/\delta)
    \right).
\end{align}

\section{Regret analysis for discrete action spaces}\label{sec:regret_discrete}

\subsection{Controlling error propagation via quantum recovery maps}

Before establishing the full regret bound, a fundamental question arises: \emph{how does a local estimation error at step $l$ propagate to affect the predicted probability of an entire trajectory?} In the classical POMDP literature, bounding the suboptimality of a policy relies on controlling the error of the estimated OOM. By telescoping the product of estimated and true OOM operators over the episode length, one isolates the estimation error at each step $l$. However, because the trajectory probability is a sequential product, this local error is subsequently multiplied by the remaining true OOM operators, raising the critical risk that the error exponentially propagates forward to the end of the episode.

Classically, this forward propagation is bounded by unpacking the sequence of OOM operators into a product of sub-stochastic emission and transition matrices, followed by the pseudo-inverse of the emission matrix~\cite[Lemmas 27 and 31]{liu2022partially}. The exponential blow-up of the local error is avoided strictly because sub-stochastic maps are contractive under the $\ell_1$ norm, leaving only a single penalty bounded by the $\alpha$-weakly revealing condition.

In our quantum input-output framework, the memory is a density matrix evolving via CPTP maps, and the observations are generated by quantum instruments. Because of this underlying non-commutative operator structure, the classical matrix decomposition into sub-stochastic maps does not apply. We must independently ensure that the composition of future quantum instruments and channels does not arbitrarily amplify the local error at step $l$.

To overcome this, we utilize the undercomplete assumption (Assumption~\ref{ass:undercomplete}) to establish a fully quantum analogue to the classical contraction. Intuitively, we isolate the only potentially non-contractive part of the forward propagation, i.e., the quantum recovery map $R^{(a)}$ which serves the equivalent role of the classical pseudo-inverse. We then demonstrate that the subsequent sequence consists purely of CPTNI maps, which act as a uniform contraction on the quantum state space.

We define the concatenation of OOM operators. For a trajectory $\tau_l$, we define the internal vectors
\begin{align}
    \mathbf{a}(\tau_{l-1},a_l) \coloneqq \left(\Pi_{l'=1}^{l-1} \mathbf{A}(o_{l'},a_{l'},a_{l'+1}  ) \right)\mathbf{a}(a_1 ), \qquad
    \mathbf{a}^k(\tau_{l-1},a_l) \coloneqq \left(\Pi_{l'=1}^{l-1} \mathbf{A}^k(o_{l'},a_{l'},a_{l'+1}  ) \right)\mathbf{a}^k(a_1).
\end{align}
Also, for a trajectory $\tau_{L-1}$ and action $a_L$ we adopt the following shorthand notation 
\begin{align}
    &\mathbf{a}_{l} \coloneqq \mathbf{a}(\tau_{l-1},a_{l}), \qquad  \mathbf{A}_{l}\coloneqq \mathbf{A}(o_{l},a_{l},a_{l+1} ), \qquad \mathbf{A}_{L-1:l}\coloneqq \mathbf{A}_{L-1}\times \cdots \times \mathbf{A}_l,  \\
    &\mathbf{a}^k_{l} \coloneqq \mathbf{a}^k(\tau_{l-1},a_{l}),\qquad  \mathbf{A}_{l}^k\coloneqq \mathbf{A}^k(o_{l},a_{l},a_{l+1} ), \qquad \mathbf{A}^k_{L-1:l} \coloneqq \mathbf{A}^k_{L-1}\times \cdots \times \mathbf{A}^k_l,
\end{align}
where we will understand that there is an underlying $\tau_L$ that will be clear from the context.

The following crucial lemma formally bounds this forward propagation. It demonstrates that the accumulated trace norm is controlled entirely by the robustness of our recovery map, $\kappa_{\mathrm{uc}}$, and the dimension of the hidden quantum memory, bypassing the need for classical emission matrix arguments.

\begin{lemma}
\label{lem:one_norm_op}
    For any operator $\tilde{\mathbf{A}}\in\lbrace \mathbf{A}^k , \mathbf{A} \rbrace$ and a diagonal matrix $X\in\mathbb{M}_{O}$
    \begin{align}
        \sum_{ a_L , \tau_{L-1:l+2},o_{l+1} }  \|\tilde{\mathbf{A}}_{L-1:l+1} X \|_1 \times  \prod_{l'=l+1}^{L-1}\pi(a_{l'+1} | \tau_{l'})  \leq  S^2 \kappa_{\mathrm{uc}} \| X \|_1 .
    \end{align}
\end{lemma}

\begin{proof}
Fix $l\in\{1,\dots,L-1\}$, fix $\tilde{\mathbf{A}}\in\lbrace\mathbf{A}^k,\mathbf{A}\rbrace$, and fix a diagonal $X\in\mathbb{M}_{O}$.
Let $a_{l+1}$ denote the action at step $l+1$ in the summation.

\paragraph{Step 1: Isolate the non-contractive recovery map.}
Define the (unnormalized) recovered memory operator
\begin{align}\label{eq:recovered_memory_single_obs}
    Y_{o_{l+1}} \coloneqq  \big( \mathrm{id}_S\otimes \mathcal{T}_{o_{l+1}}\big)\!\big( R^{(a_{l+1})}(X) \big) \in \mathbb{M}_{S},
\end{align}
Next, expand the product $\tilde{\mathbf{A}}_{L-1:l+1}$. As defined in Definition~\ref{def:oom_qhmm}, each factor $\tilde{\mathbf{A}}_t(o_t,a_t,a_{t+1})$ is
\begin{align}
    \tilde{\mathbf{A}}_t =\Tr_S\circ \mathcal{P}^{(a_{t+1})}\circ (\mathbb{E}_t\otimes \mathcal{T}_{o_t})\circ R^{(a_t)}.
\end{align}
In the composition $\tilde{\mathbf{A}}_{t+1}\circ \tilde{\mathbf{A}}_t$, the consecutive fragment
$R^{(a_{t+1})}\circ \Tr_S\circ \mathcal{P}^{(a_{t+1})}$ appears. By Assumption~\ref{ass:undercomplete} this equals $\mathcal{P}^{(a_{t+1})}$.
Iterating this cancellation over $t=l+1,\dots,L-2$, we can rewrite
$\tilde{\mathbf{A}}_{L-1:l+1}$ as a composition where \emph{all} recovery maps cancel \emph{except} the first one,
so that $\tilde{\mathbf{A}}_{L-1:l+1} X$ depends on $X$ only through one non-contractive map $R^{(a_{l+1})}(X)$ (or in particular through $Y_{o_{l+1}}$). To be precise, 
\begin{equation}
\begin{aligned}
    \tilde{\mathbf{A}}_{L-1:l+1} X =\Tr_S\circ \mathcal{P}^{(a_L)}\circ (\mathbb{E}_{L-1}\otimes \mathcal{T}_{o_{L-1}})\circ \cdots \circ \mathcal{P}^{(a_{l+2})}\circ (\mathbb{E}_{l+1}\otimes \mathcal{T}_{o_{l+1}})\big(R^{(a_{l+1})}(X)\big),
\end{aligned}
\end{equation}
and by substituting~\eqref{eq:recovered_memory_single_obs}, we have
\begin{align}
    \tilde{\mathbf{A}}_{L-1:l+1} X =\Tr_S\circ \mathcal{P}^{(a_L)}\circ (\mathbb{E}_{L-1}\otimes \mathcal{T}_{o_{L-1}})\circ  \cdots\circ \mathcal{P}^{(a_{l+2})}\circ \mathbb{E}_{l+1} ( Y_{o_{l+1}}),
\end{align}
\paragraph{Step 2: A Hermitian resolution of identity on $S$.}
Let $\{P_\mu\}_{\mu=1}^{S^2}\subset \mathbb{M}_{S}$ be any \emph{Hermitian} Hilbert-Schmidt orthonormal basis, i.e.,
\begin{align}
    P_\mu^\dagger=P_\mu,
    \qquad
    \Tr(P_\mu P_\nu)=\delta_{\mu\nu}
    \qquad \forall \mu,\nu\in[S^2].
\end{align}
We will take the normalized generalized Pauli basis\footnote{For qubits one may take the normalized Pauli basis; for general $S$ one may take the normalized generalized Pauli basis, the generalized Gell-Mann basis or the standard Hermitian basis. However, choosing different basis may lead to different scalings on $S$, as is shown in Step 4.}. 
Define the coordinates
\begin{align}
    T_\mu(Z) \coloneqq \Tr(P_\mu Z).
\end{align}
Then we can use the above to get the decomposition
\begin{equation}
    \mathbb{E}_{l+1}( Y_{o_{l+1}}) = \sum_{\mu=1}^{S^2} T_\mu(\mathbb{E}_{l+1}( Y_{o_{l+1}})) P_\mu.
\label{eq:Y-decomp-hermitian}
\end{equation}
In order to make the notation more compact, fix the variables $(a_L,\tau_{L-1:l+2})$ and define the map $    \mathcal{M}_{l+2\rightarrow L}(a_L,\tau_{L-1:l+2}): \mathbb{M}_{S} \rightarrow \mathbb{M}_{O}$ as
\begin{align}
    \mathcal{M}_{l+2\rightarrow L}(a_L,\tau_{L-1:l+2}) \coloneqq \Tr_S\circ \mathcal{P}^{(a_L)}\circ (\mathbb{E}_{L-1}\otimes \mathcal{T}_{o_{L-1}})\circ \mathcal{P}^{(a_{L-1})} \circ \cdots \circ (\mathbb{E}_{l+2}\otimes \mathcal{T}_{o_{l+2}}) \circ \mathcal{P}^{(a_{l+2})}. 
\end{align}
Using $(\mathrm{id}_S\otimes \mathcal{T}_{o})\circ \mathcal{P}^{(a)} = \Phi_o^{(a)}$, we can express it as
\begin{align}
\label{eq:M_notation}
    \mathcal{M}_{l+2\rightarrow L}(a_L,\tau_{L-1:l+2})  = \Tr_S\circ \mathcal{P}^{(a_L)} \circ \mathbb{E}_{L-1}\circ \Phi^{(a_{L-1})}_{o_{L-1}} \circ \cdots \circ \mathbb{E}_{l+2}\circ \Phi^{(a_{l+2})}_{o_{l+2}}
\end{align}
By linearity of $\tilde{\mathbf{A}}_{L-1:l+1}$ and the triangle inequality, for every fixed choice of future variables $(a_L,\tau_{L-1:l+2},o_{l+1})$ we have
\begin{equation}
\|\tilde{\mathbf{A}}_{L-1:l+1}X\|_1 \leq \sum_{\mu=1}^{S^2} | T_\mu(\mathbb{E}_{l+1}( Y_{o_{l+1}}))|\;\big\|\mathcal{M}_{l+2\rightarrow L}(a_L,\tau_{L-1:l+2}) [P_\mu]\big\|_1.
\label{eq:triangle-hermitian}
\end{equation}

\paragraph{Step 3: The forward propagation is uniformly contractive.}
Fix a component $\mu\in[S^2]$. Note then that the map $\mathcal{M}_{l+2\rightarrow L}(a_L,\tau_{L-1:l+2}) $ is a concatenation of:
\begin{enumerate}
    \item  CPTP channels $\mathbb{E}_l$ on $S$.
    \item  CPTNI branches coming from the instruments with classical outcomes
    \item The final quantum-to-classical map $\Tr_S\circ \mathcal{P}^{(a_L)}$.
\end{enumerate}
Hence, each map $\mathcal{M}_{l+2\rightarrow L}(a_L,\tau_{L-1:l+2})$ is CPTNI, and therefore its induced trace norm satisfies $\|\cdot\|_{1\to 1}\le 1$. Moreover, summing over all classical outcomes and averaging over future actions with the policy probabilities (which are convex coefficients) preserves its CPTNI property at each stage.
Let $Z\in\mathbb{M}_{S}$ be Hermitian. Then using the properties of an instrument when we sum over observations as well as the fact that the map $\mathcal{M}_{l+2\rightarrow L}(a_L,\tau_{L-1:l+2})$ outputs diagonal matrices in $\mathbb{M}_{O}$ by construction, one can prove that
\begin{align}
    \sum_{a_L,\tau_{L-1:l+2}} \|\mathcal{M}_{l+2\rightarrow L}(a_L,\tau_{L-1:l+2}) [Z]\|_1 \prod_{l'=l+1}^{L-1}\pi(a_{l'+1}| \tau_{l'})
    \leq \| Z\|_1.
\end{align}
To show the above equation, we use the Jordan decomposition $Z=Z_+-Z_-$ with $Z_\pm\succeq 0$ and defined
\begin{align}
    |Z|=Z_+ + Z_- .
\end{align}
Recall that for any fixed choice of future variables $(a_L,\tau_{L-1:l+2})$, $\mathcal{M}\coloneqq \mathcal{M}_{l+2\to L}(a_L,\tau_{L-1:l+2})$ is CPTNI.
By linearity and the triangle inequality,
\begin{equation}
\begin{aligned}
    \|\mathcal{M}(Z)\|_1 &=\|\mathcal{M}(Z_+)-\mathcal{M}(Z_-)\|_1 \le \|\mathcal{M}(Z_+)\|_1+\|\mathcal{M}(Z_-)\|_1  \\
    &= \Tr(\mathcal{M}(Z_+))+\Tr(\mathcal{M}(Z_-)) = \Tr(\mathcal{M}(|Z|)) ,
    \label{eq:branch-norm-to-trace}
\end{aligned}
\end{equation}
where we used that $\mathcal{M}(Z_\pm)\succeq 0$ as $\mathcal{M}$ is completely positive so $\|\mathcal{M}(Z_\pm)\|_1=\Tr(\mathcal{M}(Z_\pm))$. Multiplying~\eqref{eq:branch-norm-to-trace} by the policy weight and summing over $\{a_L,\tau_{L-1:l+2}\}$ gives
\begin{equation}
\label{eq:reduce-to-sigma}
\begin{aligned}
    & \sum_{a_L,\tau_{L-1:l+2}} \prod_{l'=l+1}^{L-1}\pi(a_{l'+1}|\tau_{l'}) \|\mathcal{M}_{l+2\to L}(a_L,\tau_{L-1:l+2})[Z]\|_1 \\
    & \leq  \sum_{a_L,\tau_{L-1:l+2}} \prod_{l'=l+1}^{L-1}\pi(a_{l'+1}|\tau_{l'}) \Tr\!\Big(\mathcal{M}_{l+2\to L}(a_L,\tau_{L-1:l+2})[X]\Big). 
\end{aligned}
\end{equation}
We can evaluate this sum on the right-hand side by from the final step $L$ backward to $l+2$. Recall~\eqref{eq:M_notation}, at the final step $L$, as $\mathcal{P}^{a_L}$ is trace-preserving, tracing over the system and summing over the policy action gives
\begin{align}
    \sum_{a_L} \pi(a_L|\tau_{L-1}) \Tr\circ \Tr_S \circ \mathcal{P}^{(a_L)} = \sum_{a_L} \pi(a_L|\tau_{L-1}) \Tr = \Tr. 
\end{align}
For any intermediate step $l'$,  both the quantum instrument sum $ \sum_{o_{l'+1}} \Phi_{o_{l'}}^{(a_{l'+1})} $ and the channel $\mathbb{E}_{l'+1}$ are trace-preserving. Summing over the classical outcome $o_{l'+1}$ and action $a_{l'+1}$ yields:
\begin{align}
    \sum_{a_{l'+1}} \pi(a_{l'+1}|\tau_{l'}) \Tr \circ  \mathbb{E}_{l'} \circ \sum_{o_{l'+1}}  \Phi_{o_{l'+1}}^{(a_{l'+1})} = \sum_{a_{l'+1}} \pi(a_{l'+1}|\tau_{l'}) \Tr  = \Tr. 
\end{align}
Telescoping this property backward from $L$ to $l+2$ collapses the entire sum perfectly to the initial trace $\Tr(|Z|) = \|Z\|_1$. Thus:
\begin{align}
    \sum_{a_L,\tau_{L-1:l+2} } \|\mathcal{M}_{l+2\rightarrow L} (Z)\|_1 \prod_{l'=l+1}^{L-1}\pi(a_{l'+1}| \tau_{l'})
    \leq \| Z\|_1.
\end{align}

\paragraph{Step 4: Conclude and bound the remaining coefficient sum.} From the previous steps, using that the basis $\lbrace P_\mu \rbrace$ is Hermitian, we have
\begin{equation}
\begin{aligned}
  & \sum_{a_L , \tau_{L-1:l+2},o_{l+1}}    \big\|\tilde{\mathbf{A}}_{L-1:l+1}X\big\|_1 \times  \prod_{l'=l+1}^{L-1}\pi(a_{l'+1}|\tau_{l'}) \\
  & \leq \sum_{ a_L , \tau_{L-1:l+2},o_{l+1} }   \sum_{\mu=1}^{S^2} | T_\mu(\mathbb{E}_{l+1}(Y_{o_{l+1}} ))|\|\mathcal{M}_{l+2\rightarrow L}(a_L,\tau_{L-1:l+2}) [P_\mu]\|_1 \times  \prod_{l'=l+1}^{L-1}\pi(a_{l'+1}|\tau_{l'})\\
  & \leq \sum_{o_{l+1}} \sum_{\mu=1}^{S^2} | T_\mu(\mathbb{E}_{l+1}( Y_{o_{l+1}} ))|\|P_\mu\|_1 
   =\sum_{o_{l+1}}\sum_{\mu=1}^{S^2}|\Tr\big(P_\mu \mathbb{E}_{l+1}( Y_{o_{l+1}}) \big) | \|P_\mu\|_1 
\end{aligned}
\end{equation}
Finally by applying H\"older's inequality $|\Tr(AB)| \le \|A\|_\infty \|B\|_1$ with $\|P_\mu\|_\infty = \frac{1}{\sqrt{S}}$ for the normalized generalized Pauli basis\footnote{For a general normalized basis for Hermitian matrices, $\frac{1}{\sqrt{S}}\leq \| P_\mu \|_\infty \leq 1$, where $\frac{1}{\sqrt{S}}$ is achievable for the normalized generalized Pauli basis and $1$ is achievable for the normalized standard Hermitian basis.} and $\| P_\mu \|_1 \leq \sqrt{S} \| P_\mu \|_2 = \sqrt{S}$, and the fact that $\mathbb{E}_l$ is a CPTP channel
\begin{equation}
\begin{aligned}
    &\sum_{ a_L , \tau_{L-1:l+2},o_{l+1} }    \big\|\tilde{\mathbf{A}}_{L-1:l+1}X\big\|_1 \times  \prod_{l'=l+1}^{L-1}\pi(a_{l'+1}|\tau_{l'})\\
    & \leq \sum_{o_{l+1}}\sum_{\mu=1}^{S^2} \| \mathbb{E}_{l+1}( Y_{o_{l+1}} )  \|_1 \|P_\mu\|_\infty \|P_\mu \|_1\leq  \sum_{o_{l+1}}\sum_{\mu=1}^{S^2} \| \mathbb{E}_{l+1}( Y_{o_{l+1}} ) \|_1 \\
    &\leq   \sum_{o_{l+1}}\sum_{\mu=1}^{S^2} \| \mathbb{E}_{l+1} \|_{1\rightarrow1} \|Y_{o_{l+1}}  \|_1  =    \sum_{\mu=1}^{S^2}  \sum_{o_{l+1}} \|Y_{o_{l+1}}  \|_1 \leq S^2  \sum_{o_{l+1}} \|Y_{o_{l+1}}  \|_1
\end{aligned}
\end{equation}
In order to bound the last term, we introduce the block-diagonal operator
\begin{align}
    \tilde Y \coloneqq  \sum_{o\in[O]} Y_o\otimes \proj{o} = (\mathrm{id}_S\otimes \Delta_O)\!\big(R^{(a_{l+1})}(X)\big),
    \qquad
    \Delta_O(W)\coloneqq \sum_{o\in[O]} \proj{o}W\proj{o}.
\end{align}
Since $\tilde{Y}$ is block-diagonal in the $\{|o\rangle\}$ basis,
\begin{align}
    \sum_{o\in[O]}\|Y_o\|_1 = \|\tilde{Y}\|_1.
\end{align}
Moreover, $\Delta_O$ is a pinching map, hence using that the pinching map is contractive under the trace norm
\begin{align}
    \|\tilde Y\|_1 = \big\|(\mathrm{id}_S\otimes \Delta_O)\!\big(R^{(a_{l+1})}(X)\big)\big\|_1 \leq  \big\|R^{(a_{l+1})}(X)\big\|_1 \leq  \big\|R^{(a_{l+1})}\big\|_{1\to 1,\diag} \|X\|_1,
\end{align}
Therefore,
\begin{align}
    \sum_{o\in[O]}\|Y_o\|_1 \leq \big\|R^{(a_{l+1})}\big\|_{1\to 1,\diag} \|X\|_1,
\end{align}
and the result follows by including $\kappa_{\mathrm{uc}}$ in Definition~\ref{def:Omega_uc} ($\kappa_{\mathrm{uc}}$-robust QHMM environment class)
\begin{align}
    \sum_{ a_L , \tau_{L-1:l+2},o_{l+1} }    \big\|\tilde{\mathbf{A}}_{L-1:l+1}X\big\|_1 \times  \prod_{l'=l+1}^{L-1}\pi(a_{l'+1}|\tau_{l'}) \leq S^2\kappa_{\mathrm{uc}}\|X\|_1
\end{align}
\end{proof}

By establishing this bound, we have effectively bridged the gap between the quantum stochastic process and the classical analytical framework. Because Lemma \ref{lem:one_norm_op} guarantees that the quantum OOM operators safely contain error amplification within the factor $S^2 \kappa_{\mathrm{uc}}$, we can now proceed to bound the cumulative regret using the $\ell_1$-eluder dimension pigeonhole principle, seamlessly mirroring the high-level classical POMDP analysis.

\subsection{Decomposing regret via OOM estimation errors}

\begin{lemma}\label{lem:regret_step1}
Fix $\delta\in(0,1)$ and assume the rewards are uniformly bounded $|r_l(a_l, o_l)| \le R$ for all $l \in [L], a_l \in \mathcal{A}$, and $o_l \in [O]$. Let $\omega\in\Omega_{\mathrm{uc}}$ be the true model and run Algorithm~\ref{alg:omle-qhmm}
with confidence radius $\{\beta_k\}_{k\ge1}$ chosen as in~\eqref{eq:beta_def}.
Then, with probability at least $1-\delta$,
\begin{align}
  \regret (K) \leq   L R\sum_{k=1}^K \sum_{\tau_L} \big| \mathbb{P}^{\pi^k}_{\omega^k_L}(\tau_L ) - \mathbb{P}^{\pi^k}_{\omega_L}(\tau_L )  \big|  
\end{align}
In particular, we can also bound the regret in terms of the error of the OOM between the true models and estimated models $\lbrace \omega_k \rbrace_{k\in[K]}$ as
\begin{equation}
\begin{aligned}
\label{eq:regret_step1}
 \regret(K)&\leq
    S^2 \kappa_{\mathrm{uc}} LR \sum_{l=1}^{L-1}\sum_{k=1}^K     \sum_{ a_{l+1} , \tau_{l} } \pi^k(\tau_{l+1})   \times  \| (\mathbf{A}^k_l-\mathbf{A}_l) \mathbf{a}_l\|_1 \\
    & \quad + S^2 \kappa_{\mathrm{uc}} LR  \sum_{k=1}^K   \sum_{a_1 }\pi^k(a_1) \times \|\mathbf{a}^k(a_1)-\mathbf{a}(a_1) \|_1 .
\end{aligned}
\end{equation}
\end{lemma}

\begin{proof}
We start under the assumption that the confidence region defined in~\eqref{eq:confidence_mle} is well defined, i.e, it contains the true environment $(\omega \in \mathcal{C}_k)$ with high probability, for all $k\in[K]$. Then from the optimistic principle, we can upper bound the optimal value by the one of the optimistic policy as 
\begin{align}
    V^{\pi^k} (\omega) \leq V^{\pi} (\omega) \leq V^{\pi^k} (\omega_k), 
\end{align}
where $\omega$ is the true environment, $\pi$ is the optimal strategy for true state, $\omega_k$ is the $k$-th estimation, $\pi^k$ is the optimal strategy for the $k$-th estimation. Thus, the regret in~\eqref{eq:regret} can be bounded as
\begin{equation}
\begin{aligned}
    \regret (K) &\leq \sum_{k=1}^K V^{\pi^k} (\omega_k)  - V^{\pi^k} (\omega ) = \sum_{k=1}^K\Bigg(\Exp_{\pi^k,\omega_k} \Big[\sum_{l=1}^L r_l(a_l,o_l)\Big] - \Exp_{\pi^k,\omega} \Big[\sum_{l=1}^L r_l(a_l,o_l) \Big] \Bigg) \\
    &= \sum_{k=1}^K \sum_{\tau_L}\Big(\sum_{l=1}^L r_l(a_l,o_l) \Big)\Big(\mathbb{P}^{\pi^k}_{\omega_k}(\tau_L ) - \mathbb{P}^{\pi^k}_{\omega}(\tau_L) \Big) \leq L R\sum_{k=1}^K \sum_{\tau_L} \Big| \mathbb{P}^{\pi^k}_{\omega_k}(\tau_L ) - \mathbb{P}^{\pi^k}_{\omega}(\tau_L )  \Big|  ,
\end{aligned}
\end{equation}
where in the last inequality we used that the cumulative reward for each sample is bounded by $LR$. This proves the first statement of the theorem.

For the second statement, we start using that $\omega,\omega_k\in \Omega_{\mathrm{uc}}$ and we can compute their trajectory probabilities using their OOM representation in Definition~\ref{def:oom_qhmm} and Lemma~\ref{lem:probabilities_operator_multiquantum}. Then the above regret bound can be rewritten as
\begin{align}
    \regret (K) \leq L R\sum_{k=1}^K \sum_{\tau_L} \pi^k (\tau_L)\times \big|\mathcal{T}_{o_L} \big( \mathbf{A}^k_{L-1:1}\mathbf{a}^k(a_1 ) - \mathbf{A}_{L-1:1} \mathbf{a}(a_1) \big)\big|.
\end{align}
Using Definition~\ref{def:oom_qhmm} we have that $\mathbf{a}(a_1 ),\mathbf{a}^k(a_1 )$ are diagonal matrices and when we apply $\mathbf{A},\mathbf{A}^k$, they send diagonal to diagonal matrices. Thus by summing over $o_L$ we can introduce the trace norm as
\begin{align}
      \regret (K) \leq LR \sum_{k=1}^K \sum_{ a_L , \tau_{L-1}} \pi^k(\tau_L) \times \big\|\mathbf{A}^k_{L-1:1}\mathbf{a}^k(a_1 ) - \mathbf{A}_{L-1:1} \mathbf{a}(a_1 )  \big\|_1 .
\end{align}
By the triangle inequality and the telescoping expansion of 
$\mathbf{A}^k_{L-1:1}\mathbf{a}^k(a_1) - \mathbf{A}_{L-1:1}\mathbf{a}(a_1)$, we get
\begin{equation}
\label{eq:telescoping}
\begin{aligned}
    \regret(K) &\leq LR \sum_{k=1}^K  \sum_{ a_L , \tau_{L-1}} \pi^k(\tau_L) \times
    \Big(
        \big\|\mathbf{A}^k_{L-1:1}\big(\mathbf{a}^k(a_1)-\mathbf{a}(a_1)\big)\big\|_1
        +
        \big\|\big(\mathbf{A}^k_{L-1:1}-\mathbf{A}_{L-1:1}\big)\mathbf{a}(a_1)\big\|_1
    \Big) \\
    & \leq
    LR \sum_{l=1}^{L-1} \sum_{k=1}^K   \sum_{ a_L , \tau_{L-1} }   \pi^k(\tau_L)  \times 
        \big\|
            \mathbf{A}^k_{L-1:l+1}
            \big(\mathbf{A}^k_{l}-\mathbf{A}_{l}\big)
            \mathbf{a}_{l}
        \big\|_1  \\
    & \quad+ LR \sum_{k=1}^K   \sum_{a_L , \tau_{L-1} }\pi^k(\tau_L) \times \big\|\mathbf{A}^k_{L-1:1} \big(\mathbf{a}^k(a_1)-\mathbf{a}(a_1)\big)\big\|_1.
\end{aligned}
\end{equation}
where we have defined $\mathbf{a}_l = \mathbf{A}_{l-1:1}\mathbf{a}(a_1)$ and used the telescoping identity
\begin{align}
    \mathbf{A}^k_{L-1:1} - \mathbf{A}_{L-1:1} = \sum_{l=1}^{L-1} \mathbf{A}^k_{L-1:l+1} \big(\mathbf{A}^k_{l}-\mathbf{A}_{l}\big) \mathbf{A}_{l-1:1},
\end{align}
By applying Lemma~\ref{lem:one_norm_op}, i.e. $\mathbf{A}_{L-1:l+1}^{k}$ is contractive up to a constant 
\begin{align}
    \sum_{a_L,\tau_{L-1:l+2},o_{l+1}} \big\|\mathbf{A}^k_{L-1:l+1}\big(\mathbf{A}^k_{l}-\mathbf{A}_{l}\big)\mathbf{a}_{l}\big\|_1 \times \prod_{l'=l+1}^{L-1} \pi(a_{l'+1}|\tau_{l'}) & \leq S^2\kappa_{\mathrm{uc}} \big\|\big(\mathbf{A}^k_{l}-\mathbf{A}_{l}\big)\mathbf{a}_{l}\big\|_1, \\
    \sum_{a_L,\tau_{L-1:2},o_{1}} \big\|\mathbf{A}^k_{L-1:1} \big(\mathbf{a}^k(a_1)-\mathbf{a}(a_1)\big)\big\|_1  \times \prod_{l'=1}^{L-1} \pi(a_{l'+1}|\tau_{l'}) & \leq S^2\kappa_{\mathrm{uc}} \big\|\mathbf{a}^k(a_1)-\mathbf{a}(a_1)\big\|_1,
\end{align}
we obtain the result 
\begin{equation}
\begin{aligned}
    \regret(K) & \leq
    LR  S^2\kappa_{\mathrm{uc}} \sum_{l=1}^{L-1}  \sum_{k=1}^K  \sum_{ a_{l+1} , \tau_{l} }  \pi^k(\tau_{l+1}) \big\|\big(\mathbf{A}^k_{l}-\mathbf{A}_{l}\big)\mathbf{a}_{l}\big\|_1 \\
    & \quad+ LR  S^2\kappa_{\mathrm{uc}}  \sum_{k=1}^K   \sum_{a_1} \pi^k(a_1) \times \big\| \big(\mathbf{a}^k(a_1)-\mathbf{a}(a_1)\big)\big\|_1.
\end{aligned}
\end{equation}
This concludes our proof. 
\end{proof}

\subsection{The eluder dimension for linear functions}

Using Proposition~\ref{prop:trajectory_mle} (recall we pick from the previous confidence region, so $\omega_k\in \mathcal{C}^t$ for any $t\leq k$) we have
\begin{align}
    \sum_{t=1}^{k-1}\Big( \sum_{\tau_L} \big|\mathbb{P}^{\pi^t}_{\omega_k}(\tau_L ) - \mathbb{P}^{\pi^t}_{\omega}(\tau_L ) \big|\Big)^2 = \mathcal{O}(\beta),
\end{align}
and by Cauchy-Schwarz inequality,
\begin{align}
     \sum_{t=1}^{k-1}\sum_{\tau_L} \big| \mathbb{P}^{\pi^{t}}_{\omega_k}(\tau_L ) - \mathbb{P}^{\pi^{t}}_{\omega}(\tau_L ) \big| = \mathcal{O}(\sqrt{\beta k} ).  
\end{align}
Then similarly as in Lemma~\ref{lem:regret_step1} where we sum over $o_l$, we introduce OOMs (Definition~\ref{def:oom_qhmm}) to get
\begin{equation}
\label{eq:guarantee}
\begin{aligned}
    \sum_{t=1}^{k-1}\sum_{a_l ,\tau_{l-1}} \pi^{t} (\tau_l ) &\big\| \mathbf{A}^k_{l-1}\cdots \mathbf{A}^k_1\mathbf{a}^k(a_1 ) -   \mathbf{A}_{l-1}\cdots \mathbf{A}_1\mathbf{a}(a_1 ) \big\|_1 = \mathcal{O}(\sqrt{\beta k} ),
\end{aligned}
\end{equation}
which, in particular, yields for the first action:
\begin{align}
    \label{eq:initial_estimate_error}
    \sum_{t=1}^{k-1} \sum_{a_1} \pi^{t} (a_1 ) \|\mathbf{a}^k (a_1) -  \mathbf{a} (a_1)\|_1 = \mathcal{O} (\sqrt{\beta k}),
\end{align}
where $\pi^{t} (a_1)$ is the first action chosen by the policy without prior information.

A natural question now arises: how do we relate the high-probability guarantees of the MLE to the cumulative regret established in Lemma~\ref{lem:regret_step1}? In classical analyses of tabular POMDPs \cite{liu2022partially}, it is standard to use a discrete $\ell_1$ eluder lemma. However, applying it directly introduces a logarithmic dependence on the number of trajectory prefixes $n=|\mathcal{A}\times[O]|^{l-1}$ at step $l$. This injects an artificial $\mathcal{O}(l^2)$ penalty into the step-wise error, which, when accumulated over $l \in [L-1]$, yields a suboptimal $\mathcal{O}(L^4)$ overall regret. 

To overcome this sequence-length penalty and achieve the optimal $L^2$ scaling, we unify our approach with the continuous action setting that we will study later. We rely on the general eluder lemma for linear functions over measure spaces~\cite[Corollary G.2]{liu2023optimistic}, treating our discrete trajectory sets as measure spaces equipped with the counting measure.

\begin{proposition}[General Eluder lemma for linear functions~{\cite[Corollary G.2]{liu2023optimistic}}]
    \label{prop:eluder_general}
    Let $(\mathcal{X},\mu_{\mathcal{X}})$ and $(\mathcal{Y},\mu_{\mathcal{Y}})$ be measure spaces and let $\{f^{k}\}_{k\in[K]}$ and $\{g^{k}\}_{k\in[K]}$ be functions $f^{k}:\mathcal{X}\rightarrow\mathbb{R}^{d}$ and $g^{k}:\mathcal{Y}\rightarrow\mathbb{R}^{d}$. Suppose that for all $k\in[K]$,
    \begin{equation}
        \sum_{t=1}^{k-1}\int_{\mathcal{X}}\int_{\mathcal{Y}}|\langle f^{k}(x),g^{t}(y)\rangle|\mu_{\mathcal{X}}(dx)\mu_{\mathcal{Y}}(dy)\le \zeta_{k}.
    \end{equation}
    Assume 
    \begin{align}
        & \sup_{x\in\mathcal{X}}\|f^{k}(x)\|_{\infty} \le R_{x}, \\
        & \int_{\mathcal{Y}}\|g^{t}(y)\|_{1}\mu_{\mathcal{Y}}(dy) \le R_{y}. 
    \end{align}
    Then for all $k\in[K]$,
    \begin{equation}
        \sum_{t=1}^{k}\int_{\mathcal{X}}\int_{\mathcal{Y}}|\langle f^{t}(x),g^{t}(y)\rangle|\mu_{\mathcal{X}}(dx)\mu_{\mathcal{Y}}(dy)=\mathcal{O}\Big(d\log^{2}(K)(R_{x}R_{y}+\max_{t\le k}\zeta_{t})\Big).
    \end{equation}
\end{proposition}

\subsubsection{First action}

One key difference from the classical POMDP setting of~\cite{liu2022partially,liu2023optimistic} is how the interaction starts. In the classical model, the agent receives an initial observation ``for free'', so the first decision is effectively conditioned on that observation and the initial information acquisition does not incur regret. In our quantum setting, there is no free initial observation: obtaining any observation (including at the first step) requires applying an action/measurement. Therefore, the first action is chosen without conditioning on past observations and, crucially, it directly affects the obtained reward and hence contributes non-trivially to the regret. As a consequence, the first action cannot be treated as an arbitrary or random initialization, and we isolate its contribution in a separate lemma.

\begin{lemma}
\label{lem:oom_first_action_dicrete}
Under the same assumptions of Lemma~\ref{lem:regret_step1} we have that the contribution to the regret of the first action is bounded as
\begin{align}
    \sum_{k=1}^K \sum_{a_1} \pi^k(a_1)\|\mathbf{a}^k(a_1)-\mathbf{a}(a_1)\|_1 = \widetilde{\mathcal{O}}\!\Big( A\sqrt{\beta K}\Big).
    \label{eq:boundary_diagonal}
\end{align}
\end{lemma}

\begin{proof}
We apply Proposition~\ref{prop:eluder_general}. Let $\mathcal{X}\coloneqq\{*\}$ be a singleton and $\mathcal{Y}\coloneqq\mathcal{A}$, both equipped with the counting measure. We set the target dimension $d=A$. Define $f^{k}(*) \in \mathbb{R}^{A}$ and $g^{t}(a_{1}) \in \mathbb{R}^{A}$ as:
\begin{align}
    f^{k}(*)_a &:= \|\mathbf{a}^k(a)-\mathbf{a}(a)\|_1, \\
    g^{t}(a_{1})_a &:= \pi^{t}(a_{1})\mathbbm{1}(a_{1}=a).
\end{align}
The inner product yields
\begin{align}
    \langle f^{k}(*),g^{t}(a_{1})\rangle = \pi^{t}(a_{1})\|\mathbf{a}^k(a_{1})-\mathbf{a}(a_{1})\|_1. 
\end{align}
Integrating over $\mathcal{Y}$ reproduces the empirical error sum exactly. We verify the norm bounds: for $f^k$, since $\mathbf{a}^k(a)$ and $\mathbf{a}(a)$ are probability vectors over observations,
\begin{align}
    \|\mathbf{a}^k(a)-\mathbf{a}(a)\|_1 \le 2 \implies R_{x}=2. 
\end{align}
For $g^{t}$, 
\begin{align}
    \sum_{a_{1}}\|g^{t}(a_{1})\|_1 = \sum_{a_{1}}\pi^{t}(a_{1}) = 1 \implies R_{y}=1. 
\end{align}
From~\eqref{eq:initial_estimate_error}, 
\begin{align}
    \sum_{t=1}^{k-1} \sum_{a_1} \pi^{t} (a_1 ) \|\mathbf{a}^k (a_1) -  \mathbf{a} (a_1)\|_1 = \mathcal{O} (\sqrt{\beta k}) \implies \zeta_{k}=\mathcal{O}(\sqrt{\beta k}). 
\end{align}
Applying Proposition~\ref{prop:eluder_general} directly gives the desired bound
\begin{align}
    \sum_{k=1}^K \sum_{a_1} \pi^k(a_1)\|\mathbf{a}^k(a_1)-\mathbf{a}(a_1)\|_1\leq \mathcal{O}\Big(A \log^{2}(K)(2 + \sqrt{\beta K})\Big) = \widetilde{\mathcal{O}}\!\Big( A\sqrt{\beta K}\Big). 
\end{align}
\end{proof}

\subsubsection{OOM estimation error}

We now apply the same measure-theoretic eluder decomposition to the sequence of discrete OOM operators to effectively decouple the covering dimension from the sequence depth.

\begin{lemma}[OOM estimation error contribution]
\label{lem:oom_error_discrete}
Under the same assumptions of Lemma~\ref{lem:regret_step1} we have that the contribution to the regret of the error of the OOM operators is bounded uniformly for each $l \in [L-1]$ as
\begin{equation}
    \sum_{k=1}^K 
    \sum_{a_{l+1}, \tau_{l} } \pi^k(\tau_{l+1})    
    \times 
    \big\|
        \big(\mathbf{A}^k_{l}-\mathbf{A}_{l}\big)
        \mathbf{a}_l 
    \big\|_1 
    = 
    \widetilde{\mathcal{O}}\!\big(A^2 O^2 S^2 \kappa_{\mathrm{uc}}\sqrt{K\beta} \big)
\end{equation}
\end{lemma}

\begin{proof}
From the guarantee~\eqref{eq:guarantee} together with Lemma~\ref{lem:one_norm_op}, applying triangle inequalities and the telescoping expansion similar to~\eqref{eq:telescoping} yields for any step $l \in [L-1]$:
\begin{equation}
\begin{aligned}
\label{eq:promise}
     & \sum_{t=1}^{k-1}\sum_{a_{l+1} , \tau_{l}} \pi^{t} (\tau_{l+1} )\times \big\|(\mathbf{A}^k_l - \mathbf{A}_l)  \mathbf{a}_l \big\|_1 \\
     & \leq \sum_{t=1}^{k-1}\sum_{a_{l+1} , \tau_{l}} \pi^{t} (\tau_{l+1} )\times \Big(\big\|\mathbf{A}^k_l(\mathbf{a}^{k}_l- \mathbf{a}_l)\big\|_1  + \big\|\mathbf{A}^k_l  \mathbf{a}^k_l - \mathbf{A}_l  \mathbf{a}_l \big\|_1  \Big) \\
     & = \mathcal{O}\big(  S^2 \kappa_{\mathrm{uc}}\sqrt{k\beta } \big)+ \mathcal{O}(\sqrt{k\beta}) = \mathcal{O}\big(  S^2 \kappa_{\mathrm{uc}}\sqrt{k\beta } \big),
\end{aligned}
\end{equation}
where we applied Lemma~\ref{lem:one_norm_op} in the last line. So we set 
\begin{align}
    \zeta_k = \mathcal{O}(  S^2 \kappa_{\mathrm{uc}}\sqrt{k\beta } ). 
\end{align}

Notice that $\mathbf{a}_l$ and $\mathbf{A}_l \mathbf{a}_l$ are both  PSD diagonal matrices; we can view them as probability vectors by choosing the basis $\lbrace\proj{o}\rbrace$. The operators $\mathbf{A}_l$ map PSD diagonal matrices to PSD diagonal matrices, meaning we can view them as transmission matrices. The $1$-norm, $2$-norm and $\infty$-norm of diagonal matrices are exactly the same as those of vectors. More specifically, we write the probability vector and the transmission matrices as
\begin{align}
    \mathbf{a}_l = \sum_{o'} [\mathbf{a}_l]_{o'} e_{o'}, \qquad 
    \mathbf{A}_l = \sum_{j,o'} [\mathbf{A}_l]_{j,o'} e_{j,o'},
\end{align}
where $e_{o'}= \proj{o'}$ and $e_{j,o'}(\proj{o'}) = \proj{j}$.

To invoke Proposition~\ref{prop:eluder_general} without paying an $\mathcal{O}(l^2)$ penalty from trajectory branching, we define a continuous-style eluder space for each fixed action-outcome tuple $(o_l, a_l, a_{l+1})$ with the counting measure. Let $\mathcal{X} \coloneqq [O]$ index the rows of the transmission matrix, and let $\mathcal{Y}$ be the set of trajectory prefixes $\tau_{l-1}$, both equipped with the counting measure. We set the target dimension $d=O$. Define the parameter error function $f^{k}: \mathcal{X} \rightarrow \mathbb{R}^O$ and the data feature function $g^{t}: \mathcal{Y} \rightarrow \mathbb{R}^O$ as:
\begin{align}
    f^{k}(j) &\coloneqq \left( [ \mathbf{A}^k(o_{l},a_{l},a_{l+1} ) - \mathbf{A}(o_{l},a_{l},a_{l+1} ) ]_{j, \cdot} \right)^\top \in \mathbb{R}^O, \\
    g^{t}(\tau_{l-1}) &\coloneqq \pi^{t}(\tau_{l+1}) \mathbf{a}(\tau_{l-1},a_{l}) \in \mathbb{R}^O.
\end{align}
The inner product isolates the prediction error for the specific outcome $j$:
\begin{equation}
    \langle f^{k}(j), g^{t}(\tau_{l-1}) \rangle = \pi^{t}(\tau_{l+1}) \big[ (\mathbf{A}^k(o_{l},a_{l},a_{l+1} ) - \mathbf{A}(o_{l},a_{l},a_{l+1} )) \mathbf{a}(\tau_{l-1},a_{l}) \big]_j.
\end{equation}
Integrating the absolute value over $\mathcal{X}$ (summing over $j$) exactly reproduces the $\ell_1$ norm of the prediction error vector. Integrating further over $\mathcal{Y}$ (summing over $\tau_{l-1}$) yields:
\begin{equation}
    \int_{\mathcal{X}} \int_{\mathcal{Y}} |\langle f^{k}(j), g^{t}(\tau_{l-1}) \rangle| = \sum_{\tau_{l-1}} \pi^{t}(\tau_{l+1}) \| (\mathbf{A}^k(o_{l},a_{l},a_{l+1} ) - \mathbf{A}(o_{l},a_{l},a_{l+1} )) \mathbf{a}(\tau_{l-1},a_{l}) \|_1.
\end{equation}
Summing this over all $t$ and all fixed tuples $(o_l, a_l, a_{l+1})$ is bounded by $\zeta_k$, so the precondition holds for each tuple independently.

We now verify the norm bounds for Proposition~\ref{prop:eluder_general}:

\textbf{Step 1: Norm condition for $f^{k}$ ($R_{x}$)}.
The $l_\infty$ norm of $f^k(j)$ is the maximum absolute entry in the $j$-th row of the operator error. The sum of absolute entries in a row is bounded by the induced $1\to 1$ norm of the matrix. Because $\mathbf{A}_{l}$ and $\mathbf{A}_{l}^{k}$ send positive operators to positive operators, we have:
\begin{align}
    \big\| \tilde{\mathbf{A}}_l\big\|_{1\to 1} &= \big\|\Tr_{S}\circ \mathcal{P}^{(a_{l+1})}\circ\big(\mathbb{E}_l \otimes \mathcal{T}_{o_l}\big)\circ R^{(a_l)}\big\|_{1\to 1} \\
    &\leq \|\Tr_{S}\|_{1\to1} \| \mathcal{P}^{(a_{l+1})}\|_{1\to 1} \|\big(\mathbb{E}_l\otimes \mathcal{T}_{o_l}\big)\|_{1\to 1}  \| R^{(a_l)}\|_{1\to 1} \\
    &\leq \max_{x:\|x\|_1\leq 1}\| R^{(a_l)} x\|_{1} = \| R^{(a_l)}\|_{1\to 1, \diag} \le \kappa_{\mathrm{uc}}.
\end{align}
By the triangle inequality
\begin{align}
    \|\mathbf{A}^k_l - \mathbf{A}_l\|_{1 \to 1} \leq \|\mathbf{A}^k_l\|_{1\to 1} - \|\mathbf{A}_l\|_{1 \to 1} \leq 2\kappa_{\mathrm{uc}}
\end{align}
meaning any individual entry is strictly bounded by $2\kappa_{\mathrm{uc}}$. Thus,
\begin{align}
    R_{x} = 2\kappa_{\mathrm{uc}}. 
\end{align}

\textbf{Step 2: Norm condition for $g^{t}$ ($R_{y}$)}.
We sum the $l_{1}$ norm of $g^{t}$ over $\mathcal{Y}$:
\begin{equation}
    \int_{\mathcal{Y}} \|g^{t}(\tau_{l-1})\|_{1} d\mu_{\mathcal{Y}} = \sum_{\tau_{l-1}} \pi^{t}(\tau_{l+1}) \|\mathbf{a}(\tau_{l-1},a_{l})\|_1.
\end{equation}
Since $\mathbf{a}(\tau_{l-1},a_l)$ is a classical observable state (a PSD diagonal matrix), its $1$-norm equals its trace. Because the initial memory state and channels are trace-preserving, the trace of the filtered unnormalized state is exactly the probability of observing the prefix sequence: 
\begin{align}
    \|\mathbf{a}(\tau_{l-1},a_{l})\|_1 = \Tr(\tilde{\rho}_{l}(\tau_{l-1})) = A_{\omega}(\tau_{l-1}). 
\end{align}
Substituting this back:
\begin{equation}
    \int_{\mathcal{Y}} \|g^{t}(\tau_{l-1})\|_{1} d\mu_{\mathcal{Y}} \le \sum_{\tau_{l-1}} \pi^{t}(\tau_{l+1}) A_{\omega}(\tau_{l-1}) \le \sum_{\tau_{l-1}} \pi^{t}(\tau_{l-1}) A_{\omega}(\tau_{l-1}) = \sum_{\tau_{l-1}} \mathbb{P}^{\pi^t}_{\omega}(\tau_{l-1}) = 1.
\end{equation}
Thus, 
\begin{align}
    R_y = 1. 
\end{align}

\textbf{Step 3: Conclusion}.
For each fixed tuple $(o_l, a_l, a_{l+1})$, we invoke Proposition~\ref{prop:eluder_general} with dimension $d=O$, $R_{x}=2\kappa_{\mathrm{uc}}$, $R_{y}=1$, and the per-tuple error share of $\zeta_{K}=\mathcal{O}\big(S^2\kappa_{\mathrm{uc}}\sqrt{\beta K}\big)$. This yields a bound of $\mathcal{O}\big(O \log^2(K) (2\kappa_{\mathrm{uc}} + \zeta_K)\big)$. Summing this independent bound over all $A^2 O$ possible tuples of $(o_l, a_l, a_{l+1})$ gives:
\begin{equation}
\begin{aligned}
    \sum_{k=1}^{K} \sum_{a_{l+1},\tau_l} \pi^k(\tau_{l+1}) \|(\mathbf{A}_l^k- \mathbf{A}_l )\mathbf{a}_l\|_1 &= \mathcal{O}\big(A^2 O \cdot O\log^{2}(K)(2\kappa_{\mathrm{uc}} + \mathcal{O}(S^2\kappa_{\mathrm{uc}}\sqrt{\beta K}))\big) \\
    &= \widetilde{\mathcal{O}}\big(A^{2}O^2S^2\kappa_{\mathrm{uc}}\sqrt{\beta K}\big).
\end{aligned}
\end{equation}
This bounds the error completely independently of the step index $l$, strictly eliminating the sequence-length penalty.
\end{proof}

\subsection{Cumulative regret bound for discrete actions}

By integrating this sequence-independent local error back into our global accumulation, we eliminate the compounding $l^2$ trajectory branching penalty to achieve an explicit $L^2$ scaling rate. 

\begin{theorem}[Regret bound for discrete actions]\label{thm:discrete_regret}
Fix $\delta\in(0,1)$ and assume the rewards are uniformly bounded $|r_l(a_l, o_l)| \le R$ for all $l \in [L], a_l \in \mathcal{A}$, and $o_l \in [O]$. Let the true environment satisfy $\omega\in\Omega_{\mathrm{uc}}$ (Definition~\ref{def:Omega_uc}) with robustness constant $\kappa_{\mathrm{uc}}$, and let $A\coloneqq|\mathcal{A}|$. Run Algorithm~\ref{alg:omle-qhmm} with confidence radii $\beta$ chosen as:
\begin{equation}
    \label{eq:beta_def_thm}
    \beta \coloneqq c\Big((L + A O)S^4 \log(KSAOL) + \log(K/\delta)\Big),
\end{equation}
where $c>0$ is the universal constant from the MLE bounds (Proposition~\ref{prop:concentration_mle} and~\ref{prop:trajectory_mle}). Then, with probability at least $1-\delta$,
\begin{equation}
    \label{eq:regret_final_beta}
    \regret(K) \le \widetilde{\mathcal{O}}\Big( R S^4 \kappa_{\mathrm{uc}}^{2} A^{2} O^{2} L^{2} \sqrt{K\beta} \Big).
\end{equation}
Substituting~\eqref{eq:beta_def_thm} into~\eqref{eq:regret_final_beta} and absorbing polylogarithmic factors into $\widetilde{\mathcal{O}}(\cdot)$ yields:
\begin{equation}
    \label{eq:regret_final_simplified}
    \regret(K) \le \widetilde{\mathcal{O}}\Big( R S^6 \kappa_{\mathrm{uc}}^{2} A^{2} O^{2} L^{2} \sqrt{K(L + A O)} \Big).
\end{equation}
\end{theorem}

\begin{proof}
Let $\mathcal{E}$ be the (high-probability) event on which the likelihood-ratio and trajectory $\ell_1$ concentration bounds of Proposition~\ref{prop:concentration_mle} and~\ref{prop:trajectory_mle} hold simultaneously for all $t\in[K]$ with radius $\beta$ in~\eqref{eq:beta_def_thm}. On $\mathcal{E}$, we have $\omega\in\mathcal{C}_k$ for all $k\in[K]$, so the optimism argument applies.

By Lemma~\ref{lem:regret_step1}, on $\mathcal{E}$:
\begin{align}
    \regret(K) &\le S^2\kappa_{\mathrm{uc}}LR  \sum_{l=1}^{L-1}\sum_{k=1}^K\sum_{a_{l+1},\tau_l}\pi^k(\tau_{l+1}) \big\|\big(\mathbf{A}^k_l-\mathbf{A}_l\big) \mathbf{a}_l\big\|_1\\
    &\quad+ S^2\kappa_{\mathrm{uc}}LR \sum_{k=1}^K\sum_{a_1}\pi^k(a_1)\|\mathbf{a}^k(a_1)-\mathbf{a}(a_1)\|_1 . \label{eq:regret_decomp_thm}
\end{align}
The second term in~\eqref{eq:regret_decomp_thm} is bounded by the first-action estimate (Equation~\eqref{eq:boundary_diagonal}):
\begin{equation}
    \sum_{k=1}^K\sum_{a_1}\pi^k(a_1)\|\mathbf{a}^k(a_1)-\mathbf{a}(a_1)\|_1 = \widetilde{\mathcal{O}}\!\left(A\bigl(2+\sqrt{\beta K}\bigr)\right) \le \widetilde{\mathcal{O}}\!\left(A\sqrt{\beta K}\right).
\end{equation}
For the first term in~\eqref{eq:regret_decomp_thm}, the OOM estimation-error bound proved in Lemma~\ref{lem:oom_error_discrete} gives, uniformly over $l\in[L-1]$ and after summing over the relevant $(o_l,a_l,a_{l+1})$ indices:
\begin{equation}
    \sum_{k=1}^K \sum_{a_{l+1},\tau_l}\pi^k(\tau_{l+1}) \big\|\big(\mathbf{A}^k_l-\mathbf{A}_l\big) \mathbf{a}_l\big\|_1 = \widetilde{\mathcal{O}}\big(A^2O^2 S^2\kappa_{\mathrm{uc}}\sqrt{K\beta}\big).
\end{equation}
Summing this constant uniform bound over the $L-1$ steps yields a strict scalar factor of $\mathcal{O}(L)$. Multiplying by the prefactor $S^{\frac{5}{2}}\kappa_{\mathrm{uc}}LR$ in~\eqref{eq:regret_decomp_thm} yields:
\begin{equation}
    \regret(K) \le \widetilde{\mathcal{O}}\!\Big( R S^4 \kappa_{\mathrm{uc}}^{2} A^{2} O^{2} L^{2} \sqrt{K\beta} \Big),
\end{equation}
which is exactly~\eqref{eq:regret_final_beta}. Substituting $\beta$ from~\eqref{eq:beta_def_thm} and absorbing the logarithmic factors gives~\eqref{eq:regret_final_simplified}.
\end{proof}

\section{Extension to continuous action spaces}
\label{sec:cont-actions}

In the preceding sections, we restricted our analysis to finite action sets. We now extend the framework to accommodate continuous (or arbitrarily large) action spaces. To ensure that the regret remains tractable, we will formulate a finite-dimensional embedding condition that allows us to seamlessly apply the previously developed proof strategies to the continuous domain.

\subsection{Continuous action setting}
\label{sec:continous_model}

Throughout this section, we restrict to L\"uders instruments induced by $O$-outcome POVMs $\big\lbrace M_o^{(a)} \big\rbrace_{x\in[O]}$ for all $a\in\mathcal{A}$ on the hidden memory.

\paragraph{Measurable action space and stochastic-kernel policies.}
Let $(\mathcal A,\mathcal{B}(\mathcal A))$ be a measurable action space and fix a $\sigma$-finite reference measure
$\mu$ on $\mathcal A$ (typically a probability measure). A policy $\pi=\{\pi_l\}_{l=1}^L$ is a family of stochastic kernels:
for each history $\tau_{l-1}=(a_1,o_1,\ldots,a_{l-1},o_{l-1})\in (\mathcal{A}\times[O])^{l-1}$, $\pi_l(\cdot | \tau_{l-1})$ is a probability measure on
$(\mathcal A,\mathcal{B}(\mathcal A))$ w.r.t. $\mu$ on the action coordinates. When convenient, we assume $\pi_l(\cdot| \tau_{l-1})\ll \mu$ and denote its Radon-Nikodym derivative by the same symbol:
\begin{equation}
    \pi_l(\dd a | \tau_{l-1}) = \pi_l(a | \tau_{l-1}) \mu(\dd a), \qquad  \int_{\mathcal A} \pi_l(a | \tau_{l-1})\mu(\dd a)=1 .
\label{eq:policy-density}
\end{equation}
In direct analogy with the discrete shorthand~\eqref{eq:policy_shorthand}, we define the \emph{cumulative policy density product}
\begin{equation}
    \pi(\tau_l) \coloneqq \prod_{l'=1}^{l}\pi_{l'}(a_{l'} | \tau_{l'-1}),
\label{eq:pi-tau-cont}
\end{equation}
which is the density (w.r.t.\ $\mu^{\otimes l}$ on the action coordinates) associated with sampling $a_1,\ldots,a_l$ sequentially from the kernels $\{\pi_{l'}(\cdot| \tau_{l'-1})\}_{l'\in[l]}$.

For a fixed model $\omega$ (analogous to Definition~\eqref{def:qrl_model}), define the discrete conditional probabilities $\Pr(o_l|\tau_{l-1}, a_l)$ similar to~\eqref{eq:conditional_prob_compact}, and the \emph{cumulative observation likelihood product}
\begin{align}
    \label{eq:A-tau-cont}
    A_\omega(\tau_l)\coloneqq
    \prod_{l'=1}^{l}\Pr(o_{l'} | \tau_{l'-1},a_{l'}). 
\end{align}

Let $\nu_l\coloneqq \mu^{\otimes l}\otimes \mathrm{count}^{\otimes l}$ be the product reference measure on $(\mathcal{A}\times[O])^l$ (counting measure on outcomes). The induced trajectory law under $(\pi,\omega)$ is the
probability measure $\mathbb{P}_{\omega}^{\pi}$ on $(\mathcal{A}\times[O])^l$ defined by the chain rule. Under~\eqref{eq:policy-density}, it admits a density $p_{\omega}^{\pi}$ w.r.t.\ $\nu_l$ as
\begin{equation}
    p_{\omega}^{\pi}(\tau_l)\coloneqq  \frac{\dd\mathbb{P}_{\omega}^{\pi}}{\dd\nu_l}(\tau_l) =  \pi(\tau_l) A_\omega(\tau_l).
    \label{eq:traj-density}
\end{equation}
Equivalently, for any measurable $B\subseteq (\mathcal{A}\times[O])^l$,
\begin{equation}
    \mathbb{P}_{\omega}^{\pi}(B) =  \int_B p_{\omega}^{\pi}(\tau_l)\nu_l(\dd\tau_l) = \sum_{o_{1:l}}\int_{\mathcal A^l} \mathbbm{1}\{(a_{1:l},o_{1:l})\in B\} \prod_{l'=1}^{l}\Big(\pi(a_{l'}|\tau_{l'-1})\Pr(o_{l'}|\tau_{l'-1},a_{l'})\Big) \mu^{\otimes l}(\dd a_{1:l}).
\label{eq:traj-law}
\end{equation}

Finally, whenever $P$ and $Q$ are probability measures on $(\mathcal{A}\times[O])^L$ that are absolutely continuous w.r.t.\ $\nu_L$, we write 
\begin{equation}
    \|P-Q\|_1 \coloneqq  \int_{(\mathcal{A}\times[O])^L}\left|\frac{\dd P}{\dd\nu_L}(\tau_L)-\frac{\dd Q}{\dd\nu_L}(\tau_L)\right|\nu_L(\dd\tau_L),
\label{eq:l1-tv-cont}
\end{equation}
i.e. twice the total-variation distance.

\subsection{Embedding (spanning) dimension for continuous actions}
\label{sec:embedding-dim}
Let $\mathrm{Herm}(\mathcal H_S)$ be the real vector space of Hermitian operators on $\mathcal H_S$, which has real dimension $S^2$. The ambient space of $O$-tuples $(\mathrm{Herm}(\mathcal H_S))^{O}$ has real dimension $OS^2$. However, any valid action $a \in \mathcal{A}$ is a POVM, meaning its effects must satisfy the completeness constraint $\sum_{o\in[O]}M_o^{(a)}=\mathbb{I}_S$. This imposes $S^2$ independent real linear constraints. Consequently, the set of all possible $O$-outcome POVMs forms an \emph{affine} subspace of dimension
\begin{equation}
    d_{\mathrm{POVM}} = (O-1)S^2.
    \label{eq:d-povm}
\end{equation}
We fix a Hilbert-Schmidt orthonormal basis $\{P_\mu\}_{\mu=1}^{S^2}\subset \mathbb{M}_S$ satisfying
\begin{align}\label{eq:hs_basis}
     P_\mu^\dagger = P_\mu, \qquad  \Tr(P_\mu P_\nu)=\delta_{\mu\nu}, \qquad \forall \mu,\nu\in [S^2].
\end{align}
An example is the generalized Pauli basis. For each $a\in\mathcal{A}$ and
$o\in[O]$, define the coefficients
\begin{align}
    \phi_\mu(a,o) \coloneqq \Tr\!\big[P_\mu M_o^{(a)}\big],\qquad \forall\mu\in[S^2].
\end{align}
Then
\begin{align}
    M_o^{(a)} = \sum_{\mu=1}^{S^2}\phi_\mu(a,o) P_\mu, \qquad \forall a\in\mathcal{A}, o\in[O].
\end{align}
The completeness constraint $\sum_{o}M_a^{(o)}=\mathbb{I}_S$ is equivalent to the $S^2$ linear identities
\begin{align}
    \sum_{o\in[O]}\phi_\mu(a,o) = \Tr(P_\mu), \qquad \forall \mu\in[S^2], 
\end{align}
which is another way to see~\eqref{eq:d-povm}.

Because the actions lie in an affine space rather than a vector subspace passing through the origin, taking the direct linear span of the POVM elements introduces artificial dimension artifacts. We avoid this by defining the spanning dimension directly as the dimension of the action set's \emph{affine hull}. Mathematically, this is equivalent to the linear dimension of the direction space formed by the differences between actions.

\begin{definition}[Spanning dimension]
\label{def:embedding-dim}
Given an action set $\mathcal A$ consisting of $O$-outcome POVMs, define its direction space as the linear span of the differences between its elements:
\begin{align}
    \mathsf{V}_{\mathcal{A}} \coloneqq  \mathrm{span}_{\mathbb{R}}\Big\{\big(M_o^{(a)} - M_o^{(a')}\big)_{o\in[O]} \Big| a, a'\in\mathcal{A}\Big\}.
\end{align}
The \emph{spanning dimension} of the action set is the real linear dimension of this direction space:
\begin{align}
    d_{\mathcal A} \coloneqq \dim_{\mathbb R}\big(\mathsf{V}_{\mathcal{A}}\big).
\end{align}
We say that $\mathcal{A}$ is \emph{fully POVM-spanning} if $d_{\mathcal A}=(O-1)S^2$.
\end{definition}

Notice that for any difference tuple in $\mathsf{V}_{\mathcal{A}}$, the sum over outcomes strictly vanishes
\begin{align}
    \sum_o \big(M_o^{(a)} - M_o^{(a')}\big)= 0. 
\end{align}
This naturally forms a true vector subspace that is completely independent of the identity offset $\mathbb{I}_S$. By construction, $d_{\mathcal A}\le (O-1)S^2$, with equality precisely when $\mathcal{A}$ explores the maximum allowable affine degrees of freedom. When $\mathcal{A}$ is continuous, $d_{\mathcal A}$ is the effective complexity parameter that will replace explicit dependence on $|\mathcal{A}|$ in the regret bounds.

\begin{remark}
Although we introduced $d_{\mathcal A}$ with continuous (measurable) action classes in mind, the same notion applies verbatim to \emph{finite} action sets. In particular, if $\mathcal{A}$ is a large discrete collection of POVMs (or
instruments) whose effects are fully POVM-spanning, then $d_{\mathcal{A}}=(O-1)S^2$ is independent of $|\mathcal{A}|$. Consequently, the subsequent sections about concentration and regret bounds can be stated in terms of the embedding
dimension $d_{\mathcal A}$ rather than the action-set cardinality, which is most beneficial in regimes where $d_{\mathcal A}\ll|\mathcal A|$.
\end{remark}

\begin{remark}[Qubit and Bloch-sphere example]
Consider a two-dimensional hidden memory ($S=2$) and binary outcomes ($O=2$). Let $\bm{\sigma}=(\sigma_x, \sigma_y, \sigma_z)$ denote the vector of Pauli matrices. In the unnormalized Pauli basis $\{\mathbb{I}, \sigma_x, \sigma_y, \sigma_z\}$, the effects of any POVM can be parameterized as
\begin{equation*}
    M_1^{(\alpha,r)} = \frac{1}{2}(\alpha \mathbb{I} + \mathbf{r }\cdot \bm{\sigma}), \qquad M_2^{(\alpha,r)} = \mathbb{I}  - M_0^{(a)},
\end{equation*}
where $\alpha \in \mathbb{R}$ is a scalar proportional to the trace and $\mathbf{r} \in \mathbb{R}^3$ is the generalized Bloch vector. The difference between any two actions $(M_1^{(\alpha,\mathbf{r})},M_2^{(\alpha,\mathbf{r})})$ and $(M_1^{(\alpha',\mathbf{r}')},M_2^{(\alpha',\mathbf{r}')})$ is therefore entirely characterized by the first effect
\begin{align}
    (\Delta M_1,\Delta M_2 )= \frac{1}{2}(\Delta \alpha \mathbb{I} + \Delta \mathbf{r} \cdot \bm{\sigma}, -\Delta \alpha \mathbb{I}-  \Delta \mathbf{r} \cdot \bm{\sigma} )
\end{align}
where $(\Delta M_1,\Delta M_2 ) = (M_1^{(\alpha,r)}-M_1^{(\alpha',r')},M_2^{(\alpha,r)}-M_2^{(\alpha',r')} )$, $\Delta \alpha = \alpha-\alpha'$ and $\Delta \mathbf{r} = \mathbf{r} - \mathbf{r}'$.

\begin{itemize}
    \item \textbf{Case (i): Projective measurements.} For projective measurements along a direction $\mathbf{n} \in \mathbb{S}^2$, the effects take the form
    \begin{align}
        M_1^{(\mathbf{n})} = \frac{1}{2}(\mathbb{I} + \mathbf{n}\cdot \bm{\sigma}). 
    \end{align}
    In this case, $\alpha=1$ is strictly fixed (since $\Tr(M_1^{(\mathbf{n})})=1$). Therefore, the difference between any two projective measurements takes the form
    \begin{align}
        (\Delta M_1,\Delta M_2 )= \frac{1}{2}(\Delta \mathbf{n} \cdot \bm{\sigma}, - \Delta \mathbf{n} \cdot \bm{\sigma}). 
    \end{align}
    Because $\Delta \alpha = 0$, the direction space is restricted strictly to the Pauli directions:
    \begin{equation}
        \mathsf{V}_{\mathcal{A}} \cong \text{span}_{\mathbb{R}}\{\sigma_x, \sigma_y, \sigma_z\},
    \end{equation}
    which results in a spanning dimension of exactly $d_{\mathcal{A}}=3$.
    
    \item \textbf{Case (ii): Fully POVM-spanning family.} From our previous discussion, the maximal POVM spanning dimension for $S=2$ and $O=2$ is $d_{\mathrm{POVM}} = (2-1)2^2 = 4$. To achieve $d_{\mathcal{A}}=4$, the action set must unlock the $\Delta \alpha$ degree of freedom by breaking the trace-$1$ condition of the effects. It suffices to add a single biased, identity-proportional POVM to the action set, for example:
    \begin{equation}
        M_1^{\text{bias}} = \frac{1+b}{2}\mathbb{I},
    \end{equation}
    for some bias parameter $b \in (0, 1)$. Taking the difference between this biased POVM and a projective measurement yields
    \begin{align}
        (\Delta M_1, \Delta M_2) = \frac{1}{2} ( b \mathbb{I} + \Delta \mathbf{n} \cdot \bm{\sigma}, -b\mathbb{I} - \Delta \mathbf{n} \cdot \bm{\sigma})
    \end{align}
    Together with the projective differences, this completes the span:
    \begin{equation}
        \mathsf{V}_{\mathcal{A}} \cong \text{span}_{\mathbb{R}}\{\mathbb{I}, \sigma_x, \sigma_y, \sigma_z\},
    \end{equation}
    hence making the action class fully POVM-spanning with $d_{\mathcal{A}}=4$.
\end{itemize}
\end{remark}

\subsection{Undercomplete assumption and recovery map for continuous actions}
\label{sec:uc-cont}

To apply our learning algorithms to this continuous measure space, we must parameterize the statistical degrees of freedom of the action set $\mathcal{A}$ independently of its cardinality. 
Recall Assumption~\ref{ass:undercomplete} and Definition~\ref{def:Omega_uc}: for each action $a\in\mathcal A$, we require a linear map $R^{(a)}:\mathbb M_{O}\to \mathbb M_{S}\otimes \mathbb M_{O}$ such that it holds uniformly 
\begin{equation}
    R^{(a)}\circ \Tr_S\circ \mathcal P^{(a)}=\mathcal P^{(a)},  \qquad \|R^{(a)}\|_{1\to 1,\diag}\le \kappa_{\mathrm{uc}}.
    \label{eq:uc-cont-recall}
\end{equation}

Importantly, the fact that $\mathcal{A}$ is \emph{continuous} does not by itself make this assumption harder to satisfy: in many natural action classes, one can construct the full family $\big\{\Phi_o^{(a)}\big\}_{o\in[O]}$ and $\{R^{(a)}\}_{a\in\mathcal{A}}$ by \emph{transport} from a single instrument and recovery map, with the same norm bound.

Suppose the available POVMs are generated from a fixed fiducial POVM $M^{(0)}=\{M_o^{(0)}\}_{o\in[O]}$ by unitary conjugation,
\begin{equation}
    M_o^{(a)}=U_a M_o^{(0)}\,U_a^\dagger, \qquad a\in\mathcal{A}. 
\label{eq:unitary-orbit-povm}
\end{equation}
The induced L\"uders instruments satisfy the covariance relation
\begin{equation}
    \mathcal P^{(a)}(X)\coloneqq \big(U_a\otimes \mathbb{I}_O\big) \mathcal P^{(0)}\!\big(U_a^\dagger X U_a\big)  \big(U_a^\dagger\otimes \mathbb{I}_O\big), \qquad \forall X\in\mathbb M_{S}.
\label{eq:instrument-covariance}
\end{equation}
If a recovery map $R^{(0)}$ exists for the fiducial action, then a valid recovery map for action $a$ is obtained by conjugating the output on the system register:
\begin{equation}
    R^{(a)}(X) \coloneqq  \big(U_a\otimes \mathbb{I}_O\big) R^{(0)}(X) \big(U_a^\dagger\otimes \mathbb{I}_O\big), \qquad \forall X\in\mathbb M_{O}.
\label{eq:transport-recovery}
\end{equation}
Indeed, since $\Tr_S$ is invariant under unitary conjugation on $S$, we have
\begin{align}
    R^{(a)}\circ \Tr_S \circ \mathcal{P}^{(a)}=\mathcal P^{(a)}. 
\end{align}
Moreover, unitary conjugation preserves the trace norm, hence
\begin{align}
    \|R^{(a)}\|_{1\to 1,\diag}=\|R^{(0)}\|_{1\to 1,\diag},
\end{align} 
for all $a\in \mathcal{A}$ (and analogously for a diamond-norm bound, if that were the chosen metric).

\begin{remark}[SIC-POVM example]
\label{remark:sic_povm_example}
A SIC-POVM in an $S$-dimensional space consists of $S^2$ rank-1 projectors $\{E_x\}_{x=1}^{S^2}$ satisfying $\Tr(E_i E_j) = \frac{S\delta_{ij}+1}{S+1}$. The valid POVM elements are $M_x = \frac{1}{S} E_x$, such that $\sum_{x=1}^{S^2} M_x = \mathbb{I}_S$. A SIC-POVM is particularly powerful because its elements span the entire $(S^2-1)$-dimensional space of traceless Hermitian operators. This informational completeness ensures that the mapping from the latent quantum state to the classical outcome distribution is injective, providing the exact quantum analogue to the full-rank emission matrix required classically.

Here, to illustrate how unitary covariance guarantees a uniform recovery map, we consider an action set generated by all unitary rotations of a fiducial symmetric informationally complete (SIC) POVM for $S=2$. The fiducial qubit SIC-POVM corresponds to the tetrahedral measurement on the Bloch sphere. The POVM elements are given by:
\begin{equation}
    M_o^{(0)} = \frac{1}{2} E_o = \frac{1}{4}(\mathbb{I} + \mathbf{n}_o \cdot \bm{\sigma}), \quad o \in [4],
\end{equation}
where $\{E_o\}_{o=1}^4$ are rank-$1$ projectors and $\bm{\sigma} = (\sigma_x, \sigma_y, \sigma_z)$ is the vector of Pauli matrices. The Bloch vectors $\mathbf{n}_o$ form a regular tetrahedron, conventionally parameterized as:
\begin{equation}
    \mathbf{n}_1 = \begin{pmatrix} 0 \\ 0 \\ 1 \end{pmatrix}, \quad 
    \mathbf{n}_2 = \frac{1}{3}\begin{pmatrix} 2\sqrt{2} \\ 0 \\ -1 \end{pmatrix}, \quad 
    \mathbf{n}_3 = \frac{1}{3} \begin{pmatrix} -\sqrt{2}  \\ \sqrt{6} \\ -1 \end{pmatrix}, \quad 
    \mathbf{n}_4 = \frac{1}{3} \begin{pmatrix} -\sqrt{2} \\ -\sqrt{6} \\ -1 \end{pmatrix}.
\end{equation}
Because the vectors form a regular tetrahedron, their pairwise dot products satisfy $\mathbf{n}_i \cdot \mathbf{n}_j = -\frac{1}{3}$ for $i \neq j$. By the algebra of Pauli matrices, this directly satisfies the general SIC-POVM inner product condition $\Tr(E_i E_j) = \frac{2\delta_{ij}+1}{3}$. Furthermore, following Definition~\ref{def:embedding-dim}, while the maximum possible spanning dimension for $O=4$ is $d_{\mathrm{POVM}}=12$, the action class $\mathcal{A}$ formed by the $SU(2)$ orbit of this tetrahedron is parameterized by $3\times 3$ spatial rotation matrices. The linear span of the differences of these POVM elements covers the full $9$-dimensional space of $3\times3$ real matrices, yielding a finite, well-defined spanning dimension of $d_{\mathcal{A}} = 9$.

Let the action set consist of its unitary rotations,
\begin{equation}
    M_o^{(U)} \coloneqq  U M_o^{(0)}U^\dagger, \qquad U \in SU(2).
\end{equation}
Since each effect is rank-$1$, the corresponding L\"uders instrument branch for outcome $o$ naturally reduces to a measure-and-prepare form: 
\begin{equation}
    \Phi_o^{(U)}(X) \coloneqq  \sqrt{M_o^{(U)}}X\sqrt{M_o^{(U)}} = \Tr[M_o^{(U)}X]\rho_U^{(o)}, \qquad \rho_o^{(U)} \coloneqq \frac{M_o^{(U)}}{\Tr[M_o^{(U)}]}.
\end{equation}

Let $\{|o\rangle\}_{o=1}^4$ denote the canonical basis of the classical outcome register. Then the full instrument can be written as
\begin{align}
    \mathcal{P}^{(U)}(X) = \sum_{o=1}^4 \Phi_o^{(U)}(X)\otimes \proj{o} = \sum_{o=1}^4 \Tr\!\big[M_o^{(U)}X\big] \rho_o^{(U)}\otimes \proj{o}.
\label{eq:sic-full-instrument}
\end{align}
Tracing out the system register gives the associated observable operator
\begin{equation}
    \big(\Tr_S \circ \mathcal P^{(U)}\big)\!(X) = \sum_{o=1}^4 \Tr\big[M_o^{(U)}X\big] \proj{o}.
    \label{eq:sic-observable-operator}
\end{equation}
Crucially, the informational completeness condition ensures that the coefficients $\Tr\!\big[M_o^{(U)}X\big]$ uniquely determine $X$. Consequently for any diagonal classical operator $D=\sum_{o=1}^4 d_o \proj{o}$, we can define the corresponding recovery map as:
\begin{equation}
    R^{(U)}(D) \coloneqq \sum_{o=1}^4 d_o \rho_o^{(U)}\otimes \proj{o}. 
\label{eq:recovery-sic}
\end{equation}
 Substituting~\eqref{eq:sic-observable-operator} into the recovery map yields exactly:
\begin{equation}
    \big(R^{(U)}\circ \Tr_S\circ\mathcal{P}^{(U)}\big)\!(X) = \sum_{o=1}^4 \Tr[M_o^{(U)}X]\rho_o^{(U)} \otimes \proj{o} = \mathcal{P}^{(U)}(X),
\end{equation}
so 
\begin{align}
    R^{(U)} \circ Tr_S \circ \mathcal{P}^{(U)} = \mathcal{P}^{(U)}, \qquad  \forall U\in SU(2). 
\end{align}

It remains to bound the induced trace norm. Because the summands in the recovery map have orthogonal supports on the classical register and $\big\|\rho_o^{(U)}\big\|_1 = 1$, we obtain for every diagonal $D$:
\begin{equation}
    \|R^{(U)}(D)\|_1 = \sum_{o=1}^4 |d_o| \|\rho_U^{(o)}\|_1 = \sum_{o=1}^4 |d_o| = \|D\|_1.
\end{equation}
Hence, $\|R^{(U)}\|_{1\to1, \text{diag}} = 1$ uniformly in $U$. Thus, we conclude that one may safely take $\kappa_{\mathrm{uc}}=1$ for this continuous action class. This provides a clear example in which the recovery-map assumption holds uniformly over a continuous action set.
\end{remark}

\begin{remark}
The same argument applies to any rank-one POVM family, and more generally to any measure-and-prepare instrument
family: whenever
\begin{align}
    \mathcal P^{(a)}(X) = \sum_o \Tr(M_o^{(a)}X) \sigma_o^{(a)}\otimes \proj{o},
\end{align}
the recovery map
\begin{align}
    R^{(a)}\!\Big(\sum_o d_o \proj{o}\Big) = \sum_o d_o \sigma_o^{(a)}\otimes \proj{o}, 
\end{align}
satisfies
\begin{align}
    R^{(a)}\circ \Tr_S\circ \mathcal P^{(a)}=\mathcal P^{(a)}, \qquad \|R^{(a)}\|_{1\to 1,\diag}=1. 
\end{align}
In contrast, for higher-rank effects $M_o^{(a)}$, the branch $\Phi_o^{(a)}(X)$ generally depends on more than the single scalar $\Tr\big[M_o^{(a)}X\big]$, so such a recovery map need not exist.
\end{remark}

\begin{remark}[Dimensionality of outcomes and classical memory]
While the preceding example utilizes a SIC-POVM, 
the POVM and recovery map can be significantly simplified if the hidden memory is classical, i.e., restricted to diagonal density matrices within a fixed basis. In this case, the recovery map $R^{(a)}$ only needs to invert classical probability distributions rather than full quantum states. Consequently, fewer measurement outcomes are required to ensure the existence of a valid recovery map. As we will see in Section~\ref{sec:work_extraction}, this relaxation naturally arises in physical applications. Specifically, we will use this exact setup for state-agnostic work extraction, where the hidden memory is classical and the agent interacts with the system using only simple, two-outcome projective measurements.
\end{remark}

\subsection{Concentration bounds for continuous actions}
\label{sec:cont-conc}

In the continuous-action setting, each length-$L$ trajectory $\tau_L=(a_1,o_1,\ldots,a_L,o_L)\in(\mathcal{A}\times[O])^L$ is distributed according to the probability measure $\mathbb {P}_\omega^\pi$ induced by the true model $\omega$ and policy $\pi$ (a stochastic kernel). Recall the density $p_{\omega}^{\pi}(\tau_L) \coloneqq \pi(\tau_L)\,A_\omega(\tau_L)$ as defined in~\eqref{eq:traj-density}, where $\pi(\tau_L)$ is the cumulative policy density product in~\eqref{eq:pi-tau-cont} and $A_\omega(\tau_L)$ is the cumulative observation likelihood product in~\eqref{eq:A-tau-cont}.
Consequently, the log-likelihood objective for the MLE mirrors the discrete case. Because the policy density $\pi(\tau_L)$ is independent of the model parameters $\omega$, it strictly factors out of the likelihood ratio. Thus, maximizing the sum of $\log p_{\omega}^{\pi^i}(\tau^i)$ is equivalent to maximizing the sum of the observation log-likelihoods $\log A_{\omega}(\tau^i)$.

\subsubsection{Effective parameter dimension}
As in the discrete-action case (Section~\ref{subsec:embedding_dimension_qhmm}), the MLE concentration
bounds depend on the model class only through the dimension of a finite-dimensional real embedding.
For the continuous-action POVM model class, we denote by $d_\omega$ the corresponding effective
dimension (i.e.\ the number of real degrees of freedom needed to specify the trajectory likelihoods), derived via an analogous parameter-counting argument to that of the discrete setting. Specifically, we substitute the explicit dependence on the action-set cardinality with the spanning dimension $d_{\mathcal{A}}$ introduced in Definition~\ref{def:embedding-dim}. Concretely, we obtain:
\begin{equation}
d_\omega \coloneqq L S^{4} + d_{\mathcal A} S^{2}  + S^{2},
\label{eq:domega-cont}
\end{equation}
where when the action class is fully POVM-spanning, $d_{\mathcal A}=(O-1)S^2$.

\subsubsection{MLE lemmas}

The following two propositions establish the continuous-action counterparts of the likelihood-ratio and trajectory $\ell_1$ concentration bounds. They follow directly from~\cite[Propositions~13--14]{liu2022partially} and~\cite[Propositions~B.1--B.2]{liu2023optimistic} upon upper bounding the associated bracketing entropy via our finite-dimensional embedding of size $d_{\omega}$. To rigorously interface with these results, we evaluate the trajectory laws $\mathbb{P}_{\omega}^{\pi}$ on the space $(\mathcal{A}\times[O])^L$ dominated by the reference measure $\nu_L \coloneqq  \mu^{\otimes L}\otimes \text{count}^{\otimes L}$, identifying the likelihood with the Radon-Nikodym density $p_{\omega}^{\pi} \coloneqq \dd \mathbb{P}_{\omega}^{\pi}/\dd\nu_L$. We provide the adaptation of the original proofs found in~\cite[Appendix B.1 and B.2]{liu2022partially} in our Appendix~\ref{sec:appendix_mle_continuous}.

\begin{proposition}
\label{prop:lr-cont}
There exists a universal constant $c>0$ such that for any $\delta\in(0,1]$, with probability at least $1-\delta$,
for all $t\in[K]$ and all $\omega'\in\Omega$,
\begin{equation}
    \sum_{i=1}^{t} \log\!\left( \frac{p_{\omega'}^{\pi^i}(\tau^i)}{p_{\omega}^{\pi^i}(\tau^i)} \right) \leq c\Big( d_\omega L \log\!\big(K O\big) + \log(K/\delta) \Big),
    \label{eq:lr-cont}
\end{equation}
where $\omega$ is the true model and $\tau^i\sim \mathbb{P}_{\omega}^{\pi^i}$ is the trajectory observed in episode $i$.
\end{proposition}

\begin{proposition}
\label{prop:l1-cont}
There exists a universal constant $c>0$ such that for any $\delta\in(0,1]$, with probability at least $1-\delta$,
for all $t\in[K]$ and all $\omega'\in\Omega$,
\begin{equation}
    \sum_{i=1}^{t} \Big\|\mathbb{P}_{\omega'}^{\pi^i}-\mathbb{P}_{\omega}^{\pi^i}\Big\|_1^{\,2} \leq c\left( \sum_{i=1}^{t} \log\!\left( \frac{p_{\omega}^{\pi^i}(\tau^i)}{p_{\omega'}^{\pi^i}(\tau^i)} \right) + d_\omega L \,\log\!\big(K O\big) + \log(K/\delta) \right),
\label{eq:l1-cont}
\end{equation}
where $\|\cdot\|_1$ is the $\ell_1$ distance between measures defined in~\eqref{eq:l1-tv-cont} (i.e. twice total variation; integrals over actions are w.r.t. $\mu$ and sums over outcomes are w.r.t. counting measure).
\end{proposition}

\subsubsection{Choice of confidence radius}

For the continuous-action variant of Algorithm~\ref{alg:omle-qhmm}, we fix the confidence radius $\beta$ across all episodes. Analogous to the discrete setting, this radius is determined by the complexity term bounding the right-hand side of Propositions~\ref{prop:lr-cont} and~\ref{prop:l1-cont}: 
\begin{equation}
    \beta \coloneqq  c\Big( d_\omega  L\log\!\big(K O\big) + \log(K/\delta) \Big),
    \label{eq:beta-cont}
\end{equation}
where $c>0$ is the universal constant from Propositions~\ref{prop:lr-cont} and~\ref{prop:l1-cont}.
This is the radius we plug into the MLE confidence sets in the continuous-action analysis.

\subsection{Decomposing regret via OOM estimation errors for continuous actions}

The definition of regret in the continuous-action setting is identical to the discrete case, with the corresponding expectations taken under the trajectory measures introduced in Section~\ref{sec:continous_model}. We run the same algorithm as in Algorithm~\ref{alg:omle-qhmm}, and in the regret analysis we use the undercomplete OOM decomposition from Section~\ref{sec:oom_uc}.

To avoid repetition, we only highlight the steps that differ from the discrete-action proof in Section~\ref{sec:regret_discrete}. In particular, the decomposition of the regret into estimation errors of the OOM operators relative to their empirical estimates proceeds verbatim: Lemma~\ref{lem:one_norm_op} extends directly to the continuous setting (replacing sums over actions by integrals with respect to $\mu$). Concretely, we will rely on the following continuous-action analogue.

\begin{corollary}[Lemma~5.2 for continuous actions]
\label{cor:lem52-cont}
Fix $\delta\in(0,1)$ and assume the rewards are uniformly bounded $|r_l(a_l, o_l)| \le R$ for all $l \in [L], a_l \in \mathcal{A}$, and $o_l \in [O]$.
Let $\omega\in\Omega_{\mathrm{uc}}$ be the true model and run Algorithm~\ref{alg:omle-qhmm} with confidence radius
$\beta_k\equiv\beta$ chosen as in~\eqref{eq:beta-cont}. Then, with probability at least $1-\delta$,
\begin{equation}
    \regret(K) \leq  LR\sum_{k=1}^K \Big\|\mathbb{P}^{\pi^k}_{\omega_k}-\mathbb{P}^{\pi^k}_{\omega}\Big\|_1.
    \label{eq:regret-tv-cont}
\end{equation}
Moreover, we can bound the regret in terms of OOM estimation error as
\begin{equation}
\label{eq:regret-oom-cont}
\begin{aligned}
    \regret(K) & \leq S^2\kappa_{\mathrm{uc}} LR  \sum_{k=1}^K\sum_{l=1}^{L-1} \int_{(\mathcal{A}\times[O])^{l+1}} \pi^k(\tau_{l+1}) \big\| \big(\mathbf{A}_l^k-\mathbf{A}_l\big) \mathbf{a}_l \big\|_1  \nu_{l+1}(\dd \tau_{l+1}) \\
    & \quad +  S^2\kappa_{\mathrm{uc}} LR \sum_{k=1}^K \int_{\mathcal{A}}  \pi_1^k(a_1) \big\|\mathbf{a}^k(a_1)-\mathbf{a}(a_1)\big\|_1 \mu(\dd a_1),
\end{aligned}
\end{equation}
where $\nu_{l+1}:=\mu^{\otimes(l+1)}\otimes\mathrm{count}^{\otimes(l+1)}$ is the product reference measure and
$\pi^k(\tau_{l+1})$ denotes the corresponding trajectory density induced by $(\pi^k,\omega)$.
\end{corollary}

\begin{proof}
Let $\mathcal{E}$ be the high-probability event upon which the continuous likelihood-ratio trajectory $\ell_1$ concentration bounds of Propositions~\ref{prop:lr-cont} and~\ref{prop:l1-cont} hold simultaneously for all $t\le K$ with radius $\beta$.
Conditioned on $\mathcal{E}$, the true model $\omega$ remains in the confidence set, and the optimism argument follows identically to Lemma~\ref{lem:regret_step1}, yielding~\eqref{eq:regret-tv-cont}. Note that discrete sums over actions are naturally replaced by total variation integrals on $(\mathcal{A}\times[O])^L$.

For the bound in~\eqref{eq:regret-oom-cont}, we perform the identical telescoping expansion of the OOM product as in the discrete case. Applying the continuous-action analogue of Lemma~\ref{lem:one_norm_op} to control the trace norms yields the result, where sums over actions are replaced by integrals with respect to $\mu$.
\end{proof}

\subsection{The eluder dimension for linear functions (continuous actions)}
\label{sec:eluder_cont}

Similar to the discrete case in~\eqref{eq:guarantee} and~\eqref{eq:initial_estimate_error}, we get from the MLE bound that 
\begin{equation}
\label{eq:guarantee_cont}
\begin{aligned}
    \sum_{t=1}^{k-1} \sum_{o_{1:l-1}} \int_{\mathcal{A}^l} \pi^t(\tau_l) \big\| \mathbf{A}^k_{l-1}\cdots \mathbf{A}^k_1\mathbf{a}^k(a_1 ) -   \mathbf{A}_{l-1}\cdots \mathbf{A}_1\mathbf{a}(a_1 ) \big\|_1 \mu^{\otimes l} (\dd a_{1:l})= \mathcal{O}(\sqrt{\beta k} ),
\end{aligned}
\end{equation}
which, in particular, yields for the first action:
\begin{align}
    \label{eq:initial_estimate_error_cont}
    \sum_{t=1}^{k-1} \int_{\mathcal{A}} \pi^{t} (a_1 ) \|\mathbf{a}^k (a_1) -  \mathbf{a} (a_1)\|_1 \mu(\dd a_1) = \mathcal{O} (\sqrt{\beta k}). 
\end{align}
With these bounds in hand, we are now ready to bound the subsequent regret with the historical OOM estimation error. We apply the same general eluder lemma for linear functions (Proposition~\ref{prop:eluder_general}) as in the discrete case, with the choice of functions modified. We again bound the first action term and the remaining terms separately.

To streamline the subsequent norm bounds, we first establish a universal bound on the coefficients of any quantum operation when expanded in the Hilbert-Schmidt basis. 

\begin{lemma}[Universal Basis Coefficient Bound]
\label{lem:basis-coefficient-bound}
Let $M \succeq 0$ be a positive semi-definite operator acting on $\mathbb{M}_S$. Expanding $M$ in the fixed Hilbert-Schmidt orthonormal Hermitian basis $\{P_\mu\}_{\mu=1}^{S^2}$, the absolute sum of its coefficients is bounded as:
\begin{align}
    \sum_{\mu=1}^{S^2} |\Tr(P_\mu M)| \le S \Tr(M).
\end{align}
\end{lemma}
\begin{proof}
Applying the Cauchy-Schwarz inequality over the $S^2$ basis elements, and noting that the sum of squared coefficients is exactly the squared Hilbert-Schmidt norm of $M$, we have:
\begin{equation}
    \sum_{\mu=1}^{S^2} |\Tr(P_\mu M)| \le \sqrt{S^2} \sqrt{\sum_{\mu=1}^{S^2} \Tr(P_\mu M)^2} = S\|M\|_2.
\end{equation}
Because $M \succeq 0$, its eigenvalues are non-negative. Since the 2-norm of the eigenvalues is strictly bounded by their $1$-norm (the trace)
\begin{align}
    \|M\|_2 \le \|M\|_1 = \Tr(M), 
\end{align}
which completes the proof.
\end{proof}

\subsubsection{First action}

We begin by establishing a linear decomposition of the POVM elements associated with each action, which we will subsequently apply to the initial observable operators. 

\begin{lemma}[Linearity of the initial observable operator in the first action]
    \label{lem:linear-first-action}
    Recall the initial observable vector $\mathbf{a}(a_1):=\Tr_S[\mathcal P^{(a_1)}(\rho_1)]$ (and similarly $\mathbf{a}^k(a_1):=\Tr_S[\mathcal P^{(a_1)}(\rho_1^k)]$) from Definition~\ref{def:oom_qhmm}, where $\mathcal P^{(a_1)}$ is the (first-step) instrument associated to action $a_1$ with POVM elements $\{M_{o_1}^{(a_1)}\}_{o_1\in[O]}$. 
    Let $\Delta\rho_1^k\coloneqq\rho_1^k-\rho_1$. Moreover, fixing the Hilbert-Schmidt orthonormal Hermitian basis $\{P_\mu\}_{\mu=1}^{S^2}$ used throughout we have
    \begin{align}
        \|\mathbf{a}^k(a_1)-\mathbf{a}(a_1)\|_1=\sum_{o_1\in[O]}\Big|\sum_{\mu=1}^{S^2}\phi_\mu(a_1,o_1) u_\mu^k\Big| ,
    \end{align}
    where $\phi_\mu(a_1,o_1):=\Tr(P_\mu M_{o_1}^{(a_1)})$ and $u_\mu^k:=\Tr(P_\mu\Delta\rho_1^k)$.
\end{lemma}

\begin{proof}
Fix $a_1\in\mathcal A$. By Definition~\ref{def:oom_qhmm}, the reduced observable operator is classical in the outcome register, i.e.
\begin{equation}
    \Tr_S[\mathcal P^{(a_1)}(\rho)] =\sum_{o_1\in[O]}\Tr(M_{o_1}^{(a_1)}\rho)\proj{o_1}, \qquad \forall \rho\in \mathsf{D}(\mathcal{H}_S).
    \label{eq:first-action-reduction}
\end{equation}
Applying~\eqref{eq:first-action-reduction} to $\rho_1^k$ and $\rho_1$ and subtracting gives
\begin{equation}
    \mathbf{a}^k(a_1)-\mathbf{a}(a_1) = \sum_{o_1\in[O]} \Tr(M_{o_1}^{(a_1)} \Delta\rho_1^k) \proj{o_1}.
\label{eq:first-action-diff-diag-proof}
\end{equation}
Since~\eqref{eq:first-action-diff-diag-proof} is diagonal in the $\{|o_1\rangle\}$ basis, its trace norm equals the $\ell_1$-norm of the diagonal entries, hence
\begin{equation}
    \|\mathbf{a}^k(a_1)-\mathbf{a}(a_1)\|_1 =\sum_{o_1\in[O]}\Big|\Tr(M_{o_1}^{(a_1)}\Delta\rho_1^k)\Big|.
\label{eq:first-action-L1-proof}
\end{equation}
Now expand each POVM element in the fixed Hilbert-Schmidt orthonormal Hermitian basis as
\begin{align}
    M_{o_1}^{(a_1)}=\sum_{\mu=1}^{S^2}\phi_\mu(a_1,o_1) P_\mu, 
\end{align}
where $\phi_\mu(a_1,o_1) =\Tr(P_\mu M_{o_1}^{(a_1)})$. 
By linearity of the trace,
\begin{equation}
    \Tr(M_{o_1}^{(a_1)}\Delta\rho_1^k) =\sum_{\mu=1}^{S^2}\phi_\mu(a_1,o_1) \Tr(P_\mu\Delta\rho_1^k) =\sum_{\mu=1}^{S^2}\phi_\mu(a_1,o_1) u_\mu^k.
    \label{eq:first-action-linear-proof}
\end{equation}
Substituting~\eqref{eq:first-action-linear-proof} into~\eqref{eq:first-action-L1-proof} yields
\begin{align}
    \|\mathbf{a}^k(a_1)-\mathbf{a}(a_1)\|_1 =\sum_{o_1\in[O]}\Big|\sum_{\mu=1}^{S^2}\phi_\mu(a_1,o_1) u_\mu^k\Big|,
\end{align}
which is exactly the claimed identity.
\end{proof}

We can now generalize the bound on the historical OOM estimation error  to a bound on the subsequent regret: 

\begin{lemma}[First-action contribution for continuous actions]
\label{lem:first-action-cont}
    Work on the same high-probability event and under the same assumptions as in Lemma~\ref{lem:regret_step1}. Assume moreover that the continuous bound from the MLE holds i.e. 
    \begin{equation}
        \sum_{t=1}^{k-1}\int_{\mathcal A}\pi_1^t(a_1) \|\mathbf{a}^k(a_1)-\mathbf{a}(a_1)\|_1 \mu(\dd a_1)\leq  \zeta_k , \qquad \forall k\in[K],
    \label{eq:assum_first_action}
\end{equation}
with $\zeta_k=O(\sqrt{\beta k})$. Then
\begin{equation}
    \sum_{k=1}^{K}\int_{\mathcal{A}}\pi_1^k(a_1) \|\mathbf{a}^k(a_1)-\mathbf{a}(a_1)\|_1 \mu(\dd a_1) =\widetilde{\mathcal O}\!\Big( S^2\log^2(K) \big(S^2+\sqrt{\beta K}\big) \Big).
\label{eq:first-action-bound-cont}
\end{equation}
\end{lemma}

\begin{proof}

Let $ \mathcal{X} \coloneqq \{\star\}$ (a singleton) equipped with counting measure $\mathrm{count}_{\mathcal{X}}$, and $\mathcal{Y} \coloneqq \mathcal{A}\times[O] $ equipped with  $\nu\coloneqq\mu\otimes \mathrm{count}$. For each $k\in[K]$ define the function $f^k:\mathcal X\to\mathbb R^{S^2}$ by
\begin{align}
    f^{k,\mu}(\star) \coloneqq  u_\mu^k=\Tr(P_\mu\Delta\rho_1^k), \qquad \mu\in[S^2].
\end{align}
For each $t\in[K]$, define $g^t:\mathcal{Y} \to \mathbb{R}^{S^2}$ by
\begin{align}
    g^{t,\mu}(a_1,o_1) \coloneqq  \pi_1^t(a_1)\,\phi_\mu(a_1,o_1), \qquad \mu\in[S^2].
\end{align}
Then, for every $(a_1,o_1)\in\mathcal Y$,
\begin{align}
    \langle f^k(\star),g^t(a_1,o_1)\rangle =\pi_1^t(a_1)\sum_{\mu=1}^{S^2}\phi_\mu(a_1,o_1) u_\mu^k.
\end{align}
Hence, using the previous Lemma~\ref{lem:linear-first-action} we have
\begin{equation}
\begin{aligned}
    & \int_{\mathcal A}\pi_1^t(a_1) \|\mathbf{a}^k(a_1)-\mathbf{a}(a_1)\|_1 \mu(\dd a_1)
    =\int_{\mathcal A}\sum_{o_1\in[O]} \Big|\sum_{\mu=1}^{S^2}\phi_\mu(a_1,o_1) u_\mu^k\Big|\mu(\dd a_1) \\
    &=\int_{\mathcal{A}}\sum_{o_1\in[O]} \big|\langle f^k(\star),g^t(a_1,o_1)\rangle\big| \mu(\dd a_1) =\int_{\mathcal Y}\big|\langle f^k(\star),g^t(y)\rangle\big| \nu(\dd y) .
\label{eq:eluder-form-first-action}
\end{aligned}
\end{equation}
Therefore, Assumption \eqref{eq:assum_first_action} can be rewritten as
\begin{align}
    \sum_{t=1}^{k-1}\int_{\mathcal Y}\big|\langle f^k(\star),g^t(y)\rangle\big| \nu(dy) \leq \zeta_k,
\end{align}
which is exactly the ``mixed term'' condition in the general eluder lemma (Proposition~\ref{prop:eluder_general}).

It remains to bound the corresponding norms of both $f^k,g^t$.
First, since $\rho_1^k$ and $\rho_1$ are density operators, $\Delta\rho_1^k$ is Hermitian and
$\|\Delta\rho_1^k\|_1\le 2$. For a Hilbert-Schmidt orthonormal basis,
$\|u^k\|_2=\|\Delta\rho_1^k\|_2\le \|\Delta\rho_1^k\|_1\le 2$, hence $\|u^k\|_\infty\le 2$, so
\begin{align}
    \sup_{x\in\mathcal X}\|f^k(x)\|_\infty = \|f^k(\star)\|_\infty \leq 2.
\end{align}
Leveraging Lemma~\ref{lem:basis-coefficient-bound} on the POVM elements $M_{o_1}^{(a_1)}$, we bound the integral of $g^t$:
\begin{align}
    \int_{\mathcal Y}\|g^t(y)\|_1 \nu(\dd y) &=\int_{\mathcal{A}} \mu(\dd a_1) \pi_1^t(a_1) \sum_{o_1\in[O]} \sum_{\mu=1}^{S^2} \big|\phi_\mu(a_1,o_1)\big| \\
    &\leq \int_{\mathcal{A}} \mu(\dd a_1) \pi_1^t(a_1)\sum_{o_1\in[O]} S \Tr(M_{o_1}^{(a_1)}) \\
    &= \int_{\mathcal{A}} \mu(\dd a_1) \pi_1^t(a_1) S \Tr(\mathbb{I}_S)  = S^2  . 
\end{align}
Finally we invoke the general eluder lemma (Proposition~\ref{prop:eluder_general}) with dimension $d=S^2$,
envelope parameters $R_x=2$ and $R_y=S^2$, and parameter
$\max_{k\le K}\zeta_k=O(\sqrt{\beta K})$, and this yields the result.
\end{proof}

\subsubsection{Remaining terms}

Analogous to Lemma~\ref{lem:linear-first-action}, we decompose the second contribution to the regret into an action-dependent component and an action-independent component. The key insight is that, because the action set (the POVMs or instruments) is known and identical for both the estimated and true models, the trajectory error can be decoupled. This allows us to isolate the error originating purely from the mismatch in the estimated quantum channels.

\begin{lemma}[OOM Estimation Error Decomposition]\label{lem:quantum_forward_factorization}
Let $\tilde{\rho}_l(\tau_{l-1})$ be the unnormalized filtered memory state defined in~\eqref{eq:unnormalized_filter}, and fix $\{F_\mu\}_{\mu=1}^{S^2}$ as the orthonormal Hilbert-Schmidt basis as in~\eqref{eq:hs_basis}. Because the instruments are known, we can define the following known, action-dependent coefficients:
\begin{align}\label{eq:action_dependent_terms}
    \phi_\mu(a_{l+1}, o_{l+1}) &:= \Tr\big(P_\mu M_{o_{l+1}}^{(a_{l+1})}\big), \\
    c_\lambda(\tau_l) &:= \Tr\big(P_\lambda \Phi_{o_l}^{(a_l)}(\tilde{\rho}_l(\tau_{l-1}))\big).
\end{align}
Then, the scalar physical prediction error for outcome $o_{l+1}$ admits the exact multi-linear decomposition:
\begin{equation} \label{eq:quantum_factorized}
    [\Delta \mathbf{A}_l^k \mathbf{a}_l]_{o_{l+1}} = \sum_{\mu=1}^{S^2} \sum_{\lambda=1}^{S^2} \underbrace{\Tr\big( P_\mu \Delta \mathbb{E}_l^k (P_\lambda) \big)}_{\text{Action-Independent Error}} \cdot \underbrace{\phi_\mu(a_{l+1}, o_{l+1}) c_\lambda(\tau_l)}_{\text{Action-Dependent Data}},
\end{equation}
where $\Delta \mathbf{A}_l^k \coloneqq  \mathbf{A}_l^k - \mathbf{A}_l$ is the observable operator error, and $\Delta\mathbb{E}_l^k \coloneqq \mathbb{E}_l^k - \mathbb{E}_l$ is the underlying channel error, which is completely independent of the chosen actions.
\end{lemma}

\begin{proof}
Recall that the OOM operator is given by
\begin{align}
    & \mathbf{A}_l  =\mathbf{A}(o_l, a_l, a_{l+1}) = \Tr_S \circ \mathcal{P}^{(a_{l+1})} \circ (\mathbb{E}_l \otimes \mathcal{T}_{o_l}) \circ R^{(a_l)}, \\ 
    &  \mathbf{a}_l  = \mathbf{a}(\tau_{l-1},a_l)  = \mathbf{A}_{l-1:1} \mathbf{a}(a_1) =  \mathbf{A}_l\cdots \mathbf{A}_1 (\Tr_S\circ\mathcal{P}^{(a_1)})\!(\rho_1). 
\end{align}
Because the instruments are known, the estimated operator $\mathbf{A}_l^k$ differs from the true operator only in the channel $\mathbb{E}_l^k$. 
Using the undercomplete assumption (Assumption~\ref{ass:undercomplete}), the true classical outcome state at step $l$ is:
\begin{equation}
\begin{aligned}
    \mathbf{a}_l = \big(\Tr_S \circ \mathcal{P}^{(a_l)} \circ \mathbb{E}_{l-1}\circ \Phi^{(a_{l-1})}_{o_{l-1}} \circ \dots \circ \mathbb{E}_{1}\circ\Phi^{(a_{1})}_{o_{1}}\big)\!(\rho_1) = \Tr_S \big[\mathcal{P}^{(a_l)}(\tilde{\rho}_l(\tau_{l-1}))\big] . 
\end{aligned}
\end{equation}
By the undercomplete assumption (Assumption~\ref{ass:undercomplete}), applying the recovery map perfectly lifts this observable state back to the quantum memory as
\begin{equation}
    R^{(a_l)}\big(\mathbf{a}_l\big) = \mathcal{P}^{(a_l)}(\tilde{\rho}_l(\tau_{l-1})).
\end{equation}
Applying the superoperator $(\mathbb{E}_l^k \otimes \mathcal{T}_{o_l})$ to this state extracts the specific $o_l$-th component and applies the estimated channel, yielding exactly $\mathbb{E}_l^k \big( \Phi_{o_l}^{(a_l)}(\tilde{\rho}_l(\tau_{l-1})) \big)$. Finally, applying the subsequent instrument $\Tr_S \circ \mathcal{P}^{(a_{l+1})}$ maps this back to the classical probability distribution over $o_{l+1}$. 

Subtracting the true operator $\mathbf{A}_l$, and leveraging the fact that the instrument $\mathcal{P}^{(a_{l+1})}$ and the recovery map $R^{(a_l)}$ are identical for both the estimated and true models, we obtain the scalar error:
\begin{equation}
    [\Delta \mathbf{A}_l^k \mathbf{a}_l]_{o_{l+1}} = \Tr\left( M_{(o_{l+1})}^{a_{l+1}} (\mathbb{E}_l^k - \mathbb{E}_l)\!\Big( \Phi_{o_l}^{(a_l)}(\tilde{\rho}_l(\tau_{l-1})) \Big) \right).
\end{equation}
Expanding the operators $M_{(o_{l+1})}^{a_{l+1}}$ and $\Phi_{o_l}^{(a_l)}(\tilde{\rho}_l(\tau_{l-1}))$ in the orthonormal Hilbert-Schmidt basis $P_\mu$ and factoring out the trace directly yields the multi-linear sum in~\eqref{eq:quantum_factorized}. 
\end{proof}

Following the same logic used for the first action, we now bound the contribution of the OOM operator errors for all subsequent steps $l \in [L-1]$. We first utilize the multi-linear structure of the OOM error derived in Lemma~\ref{lem:quantum_forward_factorization}.

\begin{lemma}[OOM estimation error contribution: continuous actions]\label{lem:oom_error_continuous}
Fix $l \in [L-1]$. Let $\dd\tilde{\nu}_{l+1} = \mu^{\otimes (l+1)}(\dd a_{1:l+1}) \times \text{count}^{\otimes l}(\dd o_{1:l})$ be the reference measure over trajectories up to the $(l+1)$-th action. Assume that for all $k \in [K]$, the cumulative local observable error at step $l$ satisfies:
\begin{equation} \label{eq:local_observable_error_assumption}
    \sum_{t=1}^{k-1} \int_{(\mathcal{A} \times [O])^l \times \mathcal{A}} \pi^t(\tau_{l+1}) \big\| \big( \mathbf{A}_l^k - \mathbf{A}_l \big) \mathbf{a}_l \big\|_1 \dd \tilde{\nu}_{l+1} \le \zeta_k,
\end{equation}
with $\zeta_k = \mathcal{O}(S^2 \kappa_{\mathrm{uc}} \sqrt{\beta k})$. Then we can bound the cumulative contribution to the regret as:
\begin{equation}
    \sum_{k=1}^{K} \int_{(\mathcal{A} \times [O])^l \times \mathcal{A}} \pi^k(\tau_{l+1}) \big\| \big( \mathbf{A}_l^k - \mathbf{A}_l \big) \mathbf{a}_l\big\|_1 \dd\tilde{\nu}_{l+1} = \widetilde{\mathcal{O}}(S^6 \kappa_{\mathrm{uc}} \sqrt{\beta K}).
\end{equation}
\end{lemma}

\begin{proof}
We apply the general eluder lemma (Proposition~\ref{prop:eluder_general}) to the multi-linear factorization derived in Lemma~\ref{lem:quantum_forward_factorization}. Let the discrete index space be a singleton $\mathcal{X} := \{*\}$ equipped with the counting measure. 
To bound the integral of the $\ell_1$ norm, we absorb the norm's internal sum over $o_{l+1}$ directly into the eluder integration domain. Thus, we define the expanded trajectory space $\mathcal{Y} \coloneqq (\mathcal{A} \times [O])^{l+1}$ equipped with the measure $\dd\nu_{\mathcal{Y}} = \dd\tilde{\nu}_{l+1} \times \text{count}(\dd o_{l+1})$. We define the target dimension $d = S^4$, mapped to the index pair $(\mu, \lambda) \in [S^2] \times [S^2]$. 
Using~\eqref{eq:quantum_factorized}, we define the functions $f^k: \mathcal{X} \to \mathbb{R}^{S^4}$ and $g^t: \mathcal{Y} \to \mathbb{R}^{S^4}$ as:
\begin{align}
    f^k(*)_{(\mu, \lambda)} & \coloneqq  \Tr\big(P_\mu \Delta\mathbb{E}_l^k(P_\lambda)\big), \\
    g^t(\tau_{l+1})_{(\mu, \lambda)} &:= \pi^t(\tau_{l+1}) \phi_\mu(a_{l+1}, o_{l+1}) c_\lambda(\tau_l).
\end{align}
Their inner product perfectly reconstructs the scalar prediction error for a specific outcome $o_{l+1}$:
\begin{equation}
    \langle f^k(*), g^t(\tau_{l+1}) \rangle = \pi^t(\tau_{l+1}) [\Delta \mathbf{A}_l^k \mathbf{a}_l]_{o_{l+1}}.
\end{equation}
Consequently, integrating the absolute value of this inner product over $\mathcal{Y}$ exactly reproduces the integral of the $\ell_1$ norm over the restricted space $(\mathcal{A} \times [O])^l \times \mathcal{A}$. We now derive the required norm bounds for the eluder lemma.

\textbf{Step 1: Norm condition for $f^k$ ($R_x$).} 
By H\"older inequality, $|\Tr(AB)| \le \|A\|_\infty \|B\|_1$. For the basis elements, we have $\|P_\mu\|_\infty =\frac{1}{\sqrt{S}}$ for the generalized Pauli basis. The channel error norm is bounded as $\|\Delta\mathbb{E}_l^k\|_{1 \to 1} \le 2$. Applying Cauchy-Schwarz inequality to the $S \times S$ matrices yields $\|P_\lambda\|_1 \le \sqrt{S}\|P_\lambda\|_2 = \sqrt{S}$. Because $\mathcal{X}$ is a singleton, this directly bounds the $\ell_\infty$ norm, yielding the uniform bound
\begin{equation}
    \int_{\mathcal{X}} \|f^k(x)\|_\infty =|f^k(*)_{(\mu, \lambda)}| \le \|P_\mu\|_\infty \|\Delta\mathbb{E}_l^k(P_\lambda)\|_1 \le \frac{1}{\sqrt{S}} \cdot 2 \cdot \sqrt{S} = 2.
\end{equation}
Therefore, 
\begin{align}
    R_x  = 2. 
\end{align}

\textbf{Step 2: Norm condition for $g^t$ ($R_y$).} 
We integrate the $\ell_1$ norm of $g^t$ over the expanded space $\mathcal{Y}$. Applying Lemma~\ref{lem:basis-coefficient-bound} to the measurement features, we have $\sum_{\mu}|\phi_\mu(a_{l+1},o_{l+1})| \le S\Tr(M_{o_{l+1}}^{(a_{l+1})})$. Summing over the outcomes yields $S \Tr(\sum_{o_{l+1}} M_{o_{l+1}}^{(a_{l+1})}) = S\Tr(\mathbb{I}_S) = S^2$. 
Similarly, applying Lemma~\ref{lem:basis-coefficient-bound} to the action-dependent terms from~\eqref{eq:action_dependent_terms} gives:
\begin{equation}
    \sum_{\lambda=1}^{S^2} |c_\lambda(\tau_l)| \le S \big\| \Phi_{o_l}^{(a_l)}(\tilde{\rho}_l(\tau_{l-1})) \big\|_1 = S \Tr\big( \Phi_{o_l}^{(a_l)}(\tilde{\rho}_l(\tau_{l-1})) \big).
\end{equation}
By factoring out the uniform bound $S^2$ from the measurement features, we integrate $\|g^t(y)\|_1$ over $\dd\nu_{\mathcal{Y}}$ by unrolling the trajectory step-by-step:
\begin{equation}
\begin{aligned}
    \int_{\mathcal{Y}} \|g^t(y)\|_1 \dd\nu_{\mathcal{Y}} &\le \int \dd\tilde{\nu}_{l} \pi^t(\tau_{l})  \sum_{\lambda=1}^{S^2} |c_\lambda|   \int \mu(\dd a_{l+1}) \pi_{l+1}^t(a_{l+1}|\tau_l) \sum_{o_{l+1}}  \sum_{\mu=1}^{S^2} |\phi_\mu|  \\
    &\le S^2 \int \dd\tilde{\nu}_{l-1} \pi^t(\tau_{l-1}) \int \mu(\dd a_l) \pi_l^t(a_l|\tau_{l-1}) \sum_{o_l} S \Tr\big( \Phi_{o_l}^{(a_l)}(\tilde{\rho}_l(\tau_{l-1})) \big)  \\
    &= S^3 \int \dd\tilde{\nu}_{l-1} \pi^t(\tau_{l-1}) \Tr(\tilde{\rho}_l(\tau_{l-1})) = S^3.
\end{aligned}
\end{equation}
Here we utilized the trace-preserving property of the instrument (summing over $o_l$) and the fact that the policy density integrates to 1 over $a_l$. Finally, because $\Tr(\tilde{\rho}_l(\tau_{l-1}))$ is exactly the observation probability $A_\omega(\tau_{l-1})$, the integrand $\pi^t(\tau_{l-1})A_\omega(\tau_{l-1})$ is precisely the true trajectory density $p_\omega^{\pi^t}(\tau_{l-1})$, which integrates to 1. Thus,
\begin{align}
    R_y = S^3. 
\end{align}

\textbf{Step 3: Conclusion.} 
By invoking the general eluder lemma (Proposition~\ref{prop:eluder_general}) with dimension $d = S^4$ and the assumption in~\eqref{eq:local_observable_error_assumption}, we obtain:
\begin{align}
    \sum_{k=1}^K \int_{\mathcal{Y}} \big|\langle f^k(*), g^k(\tau_{l+1}) \rangle\big| d\nu_{\mathcal{Y}} &= \mathcal{O}\left( S^4 \log^2(K) (R_x R_y + \zeta_K) \right) = \widetilde{\mathcal{O}}(S^6 \kappa_{\mathrm{uc}} \sqrt{\beta K}).
\end{align}
\end{proof}

\subsection{Cumulative regret bound for continuous actions}

We now combine the continuous first-action bound and the OOM estimation error bounds to establish the final regret scaling for the continuous-action setting. We start by explicitly verifying that the cumulative prediction error assumption ($\zeta_k$) required by the Eluder lemmas holds with high probability.

\begin{theorem}[Regret bound for continuous actions]
\label{thm:continuous_regret}
Fix $\delta \in (0,1)$ and assume the rewards are uniformly bounded $|r_l(a_l, o_l)| \le R$ for all $l \in [L], a_l \in \mathcal{A}$, and $o_l \in [O]$. Let the true environment satisfy $\omega_L \in \Omega_{\mathrm{uc}}$ with robustness constant $\kappa_{\mathrm{uc}}$. Run Algorithm~\ref{alg:omle-qhmm} with continuous POVM action spaces and confidence radius $\beta$ chosen as in~\eqref{eq:beta-cont}:
\begin{equation}
    \beta \coloneqq c\Big( d_\omega L \log\!\big(K  O\big) + \log(K/\delta) \Big),
\end{equation}
where $d_\omega$ is the continuous effective embedding dimension. Then, with probability at least $1-\delta$,
\begin{equation}
    \regret(K) \le \widetilde{\mathcal{O}}\left( R S^8 \kappa_{\mathrm{uc}}^2 L^2 \sqrt{K \beta} \right).
\end{equation}
In particular, by substituting the explicit form of $\beta$ and absorbing polylogarithmic factors, we obtain
\begin{equation}
    \regret(K) \le \widetilde{\mathcal{O}}\Big( R S^{10} \kappa_{\mathrm{uc}}^{2} L^{2} \sqrt{K(L^2 +LO)} \Big).
\end{equation}
\end{theorem}

\begin{proof}
Let $\mathcal{E}$ be the high-probability event upon which the continuous likelihood-ratio and trajectory $\ell_1$ concentration bounds (Propositions~\ref{prop:lr-cont} and~\ref{prop:l1-cont}) hold simultaneously for all $t \in [K]$ with radius $\beta$. On this event, the true model $\omega$ remains in the confidence set, enabling the standard optimistic inequality $V^{\pi^k}(\omega) \le V^{\pi^k}(\omega_k)$. 

\paragraph{Step 1: Verifying the past error bound ($\zeta_k$).}
By Proposition~\ref{prop:l1-cont}, the empirical sum of the squared trajectory $\ell_1$ distances is bounded by $\mathcal{O}(\beta)$. Applying the Cauchy-Schwarz inequality over the $k-1$ steps yields the global trajectory bound:
\begin{equation} \label{eq:global_traj_bound}
    \sum_{t=1}^{k-1} \big\| \mathbb{P}^{\pi^t}_{\omega_k} - \mathbb{P}^{\pi^t}_\omega \big\|_1 \le \mathcal{O}\left( \sqrt{\beta k} \right).
\end{equation}
To extract the local error at step $l$, we follow the exact same reduction used in the discrete setting (specifically, the derivation leading to~\eqref{eq:promise}). By telescoping the OOM operator products and marginalizing over the future trajectory variables (the subsequent actions and outcomes), we isolate the step-$l$ error. Applying the continuous-action analogue of Lemma~\ref{lem:one_norm_op}, the operator norm of the remaining future sequence is uniformly bounded by $S^{\frac{5}{2}} \kappa_{\mathrm{uc}}$. Therefore, pushing the global $\ell_1$ trajectory bound through this future sequence yields:
\begin{equation}
    \sum_{t=1}^{k-1} \int_{(\mathcal{A} \times [O])^l \times \mathcal{A}} \pi^t(\tau_{l+1}) \big\| \big(\mathbf{A}_l^k - \mathbf{A}_l\big) \mathbf{a}_l\big\|_1 \dd \nu_{l+1} \le \mathcal{O}\left( S^2 \kappa_{\mathrm{uc}} \sqrt{\beta k} \right).
\end{equation}
This explicitly verifies the assumption in~\eqref{eq:local_observable_error_assumption} for Lemma~\ref{lem:oom_error_continuous} (and similarly for the first-action Lemma~\ref{lem:first-action-cont}) with $\zeta_k = \mathcal{O}(S^2 \kappa_{\mathrm{uc}} \sqrt{\beta k})$.

\paragraph{Step 2: Bounding the total regret.}
Using Corollary~\ref{cor:lem52-cont}, the cumulative regret is bounded by the sum of the OOM estimation errors:
\begin{equation}
\begin{aligned} \label{eq:regret_sum_cont}
    \regret(K) &\le S^2\kappa_{\mathrm{uc}} L R \sum_{k=1}^K \int_{\mathcal{A}} \pi_1^k(a_1) \big\| \mathbf{a}^k(a_1) - \mathbf{a}(a_1) \big\|_1 \mu(\dd a_1) \\
    &\quad + S^2\kappa_{\mathrm{uc}} L R  \sum_{l=1}^{L-1} \sum_{k=1}^K \int_{(\mathcal{A} \times [O])^l \times \mathcal{A}} \pi^k(\tau_{l+1}) \big\| \big(\mathbf{A}_l^k - \mathbf{A}_l\big) \mathbf{a}_l \big\|_1 \dd \nu_{l+1}.
\end{aligned}
\end{equation}
Now we apply Lemma~\ref{lem:first-action-cont} to bound the first-action contribution by $\widetilde{\mathcal{O}}(S^2 \sqrt{\beta K})$. 
Similarly, applying Lemma~\ref{lem:oom_error_continuous} to each of the $L-1$ remaining terms yields a bound of $\widetilde{\mathcal{O}}(S^6 \kappa_{\mathrm{uc}} \sqrt{\beta K})$ per step.

Substituting these bounds back into~\eqref{eq:regret_sum_cont}:
\begin{align}
    \regret(K) \leq S^2\kappa_{\mathrm{uc}} L R \left[ \widetilde{\mathcal{O}}\big(S^2 \sqrt{\beta K}\big) + (L-1) \widetilde{\mathcal{O}}\big(S^6 \kappa_{\mathrm{uc}} \sqrt{\beta K}\big) \right]. 
\end{align}
This completes the proof.
\end{proof}

\section{Information-Theoretic Lower Bounds}
\label{sec:lower_bounds}

In this section, we establish information-theoretic lower bounds for our input-output QHMM learning framework. We prove that the $\widetilde{\mathcal{O}}(\sqrt{K})$ scaling achieved by OMLE is statistically optimal in the number of episodes, and that the polynomial dependence on the robustness constant $\kappa_{\mathrm{uc}}$ is unavoidable. We achieve this by reducing the problem to the Multi-Armed Quantum Bandit (MAQB) framework~\cite{lumbreras2022multi} and by embedding classical POMDP hardness results.

While an $\Omega(\sqrt{K})$ regret lower bound could be trivially obtained by restricting the latent memory to diagonal states and embedding a classical Multi-Armed Bandit, such a reduction bypasses the quantum nature of our model class. Specifically, for a fully quantum latent memory ($S \ge 2$), the undercomplete assumption is satisfied when the measurement instruments are informationally complete to guarantee the existence of a valid recovery map. To demonstrate that the $\widetilde{\mathcal{O}}(\sqrt{K})$ regret scaling remains strictly optimal even under these quantum requirements, we instead introduce a reduction to the MAQB problem explicitly utilizing SIC-POVMs.

\subsection{The MAQB Framework and SIC-POVMs}

The MAQB problem is a sequential decision-making model where a learner interacts with an unknown, stationary quantum state $\rho$. At each round $t$, the learner selects an action corresponding to an observable $O_{A_t}$ from a predefined action set $\mathcal{A}$. The agent performs this measurement on the state $\rho$ and receives a classical reward $X_t$. The probability distribution of this reward is governed by Born's rule, meaning the expected reward for choosing an observable $O_i$ is simply its expectation value, $\Tr(\rho O_i)$. The same formulation also applies for POVMs instead of an action set, provided we assign a reward to each POVM element such that the expected reward is given by a linear combination of those elements.

The learner's objective is to maximize their cumulative reward over $N$ rounds. Performance is measured by the expected cumulative regret, which quantifies the difference between the reward obtained by always playing the optimal observable and the actual rewards gathered by the learner's policy $\pi$. Mathematically, this is defined as:
\begin{equation}
    \regret(N,\rho) = \sum_{t=1}^N \Big( \max_{O_i \in \mathcal{A}} \Tr(\rho O_i) - \mathbb{E}_{\rho, \pi}[\Tr(\rho O_{A_t})]\Big),
\end{equation}
where $\mathbb{E}_{\rho, \pi}$ is the expectation with respect to the probability distribution of actions and rewards.

To systematically analyze this regret, it is convenient to rewrite it in terms of sub-optimality gaps. We define the gap for any action $i$ as 
\begin{align}
    \Delta_i = \max_{O_j \in \mathcal{A}} \Tr(\rho O_j) - \Tr(\rho O_i),
\end{align}
which represents the expected reward lost by pulling arm $i$ instead of the optimal arm. Letting $T_i(N) = \sum_{t=1}^N \mathbbm{1}\{A_t = i\}$ denote the random variable representing the number of times action $i$ is selected over $N$ rounds, the expected regret can be equivalently expressed as:
\begin{equation}\label{eq:regret_maqb_actions}
    \regret(N,\rho) = \sum_{i \in \mathcal{A}} \Delta_i \mathbb{E}_{\rho, \pi}[T_i(N)].
\end{equation}

To seamlessly embed the MAQB framework into our QHMM setting, we must satisfy the undercomplete assumption (Assumption~\ref{ass:undercomplete}). Standard MAQB lower bounds --- which typically rely on 2-outcome POVMs~\cite{lumbreras2022multi} --- cannot be directly applied. Those standard measurements generally violate our assumption by irrecoverably tracing out latent degrees of freedom for systems where $S \ge 2$.

To construct a strictly valid undercomplete environment based on the MAQB setting, we design instruments based on SIC-POVMs. This guarantees the invertibility required for the recovery map and ensures a uniform robustness constant $\kappa_{\mathrm{uc}}$, allowing us to map the MAQB lower bound directly into our QHMM setting.

Before presenting the bounds, we recall a standard tool from information theory used to bound the probability of distinguishing two probability measures, alongside the divergence decomposition lemma~\cite[Chapter 15]{lattimore2020bandit}.

\begin{lemma}[Bretagnolle-Huber Inequality]
\label{lem:bretagnolle_huber}
Let $\mathbb{P}$ and $\mathbb{Q}$ be probability measures on the same measurable space $(\Omega, \mathcal{F})$, and let $E \in \mathcal{F}$ be an arbitrary event. Let $E^c = \Omega \setminus E$ be its complement. Then,
\begin{equation}
    \mathbb{P}(E) + \mathbb{Q}(E^c) \ge \frac{1}{2} \exp(-\D(\mathbb{P} \| \mathbb{Q})),
\end{equation}
where $\D(\mathbb{P} \| \mathbb{Q})$ is the Kullback-Leibler (KL) divergence. Furthermore, the KL divergence is upper bounded by the $\chi^2$-divergence: $\D(\mathbb{P} \| \mathbb{Q}) \le \chi^2(\mathbb{P}, \mathbb{Q}) = \sum_x \frac{(\mathbb{P}(x) - \mathbb{Q}(x))^2}{\mathbb{Q}(x)}$.
\end{lemma}

When a policy $\pi$ interacts with an environment over $N$ rounds, the divergence between the trajectory distributions induced by two environments $\rho_1$ and $\rho_l$ decomposes as:
\begin{equation}
\label{eq:div_decomp}
    \D(\mathbb{P}_{\rho_1}^\pi \| \mathbb{P}_{\rho_l}^\pi) = \sum_{a \in \mathcal{A}} \mathbb{E}_{\rho_1, \pi}[T_a(N)] \D(P_{\rho_1}^{(a)} \| P_{\rho_l}^{(a)}),
\end{equation}
where $T_a(N)$ is the number of times action $a$ is selected, and $P_\rho^{(a)}$ is the marginal outcome distribution for action $a$.

\subsection{Regret Lower Bound via MAQB with SIC-POVMs}

We first prove a lower bound for a multi-armed quantum bandit where the actions correspond to a SIC-POVM. Note that for $S=2$, this corresponds exactly to the special case we studied in Remark~\ref{remark:sic_povm_example}. The technique we use to derive the lower bounds is adapted from~\cite{lumbreras2022multi} to fit our specific setting.

\begin{proposition}[MAQB Lower Bound]
\label{prop:sic_povm_bandit}
Consider a MAQB problem over $N$ rounds where the learner interacts with an unknown state $\rho \in \mathsf{D}(\mathcal{H}_2)$. The action set is $\mathcal{A} = [4]$. For each action $a$, the learner performs a measurement corresponding to the full qubit SIC-POVM and receives a classical outcome $o \in [4]$ drawn from the distribution $\Pr(o|a) = \Tr(M_{\pi_a(o)} \rho)$, where $\pi_a$ is the transposition permutation swapping indices $1$ and $a$. While the measurement yields $4$ possible outcomes, the reward is evaluated as $1$ if $o=1$ and $0$ otherwise. For any adaptive policy $\pi$, there exists an environment state $\rho$ such that the expected cumulative regret over $N$ rounds is bounded as:
\begin{equation}
    \regret(N,\rho) = \Omega\big( \sqrt{N} \big).
\end{equation}
\end{proposition}

\begin{proof}
\textbf{Step 1: Construction of candidate environments.} Fix a policy $\pi$ and let $\Delta \in (0, 1/6]$ be a parameter to be determined later. We define two valid candidate target density matrices:
\begin{equation}
    \rho_1 = \frac{1-\Delta}{2} \mathbb{I} + \Delta E_1, \quad \quad \rho_l = \frac{1-3\Delta}{2} \mathbb{I} + \Delta E_1 + 2\Delta E_l,
\end{equation}
where $l = \arg\min_{j>1} \mathbb{E}_{\rho_1, \pi}[T_j(N)]$ is the arm pulled least often in expectation among the suboptimal arms. Both states are positive semi-definite since $\Delta \le 1/6$.

The expected reward for action $i$ under state $\rho$ is $\mu_i(\rho) = \Pr(1|i) = \Tr(M_{\pi_i(1)} \rho) = \frac{1}{2}\Tr(E_i \rho)$. 
For $\rho_1$, we have $\mu_1(\rho_1) =  \frac{1+\Delta}{4}$, and for $i \neq 1$, $\mu_i(\rho_1) =  \frac{3-\Delta}{12}$. The optimal arm is $i=1$, yielding a uniform sub-optimality gap for any $i \neq 1$ of $\mu_1(\rho_1) - \mu_i(\rho_1) = \frac{\Delta}{3} = \Delta_*$. For $\rho_l$, we have $\mu_l(\rho_l) = \frac{3+5\Delta}{12}$ and $\mu_1(\rho_l) = \frac{3+\Delta}{12}$. The optimal arm is exactly $i=l$. The gap for playing arm $1$ instead of $l$ is identical: $\mu_l(\rho_l) - \mu_1(\rho_l) = \frac{\Delta}{3} = \Delta_*$.

\textbf{Step 2: Bounding the KL-divergence.} We now bound the KL divergence between the trajectory distributions $\mathbb{P}_{\rho_1}^\pi$ and $\mathbb{P}_{\rho_l}^\pi$. By the decomposition in~\eqref{eq:div_decomp}, the divergence sums over the actions. Crucially, because the outcomes of action $a$ are simply a permutation of the fixed SIC-POVM, the unordered set of outcome probabilities is identical for every action. Consequently, the KL divergence is strictly independent of $a$. We bound it using the $\chi^2$-divergence on the unpermuted probabilities:
\begin{equation}
    \D(P_{\rho_1}^{(a)} \| P_{\rho_l}^{(a)}) \le \chi^2(P_{\rho_1}, P_{\rho_l})= \sum_{x=1}^{4} \frac{(\Tr(\rho_1 M_x) - \Tr(\rho_l M_x))^2}{\Tr(\rho_l M_x)}.
\end{equation}
The difference in probabilities is $\Tr\left( M_x (\Delta \mathbb{I} - 2\Delta E_l) \right) = \Delta \left( \frac{1}{2} - \Tr(E_x E_l) \right)$. Substituting $\Tr(E_x E_l) = \frac{2\delta_{xl}+1}{3}$, the differences evaluate to $-\frac{\Delta}{2}$ for $x=l$, and $\frac{\Delta}{6}$ for $x \neq l$. Summing the squares yields:
\begin{equation}
    \sum_{x=1}^{4} (P_{\rho_1}(x) - P_{\rho_l}(x))^2 = \left(-\frac{\Delta}{2}\right)^2 + 3\left(\frac{\Delta}{6}\right)^2 =  \frac{\Delta^2}{3}.
\end{equation}
We bound the denominator using $\rho_l \ge \frac{1-3\Delta}{2} \mathbb{I}$, which implies $P_{\rho_l}(x) \ge \frac{1-3\Delta}{4}$. Thus, the $\chi^2$-divergence is bounded by:
\begin{equation}
    \chi^2 \le \frac{4}{1-3\Delta} \frac{\Delta^2}{3} = \frac{4\Delta^2}{3(1-3\Delta)} \le \frac{8\Delta^2}{3},
\end{equation}
where we used $1-3\Delta \ge 1/2$ (since $\Delta \le 1/6$). Because this bound is independent of the policy's action choices, we can directly use the decomposition~\eqref{eq:div_decomp} to show that the total divergence is bounded uniformly by: 
\begin{align}
\D(\mathbb{P}_{\rho_1}^\pi \| \mathbb{P}_{\rho_l}^\pi) \le \sum_a \mathbb{E}[T_a(N)] \frac{8\Delta^2}{3} = \frac{8N\Delta^2}{3}.
\end{align}

\textbf{Step 3: Minimax regret bound.} Using the expression of the regret~\eqref{eq:regret_maqb_actions} and applying Markov's inequality, we establish: 
\begin{align}
    \regret(N, \rho_1) \ge \Delta_* \mathbb{E}_{\rho_1}[N - T_1] \ge \frac{N \Delta_*}{2} \mathbb{P}_{\rho_1}(T_1 \le N/2).
\end{align}
By the definition of $l$, we also have: 
\begin{align}
    \regret(N, \rho_l) \ge \Delta_* \mathbb{E}_{\rho_l}[T_1] \ge \frac{N \Delta_*}{2} \mathbb{P}_{\rho_l}(T_1 > N/2).
\end{align}
Summing these bounds and applying Lemma~\ref{lem:bretagnolle_huber}, we obtain:
\begin{equation}
   \regret (N,\rho_1) + \regret (N,\rho_l) \ge \frac{N \Delta_*}{4} \exp\left(-\frac{8N\Delta^2}{3}\right).
\end{equation}
Setting $\Delta = \sqrt{\frac{3}{8N}}$ yields $\exp(-1)$. Substituting $\Delta_* = \frac{\Delta}{3} = \frac{1}{\sqrt{24N}}$, we conclude:
\begin{equation}
      \regret (N,\rho_1) + \regret (N,\rho_l)\ge \frac{N}{4\sqrt{24N}} e^{-1} = \Omega\big(\sqrt{N}\big).
\end{equation}
Since the maximum regret over the two hypotheses must be at least half of their sum, the lower bound follows.
\end{proof}

\subsection{Lower Bound for QHMMs via Reduction to MAQB}

We now formally map the MAQB problem into our QHMM framework. Because our instruments measure the full SIC-POVM, the recovery map exists and the robustness constant $\kappa_{\mathrm{uc}}$ is exactly $1$.

\begin{theorem}[Regret Lower Bound for QHMM]
\label{thm:qhmm_lower_bound}
Fix the hidden memory dimension $S = 2$ and the episode length $L \ge 1$. Let the action space $\mathcal{A}$ be discrete with size $A = 4$. For any learning policy $\pi$, there exists an $L$-step QHMM environment $\omega_L \in \Omega_{\mathrm{uc}}$ with $O=4$ outcomes and robustness constant $\kappa_{\mathrm{uc}} = 1$ such that the expected cumulative regret over $K$ episodes is lower bounded by:
\begin{equation}
    \regret (K) = \Omega\big(\sqrt{KL}\big).
\end{equation}
\end{theorem}

\begin{proof}
\textbf{Step 1: Environment construction.} Let $\rho^*$ be the target environment state from Proposition~\ref{prop:sic_povm_bandit}. We define the QHMM environment $\omega$ as follows. The initial memory state is $\rho_1 = \rho^*$. The inter-round memory evolution is governed by the completely replacing channel $\mathbb{E}_l(X) = \Tr(X) \rho^*$ for all $l \in [L-1]$. Crucially, because this CPTP map completely depolarizes the memory, any measurement-induced disturbance is immediately erased. Consequently, the normalized filtered memory state at the beginning of \emph{every} step is strictly identical to be $\rho^*$.

For each action $a \in \mathcal{A} = [4]$, we define the instrument $\mathcal{P}^{(a)}$ with $O=4$ where the maps are defined as
\begin{equation}
    \Phi_o^{(a)}(X) = \Tr(M_{\pi_a(o)} X) \rho_0, \quad o \in [4],
\end{equation}
where $M_x$ is the qubit SIC-POVM, $\pi_a$ is the transposition permutation from Proposition~\ref{prop:sic_povm_bandit}, and $\rho_0$ is a fixed, arbitrary pure state (e.g., $\proj{1}$). Summing over all outcomes gives $\sum_o \Phi_o^{(a)}(X) = \Tr(X) \rho_0$, confirming the instrument is valid and strictly CPTP.

\textbf{Step 2: Verifying the undercomplete assumption.} We must verify that this environment satisfies the undercomplete assumption. The classical observable operator obtained by tracing out the system is:
\begin{equation}
    \Tr_S[\mathcal{P}^{(a)}(X)] = \sum_{o = 1}^{4} \Tr(M_{\pi_a(o)} X) \proj{o}.
\end{equation}
Because the instrument destroys the pre-measurement state $X$ and prepares the fixed state $\rho_0$ for every outcome, the quantum output can be perfectly reconstructed solely from the classical labels. We define the linear recovery map $R^{(a)}: \mathbb{M}_{O} \to \mathbb{M}_{S} \otimes \mathbb{M}_{O} $ simply as:
\begin{equation}
    R^{(a)}(D) = \rho_0 \otimes D,
\end{equation}
for any diagonal matrix $D = \sum_o d_o \proj{o}$. Then, leveraging the informational completeness of SIC-POVMs, applying the recovery map to the observable operator perfectly recovers $\mathcal{P}^{(a)}(X)$. The $1 \to 1$ diagonal norm is precisely:
\begin{equation}
    \|R^{(a)}\|_{1 \to 1, \text{diag}} = \sup_{\|D\|_1 \le 1} \|\rho_0 \otimes D\|_1 = \sup_{\|D\|_1 \le 1} \|\rho_0\|_1 \|D\|_1 = 1.
\end{equation}
Hence, $\kappa_{\mathrm{uc}} = 1$, and the environment strictly belongs to the class $\Omega_{\mathrm{uc}}$ defined in Definition~\ref{def:Omega_uc}.

\textbf{Step 3: Reduction to MAQB.} Let the reward function be defined as $r_l(1) = 1$, and $r_l(o) = 0$ for $o \neq 1$, identically for all steps $l \in [L]$. Since the state prior to every action is $\rho^*$, the conditional probability of receiving reward $1$ is $\mathrm{Pr}(1 | a_l,\tau_{l-1}) = \Tr(M_{\pi_{a_l}(1)} \rho^*) = \Tr(M_{a_l} \rho^*)$. Because the hidden state resets independently at every step, this interaction is mathematically isomorphic to executing a Multi-Armed Quantum Bandit policy against $\rho^*$ for $N = K \times L$ rounds. By Proposition~\ref{prop:sic_povm_bandit}, the cumulative regret incurred by any policy is bounded below by $\Omega( \sqrt{N} ) = \Omega(\sqrt{KL})$.
\end{proof}

\subsection{Lower Bound for the Observability Parameter via Classical POMDPs}\label{sec:reduction_pomdp_classical}

While the MAQB reduction establishes the $\Omega(\sqrt{K})$ scaling, it relies on an environment where measurements are perfectly informative ($\kappa_{\mathrm{uc}}=1$). To show that the polynomial dependence on the undercomplete robustness constant $\kappa_{\mathrm{uc}}$ is strictly necessary, we demonstrate that our input-output QHMM framework fully subsumes classical POMDPs. Consequently, the statistical hardness results for classical POMDPs directly apply to our quantum setting.

To rigorously state this classical result, we briefly introduce the standard classical POMDP notation that was used in~\cite{liu2022partially}. Let $H$ denote the episode horizon (which directly corresponds to our interaction length $L$), $A$ the number of available actions, $S$ the number of latent states, and $O$ the number of possible observations. Furthermore, let $\mathbb{O}_h \in \mathbb{R}^{O \times S}$ be the column-stochastic emission matrix at step $h \in [H]$, where the entry $[\mathbb{O}_h]_{o, s}$ specifies the probability of observing outcome $o$ given the latent state $s$. Finally, let $\sigma_S(\mathbb{O}_h)$ denote the $S$-th singular value of this matrix. We can now state the lower bound established for classical undercomplete POMDPs using the ``combinatorial lock'' construction~\cite{jin2020sample,liu2022partially}.

\begin{theorem}[Classical POMDP Lower Bound, Theorem 6 in \cite{liu2022partially}]
\label{thm:classical_lock}
For any observability parameter $\alpha \in (0, 1/2]$ and integers $H, A \ge 2$, there exists a classical POMDP with $S=2$ states and $O=3$ observations satisfying the classical $\alpha$-weakly revealing condition ($\min_h \sigma_S(\mathbb{O}_h) \ge \alpha$), such that any algorithm requires at least $\Omega\left( \min \left\{ \frac{1}{\alpha H}, A^{H-1} \right\} \right)$ episodes to learn a $(1/2)$-optimal policy (i.e., a policy $\pi$ whose expected value satisfies $V^* - V^\pi \le 1/2$).
\end{theorem}

We now provide a corollary formally embedding classical POMDPs into our QHMM framework, proving that the classical structural hardness maps directly to our robustness constant $\kappa_{\mathrm{uc}}$.

\begin{corollary}[Necessity of $\kappa_{\mathrm{uc}}$ scaling in QHMMs]
\label{cor:qhmm_lock}
For any $\alpha \in (0, 1/2]$, discrete action space of size $A \ge 2$, and episode length $L \ge 2$, there exists a family of undercomplete QHMM environments $\omega \in \Omega_{\mathrm{uc}}$ with memory dimension $S = 2$, outcome dimension $O=3$, and robustness constant $\kappa_{\mathrm{uc}} \le 1/\alpha$, for which any learning policy incurs an expected cumulative regret scaling linearly as $\Omega(K)$ for $K \le \mathcal{O}\left( \min \left\{ \frac{1}{\alpha L}, A^{L-1} \right\} \right)$.
\end{corollary}

\begin{proof}
\textbf{Step 1: Classical POMDP embedding.} Consider a classical POMDP with states $\mathcal{S}$, actions $\mathcal{A}$, observations $\mathcal{O}$, transition matrices $T(s'|s, a)$, and emission matrices $\mathbb{O}(o|s)$. The classical $\alpha$-weakly revealing condition guarantees that the emission matrix $\mathbb{O}$ has full column rank and possesses a left-inverse $\mathbb{O}^\dagger \in \mathbb{R}^{S \times O}$ such that $\|\mathbb{O}^\dagger\|_{1 \to 1} \le 1/\alpha$. 

We embed this POMDP into a QHMM by defining the initial state as a diagonal density matrix $\rho_1 = \sum_s \Pr(s_1=s) \proj{s}$. To simultaneously capture the classical emission and state transition within our action-dependent instrument, we define $\mathcal{P}^{(a)} = \{\Phi_o^{(a)}\}_{o \in [O]}$ using completely positive measure-and-prepare branch maps:
\begin{equation}
    \Phi_o^{(a)}(X) = \sum_{s \in [S]} \mathbb{O}(o|s) \langle s | X | s \rangle \Big( \sum_{s' \in  [S]} T(s'|s, a) \proj{s'} \Big).
\end{equation}
Summing over all outcomes $o$ gives $\sum_o \Phi_o^{(a)}(X) = \sum_{s, s'} T(s'|s, a) \langle s | X | s \rangle \proj{s'}$. Since $\sum_{s'} T(s'|s, a) = 1$, the map is trace-preserving on diagonal inputs, making $\mathcal{P}^{(a)}$ a valid CPTP instrument. Because the instrument inherently handles the transition, the internal environment channel is simply the identity: $\mathbb{E}_l(X) = X$. By construction, if the system is initialized in a diagonal state, it will remain diagonal for all steps, perfectly mimicking the classical POMDP trajectory distribution.

\textbf{Step 2: Recovery map and robustness.} We construct the linear recovery map $R^{(a)}: \mathbb{M}_{O} \to \mathbb{M}_{S\times O }$ for any diagonal classical matrix $D = \sum_o d_o \proj{o}$ as:
\begin{equation}
    R^{(a)}(D) = \mathcal{P}^{(a)}\Big( \sum_{s \in \mathcal{S}} (\mathbb{O}^\dagger d)_s \proj{s}\Big),
\end{equation}
where $d$ is the vector of diagonal elements $d_o$. If $D = \Tr_S[\mathcal{P}^{(a)}(X)]$, then $d = \mathbb{O} x$, where $x_s = \langle s|X|s\rangle$. Applying the left-inverse gives $\mathbb{O}^\dagger d = \mathbb{O}^\dagger \mathbb{O} x = x$. Substituting this back yields exactly $R^{(a)}(D) = \mathcal{P}^{(a)}(X)$. The $1 \to 1$ diagonal robustness constant is bounded by:
\begin{equation}
    \|R^{(a)}(D)\|_1 \le \Big\| \sum_{s \in [S]} (\mathbb{O}^\dagger d)_s \proj{s} \Big\|_1 = \|\mathbb{O}^\dagger d\|_1 \le \|\mathbb{O}^\dagger\|_{1 \to 1} \|d\|_1 \le \frac{1}{\alpha} \|D\|_1.
\end{equation}

\textbf{Step 3: Translating PAC Hardness to Cumulative Regret.} As $\|R^{(a)}(D)\|_1 \le  \frac{1}{\alpha} \|D\|_1$, the environment satisfies Assumption~\ref{ass:undercomplete} with $\kappa_{\mathrm{uc}} \le 1/\alpha$. Because this QHMM operates strictly identically to the classical combinatorial lock of Theorem~\ref{thm:classical_lock}, the classical learning limits apply. Theorem~\ref{thm:classical_lock} establishes that any algorithm requires at least $K_0 = \Omega\big(\min\{\frac{1}{\alpha L}, A^{L-1}\}\big)$ episodes to output a $(1/2)$-optimal policy. Consequently, for any horizon $K \le K_0$, the learner's policy $\pi^k$ must be strictly suboptimal by at least a constant margin of $1/2$ for a constant fraction of the episodes. Summing this $\Omega(1)$ expected suboptimality over $K$ episodes yields a cumulative regret bounded below by $\Omega(K)$, strictly validating the stated lower bound.
\end{proof}

\begin{remark}[Consistency with the $\widetilde{\mathcal{O}}(\sqrt{K})$ Upper Bound]
The linear $\Omega(K)$ lower bound established in Corollary~\ref{cor:qhmm_lock} does not contradict our $\widetilde{\mathcal{O}}(\sqrt{K})$ upper bounds. This is because the linear lower bound only applies to an initial phase ($K \le K_0$). During these early episodes, the algorithm is essentially guessing because it has not yet gathered enough data to identify the environment. 

For these hard instances, the robustness constant $\kappa_{\mathrm{uc}}$ scales as $1/\alpha$, which makes the constant prefactor in our $\widetilde{\mathcal{O}}(\sqrt{K})$ upper bound very large. As a result, when $K \le K_0$, our $\sqrt{K}$ upper bound actually evaluates to a number greater than the maximum possible regret $K \times L$. This means the upper bound is mathematically correct, but trivially loose during the initial phase. Once the number of episodes $K$ exceeds $K_0$, the algorithm successfully learns the latent dynamics and the $\sqrt{K}$ sublinear scaling becomes the meaningful, dominating bound.
\end{remark}

\section{Application: State-agnostic quantum work extraction}\label{sec:work_extraction}

Having established the rigorous $\widetilde{\mathcal{O}}(\sqrt{K})$ learning guarantees for the input-output QHMM framework, we now map this abstract setting to a fundamental problem in quantum thermodynamics: state-agnostic work extraction. When extracting free energy from a sequence of non-i.i.d. quantum states, the temporal correlations between the states act as a hidden memory. Because the agent lacks full prior knowledge of this emission source, any sub-optimal extraction protocol inevitably leads to thermodynamic dissipation. We will demonstrate that by modeling this physical interaction as a continuous-action (or many discrete actions) QHMM, the mathematical regret incurred by the agent exactly quantifies this irreversible energy dissipation.

\subsection{Work extraction from sequential-emission HMM environments}
\label{subsec:classical-memory-emission}

We consider sequential-emission HMM environments defined in Section~\ref{sec:sequential_HMM}. The hidden memory is operationally classical, encoded as a diagonal quantum state on a two-dimensional register $M$ with computational basis $\{\ket{1},\ket{2}\}$. Thus, the memory state is restricted to the diagonal subclass
\begin{align}
\label{eq:diag_memory}
\rho_S = x_1 E_1 + x_2 E_2, \qquad E_m := \ket{m}\!\bra{m}, \qquad x_1,x_2 \ge 0, \qquad x_1+x_2=1.
\end{align}
At each time step $l \in [L]$, the internal memory undergoes a classical transition governed by the probability $T(m'| m)$, which corresponds to a completely positive trace-preserving (CPTP) channel $\mathbb{E}(\ket{m}\!\bra{m}) = \sum_{m'}T(m'| m)\ket{m'}\!\bra{m'}$. Conditional on the current memory value $m \in \{1,2\}$, the environment emits a known quantum state $\sigma_m$ on an auxiliary system $Q$. 

The agent probes the system $Q$ to extract work. To extract work, the agent intervenes on the system by selecting a continuous action $a_l \in \mathcal{A}$ as was defined in Sec.~\ref{sec:QHMM}. In our framework, this continuous action directly parameterizes a POVM instrument (as studied in Section~\ref{sec:cont-actions}) that configures the physical work-extraction protocol by specifying a target density matrix $\rho_{a_l}$ that the agent tailors its protocol for:
\begin{align}
\label{eq:action_state}
\rho_{a_l} = \lambda_{a_l} \ket{\psi_{a_l}}\!\bra{\psi_{a_l}} + (1-\lambda_{a_l}) \ket{\psi_{a_l}^\perp}\!\bra{\psi_{a_l}^\perp}.
\end{align}
Therefore, an action $a_l$ fundamentally encapsulates two parameters:
\begin{enumerate}
    \item A measurement basis $\{\ket{\psi_{a_l}}, \ket{\psi_{a_l}^\perp}\}$, which defines a binary projective measurement $M_{0}^{(a_l)} = \ket{\psi_{a_l}}\!\bra{\psi_{a_l}}$ and $M_{1}^{(a_l)} = \ket{\psi_{a_l}^\perp}\!\bra{\psi_{a_l}^\perp}$ on $Q$.
    \item A continuous target purity parameter $\lambda_{a_l} \in (0,1)$, reflecting how ``mixed'' the agent expects the state to be. We will show later in Section~\ref{sec:opt_action} that the optimal $\lambda_{a_l}$ can be chosen a priori. 
\end{enumerate}

The physical extraction protocol yields a discrete classical outcome $o_l \in \{0,1\}$, corresponding to the energy exchanged with the battery. Crucially, the probability of these outcomes conditioned on the memory value $m$ perfectly follows the Born rule for a projective measurement in the basis of $\rho_{a_l}$:
\begin{align}
p_{a_l}(o_l| m) := \Tr\!\big(M_{o_l}^{(a_l)}\sigma_m\big).
\label{eq:classical-memory-born-rule}
\end{align}
This sequential interaction maps perfectly to the quantum instruments defined in our QHMM framework. For each outcome $o\in\{0,1\}$, the corresponding branch map acting on the memory is
\begin{align}
    \Phi_o^{(a)}(X) := \sum_{m=1}^2 \Tr(E_m X)\, p_a(o| m)\, E_m,
\label{eq:classical-memory-branch-map}
\end{align}
yielding the full instrument
\begin{align}
\mathcal P^{(a)}(X) &:= \sum_{o=0}^1  \Phi_o^{(a)}(X)\otimes \ket{o}\!\bra{o} = \sum_{m=1}^2\sum_{o=0}^1 \Tr(E_m X)\, p_a(o| m)\, E_m\otimes \ket{o}\!\bra{o}.
\label{eq:classical-memory-full-instrument}
\end{align}
In particular, if the hidden memory state is \(\rho_S=x_1E_1+x_2E_2\), then
\begin{align}
\mathcal P^{(a)}(\rho_S)
=
\sum_{m=1}^2\sum_{o=0}^1
x_m\, p_a(o| m)\, E_m\otimes \ket{o}\!\bra{o}.
\label{eq:classical-memory-state-output}
\end{align}
Tracing out the hidden memory register yields the observable outcome distribution on the classical register:
\begin{align}
\Tr_S\!\big[\mathcal P^{(a)}(X)\big] = \sum_{o=0}^1 \Big(\sum_{m=1}^2 \Tr(E_m X)\, p_a(o| m)\Big)\ket{o}\!\bra{o}.
\label{eq:classical-memory-observable-operator}
\end{align}
Hence, for a memory state $\rho_S=x_1E_1+x_2E_2$, the observed outcome distribution is exactly the mixture law
\begin{align}
\Pr(o| a,\rho_S) = x_1\,p_a(o| 1)+x_2\,p_a(o| 2).
\label{eq:classical-memory-mixture-law}
\end{align}

\paragraph{Recovery map and undercompleteness.} 
To guarantee the identifiability of this process and verify the undercompleteness assumption, we define the observation matrix associated with action $a$ by
\begin{align}
\mathbb O^{(a)} := 
\begin{pmatrix}
p_a(0| 1) & p_a(0| 2) \\
p_a(1| 1) & p_a(1| 2)
\end{pmatrix}.
\label{eq:classical-memory-observation-matrix}
\end{align}
If \(\mathbb O^{(a)}\) is invertible, then the hidden memory weights can be reconstructed from the observable operator. Indeed, for a diagonal operator on the outcome register, \(D = d_0 \ket{0}\!\bra{0} + d_1 \ket{1}\!\bra{1}\), define
\begin{align}
\hat x(D)
:=
\begin{pmatrix}
\hat x_1(D)
\\
\hat x_2(D)
\end{pmatrix}
:=
\big(\mathbb O^{(a)}\big)^{-1}
\begin{pmatrix}
d_0
\\
d_1
\end{pmatrix}.
\label{eq:classical-memory-recovered-weights}
\end{align}
Then a recovery map is given by
\begin{align}
R^{(a)}(D)
:=
\sum_{m=1}^2
\hat x_m(D)\,
E_m\otimes
\Big(
\sum_{o=0}^1 p_a(o| m)\ket{o}\!\bra{o}
\Big).
\label{eq:classical-memory-recovery-map}
\end{align}
For every diagonal hidden memory operator \(X\), one then has
\begin{align}
R^{(a)}\!\Big(\Tr_S\!\big[\mathcal P^{(a)}(X)\big]\Big)
=
\mathcal P^{(a)}(X).
\label{eq:classical-memory-recovery-identity}
\end{align}
Therefore this model satisfies the same recovery property as in the undercomplete setting, restricted to the diagonal memory subspace, whenever \(\mathbb O^{(a)}\) is invertible.

Let $q_m^{(a)} := \Tr\!\big(M_1^{(a)}\sigma_m\big)$. Then $p_a(1| m)=q_m^{(a)}$ and $p_a(0| m)=1-q_m^{(a)}$, which yields
\begin{align}
\mathbb O^{(a)} = 
\begin{pmatrix}
1-q_1^{(a)} & 1-q_2^{(a)} \\
q_1^{(a)} & q_2^{(a)}
\end{pmatrix}.
\label{eq:classical-memory-binary-observation}
\end{align}
This matrix is invertible if and only if $q_1^{(a)} \neq q_2^{(a)}$. Because our framework supports continuous action spaces, the agent is not locked into pathological discrete measurements; it can dynamically select informationally complete bases to ensure this invertibility condition is strictly met. Equivalently, the recovery map exists precisely when the chosen measurement produces different binary outcome distributions for the two emitted states $\sigma_1$ and $\sigma_2$. Because our continuous action space $\mathcal{A}$ allows the agent to freely select the measurement basis, the agent can dynamically select informationally complete instruments that satisfy this condition. Since the emitted states are known (but we do not know when are they going to be emmited), one can guarantee that the environment belongs to the undercomplete class with a bounded recovery robustness $\kappa_{\mathrm{uc}}$.

\subsection{Work-extraction and action-dependent rewards}

The sequential emission dictated by the classical memory creates a series of quantum states exhibiting temporal correlations, described by the multi-time state:
\begin{equation}
\label{eq:multi_time_state}
    \rho_{Q_{1:L}}\coloneqq \sum_{m_{1:L}}\Pr(M_{1:L}=m_{1:L})\bigotimes _{l=1}^L\sigma_{m_l},
\end{equation}
where the probability $\Pr(M_{1:L}=m_{1:L})$ is defined as
\begin{equation}
    \Pr(M_{1:L}=m_{1:L}) \coloneqq x_{m_1}\prod_{l=1}^{L-1} T(m_{l+1}| m_l).
\end{equation}

For a given ambient temperature, we denote its inverse temperature as $\beta\coloneqq (k_BT)^{-1}$ to avoid confusion with transition matrices. System $Q$ is assumed to have a degenerate Hamiltonian $H_Q$, the thermal equilibrium state is the maximally mixed state $\gamma = \mathbb{I}/2$. This state is defined to have zero non-equilibrium free energy. Accordingly, any state $\rho \neq \gamma$ possesses non-equilibrium free energy ~\cite{brandao2013resource}
\begin{equation}
    \mathcal{F}_{\text{noneq}} (\rho)\coloneqq \beta^{-1}\text{D}(\rho\|\gamma)~.
\end{equation}
The goal of the agent is to transfer this non-equilibrium free energy from this multi-time state into a battery with the aid of some thermal reservoir. To be more specific, given a multi-time state $\rho_{Q_{1:L}}$, a thermal reservoir and a battery, the agent with the ability to act coherently across time steps would be able to implement a huge operation $\mathcal{W_{\text{coh}}}$ across the joint state and transfer the non-equilibrium free energy to the energy of the battery. This increase in battery's energy is what we define as \emph{work}. An illustration of this can be found in Fig.~\ref{fig:work_circuit}.

\begin{figure}[htbp]
    \centering
    \includegraphics[width=0.5\linewidth]{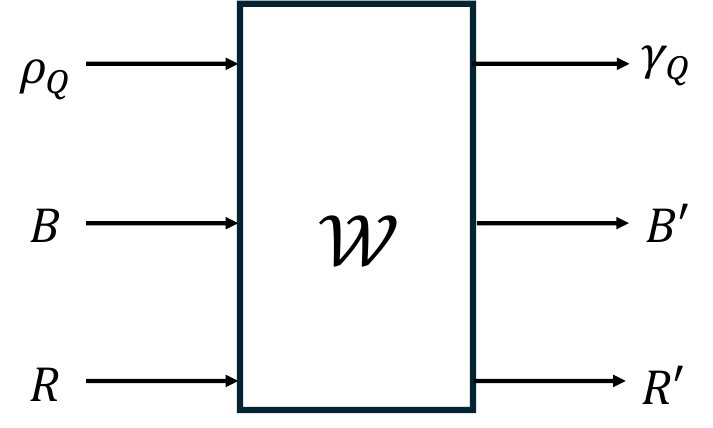}
    \caption{A circuit diagram representation of the work extraction protocol. The system $Q$ represents the system where free energy is drawn, $B$ is a battery and $R$ represents a thermal reservoir as an ancillary system. The protocol aims to transform $\rho_Q$ to a thermal state $\gamma_Q$ with the help of the thermal states from the reservoir; the free energy lost in system $Q$ will be balanced by the increase in energy of the battery $B$.}
    \label{fig:work_circuit}
\end{figure}
Here, $\D(\rho\|\sigma)=\tr{\rho\ln\rho}-\tr{\rho\ln\sigma}$ here is the quantum relative entropy, a quantum generalization of the classical KL-divergence.
The total non-equilibrium free energy of the sequence, $\mathcal{F}_{\text{noneq}} (\rho_{Q_{1:L}})= \beta^{-1}\text{D}(\rho_{Q_{1:L}}\|\gamma^{\otimes L})$, dictates the maximum possible average work $\langle W\rangle$ extractable into a battery system $B$ by any agent using strictly energy-conserving thermal operations (i.e., operations $\mathcal{U}$ satisfying $[H_Q+H_R+H_B, U]=0$) for a energy degenerate system Hamiltonian~\cite{brandao2013resource,horodecki2013fundamental,skrzypczyk2014work,aaberg2014catalytic,brandao2015second,korzekwa2016extraction}. 

The upper bound on the average work extractable by \emph{any} sequential agent is fundamentally limited by the state correlations:
\begin{equation}
\label{eq:work_upperbound}
    \langle W\rangle \leq \F_{\text{noneq}}(\rho_{Q_{1:L}}) - \beta^{-1}\delta(\overrightarrow{Q_1:Q_L}),
\end{equation}
where the term $\delta(\overrightarrow{Q_1:Q_L})$ is the \emph{causal dissipation}. It quantifies the irreversible disturbance the agent introduces into the system by performing sequential measurements without having full prior knowledge of the future states:
\begin{equation}
    \delta(\overrightarrow{Q_1:Q_L}) \coloneqq \min_{\pi} \left( \mathbb{E}_{\pi, \omega} \left[ \sum_{l=1}^{L-1} H(O_l | a_l, \tau_{l-1}) + S(\rho_{Q_L}(\tau_{L-1})) \right] \right) - S(\rho_{Q_{1:L}})~.
\end{equation}
The Shannon entropy of the classical outcome $H(O_l | a_l, \tau_{l-1})$ is computed directly from the Born rule probabilities of the instruments, and $S(\rho_{Q_L}(\tau_{L-1}))$ is the von Neumann entropy of the final unmeasured subsystem conditioned on the specific classical trajectory observed by the agent.

\paragraph{Action-dependent rewards and bounded action space.}
Under idealized set-ups~\cite{skrzypczyk2014work,huang2023engines}, the physical work extracted at time step $l$, which we denote as the random variable $W_l$, takes discrete values perfectly correlated with the instrument outcome $o_l \in \{0,1\}$ and the action's purity parameter $\lambda_{a_l}$:
\begin{equation}
\begin{split}
    \Pr(W_l=w_{l,0}) &= \Tr\!\big(M_{a}^{(a_l)}\sigma_{m_l}\big), \quad \text{where} \quad w_{l,0} = \beta^{-1}(\ln 2+\ln \lambda_{a_l})\\
    \Pr(W_l=w_{l,1}) &= \Tr\!\big(M_{1}^{(a_l)}\sigma_{m_l}\big), \quad \text{where} \quad w_{l,1} = \beta^{-1} [\ln 2+\ln (1-\lambda_{a_l})]~.
\end{split}
\label{eq:work_distribution}
\end{equation}
The derivation behind this probability distribution can be found in Appendix~\ref{sec:detailed_work_distribution}.
Because the outcome distribution perfectly matches the Born rule probabilities of our instrument~\eqref{eq:classical-memory-born-rule}, the battery measurements act simultaneously as the data collection mechanism for the QHMM and the vessel for energy extraction. 

We can thus formally define our RL reward function directly as the work extracted by setting the reward as 
\begin{align}\label{eq:reward_energy}
r(o_l, a_l) \coloneqq w_{l,o_l}.
\end{align}
Unlike standard POMDP formulations where the reward depends strictly on the observation, here the reward is explicitly parameterized by the continuous action variable $\lambda_{a_l}$. This seamlessly integrates into our OMLE objective: during the optimistic planning phase, the agent simply selects the policy and model parameter that maximizes the value function with respect to this action-dependent reward. 

\paragraph{Bounding the Action-Dependent Reward} One physical caveat of this model is that as the action's purity parameter approaches the extremes $\lambda_{a_l} \to 0$ or $1$, the work values diverge. To prevent the maximum reward constant $R$ from blowing up in our regret bounds, we restrict the action space such that $\lambda_{a_l} \in [\epsilon_K, 1 - \epsilon_K]$, where $\epsilon_K$ is a deterministic depolarization parameter strictly dependent on the time horizon $K$ (e.g., $\epsilon_K \propto 1/K$). This effectively bounds the maximum reward amplitude to $R = \mathcal{O}(k_B T \ln K)$. Since the MLE concentration bounds (Propositions~4.3 and 4.4) only depend on the observation probabilities and not the rewards, this regularization fully preserves our theory while preventing arbitrary divergence. The theoretical cumulative regret scales linearly with $R$, yielding an overall optimal sublinear scaling that introduces at most a logarithmic factor.

\subsection{Work extraction protocol}
Here we provide an overview of the work extraction protocol first proposed in~\cite{skrzypczyk2014work}. The protocol assumes the following subsystems
\begin{enumerate}
    \item Source $Q$, where a state $\rho_Q\neq\gamma$ is the state where the free energy will be supplied.
    \item Thermal Reservoir $R$, fixed at an inverse and  supplies thermal state $\gamma$, it is assumed to have a tunable Hamiltonian $H_R(\nu)=\nu\ket{1}\!\bra{1}$.
    \item Battery $B$, a system modeled after a classical weight moving up and down to store energy. It is assumed to have bi-infinite energy levels and we assume that we are far from the ground state~\cite{lipka2021second,aaberg2013truly,horodecki2013fundamental}. This battery system is equipped with a translation operator $\Gamma_\varepsilon$ which increase its energy level by $\varepsilon$.
\end{enumerate}
The difference in energy level of the battery before and after the protocol is then what we define as \emph{work}.

The protocol is designed for a specific state $\rho^* =\sum_{i}p_i\ket{\phi_i}\!\bra{\phi_i}$ and proceeds in 2 stages. In stage 1, the agent implement a unitary rotation to diagonalize $\rho^*$ in the energy eigenbasis followed by $M$ number of SWAP operations with $M$ different thermal states. The thermal states' purities decreases as the energy gap $\nu$ of $H_R$ converge towards the Hamiltonian of $Q$, $H_Q$. All the energy difference incurred during the rotation and SWAP operation will be added/taken from the battery, $B$. In the quasi-static limit where $M\to\infty$, we recover the work distribution given in~\eqref{eq:work_distribution} for a degenerate source Hamiltonian. This protocol obeys strict energy conservation when the source Hamiltonian is degenerate. 

In an ideal world, where the agent knows the identity of $\rho_Q$, this protocol would on average extract work that is equal to the non-equilibrium free energy $\beta^{-1}\D(\rho_Q\|\gamma)$. However if the agent tailors the protocol to a state that is not the actual state fed into the operation, the average work extracted would then be given by 
\begin{equation}
    \beta\langle W\rangle =\D(\rho_Q\|\gamma)-\D(\rho_Q\|\rho^*)~,
\end{equation}
which is in line with the expression in~\eqref{eq:exp_work}, where $\rho^*$ is the state that the agent tailors the protocol for. The algorithmic representation of the protocol is shown in Alg.~\ref{alg:sc_work_extraction}.

\begin{algorithm}[H]
	\caption{\textsf{$\rho^*$-ideal work extraction}} 
	\label{alg:sc_work_extraction}
 
        Require: An source state $\rho$, an estimate of the source state  $\rho^*$, a battery state $\varphi(x)$ with the battery energy $\mu=0$, a reservoir at inverse temperature $\beta$

        Set $\{\phi_{i}\}_i$ and $\{p_{i}\}_i$ to be the eigenvectors and eigenvalues of $\rho^*$

        \vspace{1mm}
        \textit{Unitary Rotation}
        \vspace{1mm}
        
        Apply unitary $U=\sum_i \ketbra{i}{\phi_i}$ to try to diagonalize the system qubit in computational basis.
        \vspace{1mm}
        
        \For{$l=1,2,\ldots,M$}{
            \vspace{1mm}
            \textit{Prepare a fresh a reservoir qubit and exchange it with the system}
            \vspace{1mm}
            
            Take a fresh thermal qubit $\gamma_\beta(\nu(l))=\frac{1}{Z_R(\nu(l))}e^{-\beta H_R(\nu(l))}$ where $\nu(l)=\beta^{-1}\ln\frac{p_{0,l}}{p_{1,l}}$, $p_{i,l} = p_i-(-1)^i l \delta p$ and $\delta p=\frac{1}{M}(p_0-\frac{1}{2})$ from the reservoir.  
            
            Apply the swap unitary $V_{\rho^*,l}= \sum_{ij}\ketbra{i}{j}_A\otimes\ketbra{j}{i}_R \otimes \Gamma_{(i-j)\nu(l)}$ on the system, the battery and the reservoir qubit. 
            
            Discard the reservoir qubit. 
        }
        
        \vspace{1mm}
        \textit{Measure the extracted work}
        \vspace{1mm}
            
        Measure the battery energy, obtain the battery energy $\mu'$ and compute the extracted work $\Delta W=\mu'-\mu$. 
\end{algorithm}
In the first stage of the protocol, we rotate the unknown qubit via unitary 
\begin{equation}
    U = \sum_i\ketbra{i}{\phi_i}~.
\end{equation}
This operation attempts to diagonalize the system qubit in the computation basis.
We then interact the system with the battery. The state of the system together with the battery is 
\begin{equation}
\label{eq:intial_joint_state}
    \rho_{AB} = \sum_{ij}\bra{\phi_i}\rho\ket{\phi_j}\ketbra{i}{j}_Q\otimes \varphi(x)_B.
\end{equation}

In the second stage of the protocol, we perform $M$ repetitions of the following process. In repetition $l$, we take a fresh thermal qubit $\gamma_\beta(\nu(l,\epsilon))$ with Hamiltonian $H_{R}(\nu(l,\epsilon))$ from the reservoir where $\nu(l,\epsilon)= \beta^{-1} \ln \frac{p_{0,l}}{p_{1,l}}$, $p_{i,l}= p_{i} - (-1)^i l \delta p$ and $\delta p = \frac{1}{M}(p_{0} - \frac{1}{2})$. Note that the reservoir qubit we take depends on which repetition we are in. Then we perform the swap unitary 
\begin{align}
    V_{\rho^*,l}  = \sum_{ij}\ketbra{i}{j}_Q\otimes\ketbra{j}{i}_R \otimes \Gamma_{(i-j)\nu(l,\epsilon)}~.
\end{align}
This unitary swaps the system and the fresh qubit from the reservoir, extracts work into the battery due to the different energy gap between $\{\ket{i}_Q\}_{i=0,1}$ and $\{\ket{i}_R\}_{i=0,1}$ and conserves energy of the system, the qubit from the reservoir and the battery. Finally, the qubit from the reservoir is discarded. At the end of each repetition $l$, the reduced state is 
\begin{align}
    \rho_{AB,l}  = \Tr_R\left(V_{\rho^*,l}\left(\rho_{AB,l-1}\otimes \gamma_\beta(\nu(l,\epsilon))\right) V_{\rho^*,l}^\dagger\right),
\end{align}
After the first repetition, we obtain 
\begin{align}
    \label{eq:joint_state_1}
    \rho_{AB,1} = \sum_{i} p_{i,1} \ketbra{i}{i}_Q \otimes \rho_{B,i,1},
\end{align}
where 
\begin{align}
    \label{eq:battery_state_1}
    \rho_{B,i,1} = \sum_j \bra{\phi_{j}}\rho\ket{\phi_{j}} \varphi(x-(i-j)\nu(l,\epsilon)), 
\end{align}
and after repetition $l$ where $l\geq 2$, we obtain
\begin{align}
    \label{eq:joint_state_tau}
    \rho_{AB,l} = \sum_{i} p_{i,l} \ketbra{i}{i}_Q \otimes \rho_{B,i,l}, 
\end{align}
where
\begin{align}
    \label{eq:battery_state_tau}
    \rho_{B,i,l} = \sum_j  p_{j,l-1} \Gamma_{(i-j)\nu(l,\epsilon)} \rho_{B,j,l-1} \Gamma_{(i-j)\nu(l,\epsilon)}^\dagger. 
\end{align}
From~\eqref{eq:joint_state_tau}, we observe that the reduced state of the system changes gradually, which resembles a quasi-static process in thermodynamics. This is the reason why we take the swap unitary in repetition. We will leave the detailed derivation for the work distribution to Appendix.~\ref{sec:detailed_work_distribution}

\subsubsection{Extracted work for different inputs}

We can now use these results to determine the expected work for any arbitrary input state.

\begin{theorem}
\label{thm:arbitrary_input_work}
Let $\{\phi_i\}_{i=0,1}$ and $\{p_i\}_{i=0,1}$ be the eigenvectors and eigenvalues of the target state $\rho^*$, and let $w_i$ be the fixed work values defined in Theorem \ref{thm:work_convergence}.
\begin{enumerate}
    \item When applying the protocol to any arbitrary state $\rho$, the probability of extracting work $\Delta W = w_i$ is given by:
    \begin{equation}
        \Pr(\Delta W=w_{i})=\langle\phi_{i}|\rho|\phi_{i}\rangle.
    \end{equation}
    \item The expected work extracted when the protocol (optimized for $\rho^*$) is applied to $\rho$ is:
    \begin{equation}
        \mathbb{E}[\Delta W]=\beta^{-1}[\D(\rho\|\mathbb{I}/2)-\D(\rho\|\rho^{*})].
    \end{equation}
    The second term, $\beta^{-1}\D(\rho\|\rho^{*})$, represents the thermodynamic dissipation due to the mismatch between the true environment $\rho$ and the protocol's target state $\rho^*$.
\end{enumerate}
\end{theorem}

\begin{proof}
It worth noting that: 
\begin{enumerate}
    \item The off-diagonal terms of the joint state $\rho_{AB}$ do not affect the reduced state of the battery after the repetitions, see~\eqref{eq:work_distribution}. Therefore, applying the protocol to $\rho$ is equivalent to applying it to the dephased state $\sum_i \langle\phi_{i}|\rho|\phi_{i}\rangle \proj{\phi_i}$. This can be viewed as a probabilistic mixture where the input is $\phi_i$ with probability $\langle\phi_{i}|\rho|\phi_{i}\rangle$. By Theorem \ref{thm:work_convergence}, an input of $\phi_i$ yields work $w_i$ (in the limit of large $M$). Thus, $\Pr(\Delta W=w_{i})=\langle\phi_{i}|\rho|\phi_{i}\rangle$.
    
    \item Using the probabilities above, the expected work is:
    \begin{align*}
        \mathbb{E}[\Delta W] &= \langle\phi_{0}|\rho|\phi_{0}\rangle w_0 + \langle\phi_{1}|\rho|\phi_{1}\rangle w_1 \\
        &= \langle\phi_{0}|\rho|\phi_{0}\rangle\beta^{-1}[\D(\phi_{0}\|\mathbb{I}/2)+\ln p_{0}]+\langle\phi_{1}|\rho|\phi_{1}\rangle\beta^{-1}[\D(\phi_{1}\|\mathbb{I}/2)+\ln p_{1}] \\
        &= \beta^{-1}\left[\langle\phi_{0}|\rho|\phi_{0}\rangle(-\Tr(\phi_{0}\ln \mathbb{I}/2)+\ln p_{0})+\langle\phi_{1}|\rho|\phi_{1}\rangle(-\Tr(\phi_{1}\ln \mathbb{I}/2)+\ln p_{1})\right]
    \end{align*}
    Letting $\mathcal{P}(\rho)=\sum_{i}\phi_{i}\rho\phi_{i}$ be the pinching map, we can rewrite this as:
    \begin{equation}
        \mathbb{E}[\Delta W] = \beta^{-1}[\Tr(\mathcal{P}(\rho)\ln\rho^{*})-\Tr(\mathcal{P}(\rho)\ln \mathbb{I}/2)] = \beta^{-1}[\Tr(\rho\ln\rho^{*})-\Tr(\rho\ln \mathbb{I}/2)].
    \end{equation}
    Finally, expressing this in terms of relative entropy yields:
    \begin{equation}
        \mathbb{E}[\Delta W]=\beta^{-1}[\D(\rho\|\mathbb{I}/2)-\D(\rho\|\rho^{*})]. 
    \end{equation}
\end{enumerate}
\end{proof}

\subsection{Cumulative dissipation as mathematical regret}
\label{sec:dissipation_regret}
The bulk of the derivations, definitions, algorithms and discussions covered in Section~\ref{sec:dissipation_regret} and Section~\ref{sec:numerics} are developed in a separate paper~\cite{huang2026timeorder}. For more in-depth discussions and understanding, please refer to the paper. Here we will cover what is relevant to the current paper. 

The expectation value of the work extracted at time step $l$, when the true quantum state is $\sigma_{m_l}$, evaluates to:
\begin{equation}
\begin{split}
\label{eq:exp_work}
    \beta\mathbb{E}[W_l] &= \ln 2 + \bra{\psi_{a_l}}\sigma_{m_l}\ket{\psi_{a_l}}\ln\lambda_{a_l}+\bra{\psi_{a_l}^\perp}\sigma_{m_l}\ket{\psi_{a_l}^\perp}\ln(1-\lambda_{a_l})\\
    &= \ln 2 + \Tr(\sigma_{m_l}\ln\rho_{a_l}) =\D(\sigma_{m_l}\|\mathbb{I}/2) -  \D(\sigma_{m_l}\|\rho_{a_l})~.
\end{split}
\end{equation}
If the agent correctly tailors the protocol such that $\rho_{a_l} = \sigma_{m_l}$, the expected work matches the local free energy $\D(\sigma_{m_l}\|\mathbb{I}/2)$. The penalty term $\D(\sigma_{m_l}\|\rho_{a_l})$ strictly represents the local thermodynamic dissipation caused by the agent's misaligned expectations~\cite{kolchinsky2017dependence,riechers2021initial}. This would be termed \emph{local dissipation} for the rest of the discussion.

To rigorously map our reinforcement learning framework to quantum thermodynamics, we distinguish between three levels of energy loss, building upon the definitions developed in~\cite{huang2026timeorder}. \emph{Local dissipation} in~\eqref{eq:exp_work} refers to the irreversible free energy loss at a single time step $l$ due to a mismatch between the agent's target state $\rho_{a_l}$ and the actual emitted state $\sigma_{m_l}$. \emph{Causal dissipation} $\delta(\overrightarrow{Q_1:Q_L})$ in~\eqref{eq:work_upperbound} represents the fundamental, unavoidable energy cost inherent in any sequential extraction process; it arises because a causal agent must disturb the system's temporal correlations to acquire information. Finally, cumulative dissipation $W_{\mathrm{diss}}$ is the \emph{excess} free energy loss across $K$ episodes because the agent does not initially know the environment's true parameters. While causal dissipation is an inherent physical limit of the optimal policy $V^*$, the cumulative dissipation is the thermodynamic equivalent of the mathematical regret that our algorithm seeks to minimize.

The cumulative value function optimized by the RL agent is the total expected work extracted:
\begin{equation}
\label{eq:work_value_function}
    V^\pi(\rho_{_{Q_{1:L}}}) = \mathbb{E}_{\pi, \omega}\left[\sum_{l=1}^L W_l\right] = \mathbb{E}_{\pi, \omega}\sum_{l=1}^L \beta^{-1} \Big(\D(\rho_{l|\tau_{l-1}}\|\mathbb{I}/2)-\D(\rho_{l|\tau_{l-1}}\|\rho_{a_l}^\pi)\Big),
\end{equation}
where $\rho_{l|\tau_{l-1}}$ is the reduced local state of $\rho_{Q_{1:L}}$ conditioned on the classical trajectory $\tau_{l-1}$ of past actions and outcomes, and $\rho_{a_l}^\pi$ corresponds to the target state dictated by the policy $\pi$ at time step $l$.

Because the agent is causally constrained, it cannot perfectly predict the upcoming state without measuring. To maximize global work extraction, the optimal policy $\pi^*$ must balance immediate work yields against the need for information acquisition. Consequently, the agent may deliberately incur a necessary local dissipation penalty $\D(\rho_{l|\tau_{l-1}}\|\rho_{a_l}^\pi)>0$ to obtain information that sharpen its future beliefs.

Consequently, the optimal value function for the optimal policy $\pi^*$, is then:
\begin{equation}
\label{eq:opt_work_value_function}
    V^*(\rho_{_{Q_{1:L}}}) = \mathbb{E}_{\pi^*, \omega}\left[\sum_{l=1}^L W_l\right] = \mathbb{E}_{\pi^*, \omega}\sum_{l=1}^L \beta^{-1} \Big(\D(\rho_{l|\tau_{l-1}}\|\mathbb{I}/2)-\D(\rho_{l|\tau_{l-1}}\|\rho_{a_l}^*)\Big).
\end{equation}
We can verify that~\eqref{eq:opt_work_value_function} perfectly matches the thermodynamic upper bound of~\eqref{eq:work_upperbound}. By expanding $V^*(\rho_{Q_{1:L}})$, we have:
\begin{equation}
    \begin{split}
        \beta V^*(\rho_{Q_{1:L}}) &= L\ln2 - \mathbb{E}_{\pi^*,\omega}\left[\sum_{l=1}^L \Big(S(\rho_{l|\tau_{l-1}})+\D(\rho_{l|\tau_{l-1}}\|\rho^{\pi^*}_{a_l})\Big)\right]\\
        &=L\ln2+\mathbb{E}_{\pi^*,\omega}\left[\sum_{l=1}^L\Tr(\rho_{l|\tau_{l-1}}\ln \rho^{\pi^*}_{a_l})\right]\\
        &= L\ln2+\mathbb{E}_{\pi^*,\omega}\left[\sum_{l=1}^L \Big(\bra{\psi_{a_l^*}}\rho_{l|\tau_{l-1}}\ket{\psi_{a_l^*}}\ln\lambda_{a_l^*} + \bra{\psi_{a_l^*}^\perp}\rho_{l|\tau_{l-1}}\ket{\psi_{a_l^*}^\perp}\ln(1-\lambda_{a_l^*})\Big)\right].
    \end{split}
\end{equation}
Because the optimal purity matches the true environment's probabilities (i.e., $\lambda_{a_l^*}=\bra{\psi_{a_l^*}}\rho_{l|\tau_{l-1}}\ket{\psi_{a_l^*}}$ and $1-\lambda_{a_l^*}= \bra{\psi_{a_l^*}^{\perp}}\rho_{l|\tau_{l-1}}\ket{\psi_{a_l^*}^{\perp}}$, as rigorously justified in Section~\ref{sec:opt_action}), the logarithmic terms convert exactly to the Shannon entropy of the observation:
\begin{equation}
    \begin{split}
        \beta V^*(\rho_{Q_{1:L}}) &= L\ln2 - \mathbb{E}_{\pi^*,\omega}\left[\sum_{l=1}^L H(O_l | a_l^*, \tau_{l-1})\right]\\
        &= L\ln2 - \mathbb{E}_{\pi^*,\omega}\left[\sum_{l=1}^{L-1}  H(O_l | a_l^*, \tau_{l-1}) + S(\rho_{Q_{L}}(\tau_{L-1}))\right]~.
    \end{split}
\end{equation}
The conversion from Shannon entropy to von Neumann entropy in the last step reflects that in the very last time step $L$, the agent minimizes local dissipation purely by choosing the eigenbasis of $\rho_{Q_{L}}(\tau_{L-1})$, without incurring any further sequential measurement penalties as there are no future timesteps. This expression is exactly equivalent to~\eqref{eq:work_upperbound}.

Because the RL reward is defined as the physical work extracted~\eqref{eq:reward_energy}, the environment parameter $\omega$ simply corresponds to the underlying multi-time state $\rho_{Q_{1:L}}$. Consequently, the suboptimality of any policy at episode $k$---the gap between the maximum extractable work $V^*(\omega)$ and the achieved work $V^{\pi^k}(\omega)$---is precisely the unrecoverable thermodynamic work dissipated during that episode. 

This allows us to directly equate the mathematical regret from~\eqref{eq:regret} with the \emph{cumulative dissipation} over $K$ episodes:
 
\begin{equation}
\begin{split}
\label{eq:cum_dissipation}
\regret (K) = W_{\mathrm{diss}}(K)  &= \sum_{k=1}^K\left[V^*(\rho_{Q_{1:L}}) - V^{\pi_k}(\rho_{Q_{1:L}})\right] \\
&= \beta^{-1} \sum_{k=1}^K\Bigg[\mathbb{E}_{\pi^k,\omega}\left[\sum_{l=1}^{L-1} H(O_l | a_l, \tau_{l-1}) + S(\rho_{Q_{L}}(\tau_{L-1}))\right]\\ 
&\quad\quad\quad\quad-\mathbb{E}_{\pi^*,\omega}\left[\sum_{l=1}^{L-1}H(O_l | a_l^*, \tau_{l-1}) +S(\rho_{Q_{L}}(\tau_{L-1}))\right]\Bigg]~.
\end{split}
\end{equation}
This precise mathematical equivalence establishes a profound physical consequence: designing a regret-minimizing RL algorithm is physically equivalent to engineering an adaptive Maxwell's demon. By deploying our continuous-action OMLE strategy over the input-output QHMM, the agent utilizes past battery outcomes to learn the environment's latent memory dynamics. Crucially, our $\widetilde{\mathcal{O}}(\sqrt{K})$ regret bounds guarantee that the agent can continuously update its extraction basis and purity parameters to achieve a strictly sublinear accumulation of irreversible thermodynamic dissipation

\section{Numerical experiments}\label{sec:numerics}

To demonstrate the validity of the bound we derived, we will showcase the learning and sequential work extraction from a simple model with a classical memory discussed in Sec.~\ref{subsec:classical-memory-emission}. We will first discuss how the optimal policy in~\eqref{eq:argmax} in Sec.~\ref{sec:omle_discrete} can be achieved utilizing dynamic programming, more specifically backward value iteration. Following that we will specify the model we simulate and showcase the simulation result.

\subsection{Computing optimal policy via dynamic programming}
In this specific setting, the hidden memory is operationally classical (restricted to the diagonal subspace $\rho_S =\sum_m x_mE_m$), and the environment sequentially emits known quantum states $\sigma_{m_l}$ conditioned on latent memory state $m_l\in\{1,2\}$.

Because the memory transitions classically via a known stochastic matrix $T(m'|m)$, the problem reduces to a Partially Observable Markov Decision Process (POMDP) where the latent state is classical, but the observations are generated by quantum instruments. The latent memory transitions essentially form an underlying Markov Chain and the agent's action $a_l$ (a choice of quantum measurement) acts strictly on the emitted quantum system, $Q_l$, rather than directly on the memory. Consequently, the agent's interventions do not induce non-commutative quantum back-action or generate coherence within the latent memory itself. 
The quantum mechanics of the problem is therefore entirely encapsulated within the observation and reward functions. Specifically, the conditional probability of extracting a specific work value $w_l$ corresponding to a measurement index $i\in\{0,1\})$ given the latent state $m_l$ and the chosen measurement action $a_l$ is simply given by Born's rule as
\begin{equation}
    \Pr(w_l|m_l,a_l) = \Tr(M^{(a_l)}_{i}\sigma_{m_l})~,
\end{equation}
as mentioned in~\eqref{eq:classical-memory-born-rule}. Because the latent state dynamics are classical, and the observation probabilities depend only on the current latent state and action, the framework perfectly satisfies the standard definition of a POMDP.

To rigorously evaluate the regret of our learning algorithm and establish the true maximum extractable work, we must first compute the exact optimal policy $\pi^*$ by reformulating this POMDP as Belief MDP~\cite{aastrom1965optimal,kaelbling1998planning}.

At any time step $l\in[L]$, the agent's sufficient statistic for the past history $\tau_{l-1}$ is its belief state $\eta_{l-1}$ defined as a classical probability distribution over the latent memory states:
\begin{equation}
    \eta_{l-1}(m) \coloneqq \Pr(M_l=m|\tau_{l-1})~.
\end{equation}
Based on this belief, the agent anticipates that the emitted quantum system will be in the expected state $\xi_{\eta_{l-1}}=\sum_{m=1}^1 \eta_{l-1}(m)\sigma_m$.

The agent selects a continuous action $a_l\in\A$, which configures the physical extraction protocol by specifying a measurement basis $\{M_{0}^{(a_l)},M_{1}^{(a_l)}\}$ and a target purity $\lambda_{a_l}$. Upon interacting with the system, the agent reads the battery and directly observes the extracted work value $w_l\in\{w_{l,0}, w_{l,1}\}$. This work value serves simultaneously as the RL reward and the observation. The probability of obtaining work $w_{l,i}$ corresponding to basis index $i$ is governed by Born's rule on the expected state:
\begin{equation}
    \Pr(w_{l,i}|\eta_{l-1},a_l)= \Tr[M_{i}^{(a_l)}\xi_{\eta_{l-1}}] =\sum_{m=1}^2\eta_{l-1}(m)\Tr[M_{i}^{(a_l)}\sigma_m]~.
\end{equation}
Observing the work value $w_l$ allows the agent to update its belief via Bayesian inference and subsequently pass it through the underlying environmental transition matrix $T$ to obtain the prior belief for the next step~\cite{cox1946probability,jaynes2003probability}. The exact deterministic belief transition function $\mathcal{B}$ is:
\begin{equation}
\label{eq:bayes}
    \eta_l(m') = \mathcal{B}(\eta_{l-1},a_l,w_{l,i})(m')\coloneqq \sum_{m=1}^2 T(m'|m)\frac{\eta_{l-1}(m)\Tr(M_{i}^{(a_l)}\sigma_m)}{\Pr(w_{l,i}|\eta_{l-1},a_l)}
\end{equation}

In the reinforcement learning framework, the agent's objective is to maximize the expected cumulative reward, which in this case would reflect the maximization of expected cumulative work extracted. As established in our thermodynamic analysis~\eqref{eq:exp_work}, this expected work extracted evaluates to a difference in relative entropies involving the expected state $\xi_{\eta_{l-1}}$, the thermal state $\gamma$, and the protocol's target state $\rho_{a_l}$:
\begin{equation}
\label{eq:work_reward}
    \mathbb{E}(w_l|\eta_{l-1},a_l) = \beta^{-1}\left[\D(\xi_{\eta_{l-1}}\|\gamma)-\D(\xi_{\eta_{l-1}}\|\rho_{a_l})\right]~.
\end{equation}
The action-value function (Q-function) can then be elegantly expressed via the Bellman optimality equation~\cite{bellman1966dynamic,puterman2014markov}. In its core, Bellman optimality is the mathematical principle that a globally optimal strategy is simply a sequence of optimal smaller steps. The Q-value is the sum of the expected immediate work and the expected work the agent can extract in the future quantified by the value function $V^*_{l+1}$~\eqref{eq:value_func} averaged over the possible measurement outcomes. It therefore quantifies the total expected work extraction of taking a specific action $a_l$ in your current state $\eta_{l-1}$, assuming the agent act perfectly from that point forward. 
\begin{equation}
    Q_l^*(\eta_{l-1},a_l)\coloneqq \mathbb{E}(w_l|\eta_{l-1},a_l)+\sum_{i\in\{0,1\}}\Pr(w_{l,i}|\eta_{l-1},a_l)\left[V^*_{l+1}(\mathcal{B}(\eta_{l-1},a_l,w_{l,i}))\right]~,
\end{equation}
Here $\mathcal{B}(\eta_{l-1},a_l,w_{l,i})$ refers to the belief state the agent transitions to upon obtaining work $w_{l,i}$.

To maximize the expected immediate work $w_l$ for any chosen eigenbasis, the optimal target purity of the protocol must always match the projection of the expected state onto that basis: $\lambda_{a_l}^*=\Tr\big(M^{(a_l)}_{0}\xi_{\eta_l}\big)$, reducing the agent's action space $\A$ (see Sec.~\ref{sec:opt_action} for more detail). The agent no longer needs to search the full space of measurement operators, but only the continuous space of measurement bases.

The optimal value function at time step $l\in[L]$ is then the maximization over these measurement bases:
\begin{equation}
\label{eq:value_func}
    V_l^*(\eta_l)=\max_{M^{(a_l)}\in\A} Q_l^*(\eta_{l-1},a_l)
\end{equation}
with the terminal boundary condition $V_L^*(\eta)=0$ for all $\eta$, signifying that at the end of $L$ time steps, there is no free energy left to be extracted in future.

\subsubsection{Backward iteration algorithm}
We solve this MDP via backward dynamic programming, a standard technique used in complex optimization of non-Markovian processes. Readers may refer to~\cite{puterman2014markov} for a detailed exposition. We discretize the one-dimensional probability simplex of the belief space $\eta\in\Delta^1$ into a finite set $\mathcal{K}_{DP}$, and similarly discretize the action space of measurement bases. The pseudocode of the said algorithm can be found in Algorithm~\ref{alg:backward_vi_work_extraction}. Note in this section, we assumes that the agent knows what the underlying process is, which would correspond to knowing $\omega_k$ in~\eqref{eq:argmax}, the agent would then use this algorithm to find the optimal policy $\pi^k$ for the process described by $\omega_k$.

\begin{algorithm}[H]
\caption{Backward Value Iteration for Sequential Work Extraction}
\label{alg:backward_vi_work_extraction}
\KwIn{Discretized belief simplex $\mathcal{K}_{DP}$, discretized basis angles $\mathcal{A}$, known emitted states $\{\sigma_1, \sigma_2\}$, transition matrix $T$, horizon $L$, inverse temperature $\beta$.}
\KwOut{Optimal policy $\pi^* = \{\pi_l^*\}_{l=1}^L$, representing a look-up table mapping any belief state $\eta \in \mathcal{K}_{DP}$ at step $l$ to the optimal measurement basis $a^* \in \mathcal{A}$.}

\textbf{Initialize:} Set $V_{L+1}^*(\eta^{(k)}) = 0$ for all $\eta^{(k)} \in \mathcal{K}_{DP}$\;

\For{step $l = L, L-1, \dots, 1$}{
    \For{each belief state $\eta^{(k)} \in \mathcal{K}_{DP}$}{
        Compute expected state: $\xi = \eta^{(k)}(1)\sigma_1 + \eta^{(k)}(2)\sigma_2$\;
        
        \For{each measurement basis $a \in \mathcal{A}$}{
            Set optimal target purity: $\lambda_a = \text{Tr}[M_0^{(a)} \xi]$\;
            Calculate expected immediate work: $\mathbb{E}(W|a,\eta) = \beta^{-1} \big[ \D(\xi||\gamma) - \D(\xi||\rho_{a}) \big]$\;
            Initialize $Q_l(\eta^{(k)}, a) = \mathbb{E}(W|a,\eta)$\;
            
            \For{each outcome index $i \in \{0, 1\}$}{
                Calculate outcome probability: $p_i = \text{Tr}[M_i^{(a)} \xi]$\;
                \If{$p_i > 0$}{
                    Compute next belief state $\eta_{l+1} = \mathcal{B}(\eta^{(k)}, a, i)$ following~\eqref{eq:bayes}\;
                    Find closest discrete belief $\eta^{(j)} \in \mathcal{K}_{DP}$ to $\eta_{l+1}$\;
                    Accumulate future value: $Q_l(\eta^{(k)}, a) \mathrel{+}= p_i V_{l+1}^*(\eta^{(j)})$\;
                }
            }
        }
        
        Update value function: $V_l^*(\eta^{(k)}) = \max_{a \in \mathcal{A}} Q_l(\eta^{(k)}, a)$\;
        Extract optimal policy mapping: $\pi_l^*(\eta^{(k)}) = \arg\max_{a \in \mathcal{A}} Q_l(\eta^{(k)}, a)$\;
    }
}
\end{algorithm}

The backward value iteration algorithm provides a standard reinforcement learning methodology to compute the maximum cumulative reward. Crucially, we can explicitly map the optimal value function obtained via these Bellman equations, $V_1^*(\eta_0)$, directly to the thermodynamic bound in~\eqref{eq:work_upperbound}.

At step $l$, the DP agent's belief state $\eta_{l-1}$ translate to an expected state $\xi_{\eta_{l-1}} = \sum_{m=1}^2\eta_{l-1}(m)\sigma_m$. Physically, this expected state is exactly conditional local state of the sequence, $\rho_{l|\tau_{l-1}}$. 
It is critical to distinguish the globally optimal policy $\pi^*$ from a purely local-greedy strategy. A local-greedy agent would maximize immediate work by always choosing a measurement basis that perfectly diagonalizes its expected state $\xi_{\eta_{l-1}}$. This would strictly eliminate the local thermodynamic dissipation penalty $\D(\xi_{\eta_{l-1}}\|\rho_{a_l})$ at current step, meaning the immediate extracted work would locally equal $\D(\xi_{\eta_{l-1}}\|\gamma)$. However, such a extraction often yields poor information regarding the underlying latent memory state, leading to highly mixed future belief states and diminished future work yields.

Because the DP agent optimizes the Q-function over the entire sequential horizon, the optimal policy may select an action $a_l^*$ that incurs a non-zero local thermodynamic dissipation penalty. Physically, the agent chooses to sacrifice immediate work to perform a measurement that extracts more information about the latent memory, unlocking greater total work extraction over the subsequent steps.
Consequently, the global optimal value function recovers the thermal dynamic bound in~\eqref{eq:work_upperbound}.

By establishing this exact mathematical equivalence, the regret of the learning agent over $K$ episodes is precisely the cumulative thermodynamic dissipation $W_{\mathrm{diss}}(K)$ in~\eqref{eq:cum_dissipation}. Minimizing regret via the estimated OOMs is, therefore, operationally equivalent to an adaptive Maxwell's demon learning to optimally balance the thermodynamic cost of information acquisition against future work exploitation.

\subsubsection{Optimal actions}
\label{sec:opt_action}
Recall that the value of the work extracted is dependent on the purity of the estimate, and the probability distribution is dependent on the eigenvectors of the estimate. The Bayesian update in~\eqref{eq:bayes} itself depends on the probability, not the magnitude of work (just depend on which of the 2 work values). That is to say, for a given eigenbasis, the inference of the agent stays the same, but we have the freedom to optimize the eigenvalues (or purity in our lingo) for the work extractions.

Recall that expected work extracted when at agent is in $\eta^{(i)}$ is given by
\begin{equation}
    \langle W\rangle = \ln2 +(\bra{\psi}\xi_{\eta^{(i)}}\ket{\psi}\ln\lambda+\bra{\psi^\perp}\xi_{\eta^{(i)}}\ket{\psi^\perp}\ln(1-\lambda))
\end{equation}
to maxmise this expression, one would obtain $\lambda=\bra{\psi}\xi_{\eta^{(i)}}\ket{\psi}$ as the optimum. that is to say the optimal state to tailor for any given eigenbasis is
\begin{equation}
    \rho^* = \sum_{i}p_i\ket{\psi_i}\!\bra{\psi_i}~,\quad p_i=\bra{\psi_i}\xi_{}\ket{\psi_i}~,
\end{equation}
in order words, the projection of the expected state onto the chosen eigenbasis. Hence the set of action consist of all the basis, and the purity will be determined automatically. For the case where there are only 2 states, one would choose a plane that the 2 quantum state reside in and choose all possible eigenbasis that lies within the plane for optimal Bayesian inference~\eqref{eq:bayes}, the intuition is that out of plane measurements do not yield more information compared to in plane measurements.

\subsection{Case Study}

To explicitly connect our theoretical learning guarantees with numerical practice, we implement the learning algorithm using the Observable Operator Model (OOM) framework (introduced in Section~\ref{sec:oom_uc}) on a canonical 2-state classical memory model (Section~\ref{subsec:classical-memory-emission}). The numerical simulation in this study is shown in Figure~\ref{fig:intro_cum_dissipation_compare}.

The environment features a diagonal memory state with two latent configurations, $\rho_S = x_1E_1+x_2E_2$ as was mentioned in~\eqref{eq:diag_memory}. Between time steps, the memory transitions via a parameter-dependent stochastic matrix:

\begin{equation}
    \mathbb{E} = \begin{pmatrix}
        \theta,&1-\theta\\
        1-\theta,&\theta
    \end{pmatrix}~.
\end{equation}
The agent is provided with the identities of the two conditionally emitted quantum states $\{\sigma_m\}_{m=1}^2$, but the transition probability $\theta\in(0,1)$ is strictly unknown. The agent's task is to learn this parameter on the fly to maximize the cumulative work extracted over $K$ episodes.

To perform Maximum Likelihood Estimation (MLE) without attempting to track the unobservable latent memory, the agent leverages the OOM superoperators. Because the memory is classical, the recovery map $R^{(a_l)}$ in~\eqref{eq:classical-memory-recovery-map} exists and is simply given by $(\mathbb{O}^{(a_l)})^{-1}$. The $\theta$-dependent OOM operator which maps the observable classical state forward upon taking action $a_{l+1}$ after observing outcome $o_l$ from action $a_l$, is explicitly constructed using the definition in Definition~\ref{def:oom_qhmm}:
\begin{equation}
    \mathbf{A}_{\theta}(o_l,a_l,a_{l+1}) = \Tr_S\circ\mathcal{P}^{(a_{l+1})}\circ(\mathbb{E}_\theta\otimes \mathcal{T}_{o_l})\circ R^{(a_l)}~.
\end{equation}
We first note that the recovery map will act on a observation vector $d= \begin{pmatrix}
    d_0\\d_1
\end{pmatrix}$ to obtain the unnormalized latent memory state 
\begin{equation}
    \hat{x} = (\mathbb{O}^{(a_l)})^{-1}d
\end{equation} as mentioned in~\eqref{eq:classical-memory-recovered-weights}.  is just the inverse of $(\mathbb{O}^{(a_l)})^{-1}$. The probability of observing an outcome conditioned on $a_l$ being taken can then be written as a diagonal matrix $\mathbb{D}$ where 
\begin{equation}
    \mathbb{D}_{o_l}^{(a_l)}= \begin{pmatrix}
        p_{a_l}(o_l|m=1)&0\\ 0 & p_{a_l}(o_l|m=2)
    \end{pmatrix}~,
\end{equation}
this is analogous to $\mathcal{T}_{o_l}$. The latent state transition is given by $\mathbb{E}_\theta$. We then have to map the latent state into the observable space using the next observation matrix conditioned on next action $\mathbb{O}^{(a_{l+1})}$. This brings us to the explicit construction of the OOM in the form of
\begin{equation}
    \mathbf{A}(o_l,a_l,a_{l+1} ) = \mathbb{O}^{(a_{l+1})}\cdot\mathbb{E}_\theta\cdot \mathbb{D}_{o_l}^{(a_l)}\cdot(\mathbb{O}^{(a_l)})^{-1}
\end{equation}
The probability of the action-outcome trajectory $\tau_{L}=(a_1,o_1,\cdots,a_L,o_L)$ can be computed by multiplying these OOM operators as shown in Lemme.~\ref{lem:probabilities_operator_multiquantum}~.
\begin{equation}
A(\tau_{L}) = \mathbf{e}_{o_L}^\mathrm{T}\mathbf{A}(o_{L-1},a_{L-1},a_{L})\cdots\mathbf{A}(o_{1},a_{1},a_{2})\mathbf{a}(a_1)
\end{equation}
where $\mathbf{e}_{o_L}$ is an elementary vector.

At the beginning of each episode $k$, the agent possesses a dataset of all past executed trajectories $D_k=\{\tau^{k'}_L\}_{k'=1}^{k-1}$ The sequential learning process operates via the following loop:
\begin{enumerate}
    \item \textbf{Estimation (OOM-based MLE)}: The agent computes the Maximum Likelihood Estimate $\hat{\theta}_k$ of the environment by maximizing the log-likelihood of the observed trajectories using the OOM trajectory probabilities:
    \begin{equation}
        \hat{\theta}_k = \arg \max_{\theta\in [0,1]}\sum_{j=1}^{k-1}\log A(\tau^j_{L})~.
    \end{equation}\\
    \item \textbf{Planning}: Using the estimated transition parameter $\hat{\theta}_k$, the agent computes the optimal policy $\pi^k$ via Backward Value Iteration algorithm (Alg.~\ref{alg:backward_vi_work_extraction}).\\
    \item \textbf{Execution}: The agent executes the extraction protocol defined by $\pi^k$ on the physical sequence $\rho_{Q_{1:L}}$, collects a new trajectory $\tau_L^k$ along with the total extracted work $W_L^k$ and appends the trajectory to the dataset $D_{k+1}$.
\end{enumerate}
We simulate this process over 50 independent runs. To evaluate the learning efficiency, we plot the cumulative work dissipation $W_{\mathrm{diss}}(K)$ in ~\eqref{eq:cum_dissipation} against the number of elapsed episodes K, yielding the graph in Figure~\ref{fig:intro_cum_dissipation_compare}. The strictly sublinear accumulation of dissipation confirms that the OOM-based estimation rapidly converges, allowing the agent to adapt its measurement basis and extraction purities on the fly to saturate the fundamental thermodynamic limits of the environment.

\section*{Acknowledgments}
J.L. and M.F. gratefully acknowledge Marco Tomamichel for funding M.F.'s visit to Singapore in 2024, which facilitated the initial discussions that seeded this project.
J.L. further thanks Erkka Haapasalo and Marco Tomamichel for fruitful discussions on multi-armed quantum bandit lower bounds several years ago, which inspired the previously unpublished SIC-POVM lower bound presented here.  

This work is supported by the National Research Foundation through the NRF Investigatorship on Quantum-Enhanced Agents (Grant No. NRF-NRFI09-0010), the Singapore Ministry of Education Tier 1 Grant RT4/23 and RG77/22 (S), the National Quantum Office, hosted in A*STAR, under its Centre for Quantum Technologies Funding Initiative (S24Q2d0009), and the Hong Kong Research Grant Council (RGC) through the General Research Fund (GRF) grant 17302724. 

Finally, the authors acknowledge the use of Google Gemini for refining the presentation and style of this manuscript.

\section*{Code Availability Statement} 

The code used to generate the numerical results and figures in this work is publicly available at
\url{https://github.com/ruochenghuang/QHMM_work_dissipation_simulation}

\newcommand{\bbE}{\mathbb{E}}
\newcommand{\bbP}{\mathbb{P}}

\bibliographystyle{ultimate}
\bibliography{biblio_workFCS}
\appendix

\section{Proofs for MLE concentration bounds with continuous actions}
\label{sec:appendix_mle_continuous}

In this appendix, we formalize the maximum-likelihood estimation (MLE) concentration bounds for environments with continuous action spaces, providing the full rigorous proofs for Proposition~\ref{prop:lr-cont} and Proposition~\ref{prop:l1-cont}. Our approach adapts the classical MLE analysis for POMDPs~\cite{liu2022partially, liu2023optimistic} to our input-output quantum hidden Markov model (QHMM), carefully handling the continuous action densities and bounding the complexity of the likelihood class via a finite-dimensional linear embedding.

A fundamental advantage of the sequential decision-making framework is that the agent's continuous policy density perfectly cancels out in the likelihood ratio. Furthermore, because the environment's response to continuous actions can be embedded via the linearization of the quantum state space, the statistical complexity depends purely on the finite dimensionality of the quantum systems.

\subsection{Trajectory densities and parameter space embedding}
\label{app:parameter-space}

Recall from Section~\ref{sec:continous_model} that the reference measure on the trajectory space $\mathcal{T}_L = (\mathcal{A}\times[O])^L$ is defined as $\nu_L := \mu^{\otimes L} \otimes \mathrm{count}^{\otimes L}$, where $\mu$ is the dominating measure on the continuous action space $\mathcal{A}$. For any policy $\pi$ and model $\omega \in \Omega$, the probability measure of a trajectory $\tau_L = (a_1, o_1, \dots, a_L, o_L)$ admits a density with respect to $\nu_L$, given by
\begin{equation}
p_{\omega}^{\pi}(\tau_L) = \pi(\tau_L)\,A_\omega(\tau_L),
\end{equation}
where $\pi(\tau_L) = \prod_{l=1}^L \pi_l(a_l | \tau_{l-1})$ is the action density governed strictly by the agent, and $A_\omega(\tau_L) = \prod_{l=1}^L \Pr_\omega(o_l | \tau_{l-1}, a_l)$ is the observation likelihood. 

Crucially, for any two models $\omega, \omega' \in \Omega$, the continuous policy density cancels out exactly in the likelihood ratio:
\begin{equation}
\frac{p_{\omega'}^{\pi}(\tau_L)}{p_{\omega}^{\pi}(\tau_L)} \;=\; \frac{\pi(\tau_L)\, A_{\omega'}(\tau_L)}{\pi(\tau_L)\, A_\omega(\tau_L)} \;=\; \frac{A_{\omega'}(\tau_L)}{A_\omega(\tau_L)},
\end{equation}
provided the denominator is non-zero. This cancellation implies that any concentration bound relying on the log-likelihood ratio only requires us to control the complexity of the function class of observation likelihoods $\mathcal{F} = \{A_\omega : \omega \in \Omega\}$, completely bypassing the need to cover the infinite-dimensional policy space or continuous action measure.

We now bound the complexity of $\mathcal{F}$ by embedding the unknown parameters into a finite-dimensional space. As defined in Section~\ref{sec:embedding-dim}, the spanning dimension of the action class $\mathcal{A}$ is $d_\mathcal{A} \le (O-1)S^2$. We fix a Hilbert--Schmidt orthonormal basis $\{P_\mu\}_{\mu=1}^{S^2}$. Because the POVM effects $\{M_o^{(a)}\}_{o\in[O]}$ lie in a space of dimension $d_\mathcal{A}$, the environment's response to \emph{any} continuous action is entirely determined by its response on this finite-dimensional basis. 

As a result, any candidate model $\omega \in \Omega$ can be parameterized by at most
\begin{equation}
d_\omega \leq  L S^4 + d_\mathcal{A}  S^2 + S^2
\end{equation}
real parameters (accounting for the initial state, the evolution channels, and the instrument elements). The physical constraints (e.g., CPTP maps, CPTNI instruments) naturally confine the parameters to a bounded, compact subset of $\mathbb{R}^{d_\omega}$.

The observation likelihood $A_\omega(\tau_L)$ is formed by taking the trace of the composition of these linear maps. By multi-linearity, $A_\omega(\tau_L)$ is a polynomial of degree $L$ in these bounded parameters. Therefore, the mapping $\omega \mapsto A_\omega(\tau_L)$ is uniformly Lipschitz continuous. That is, there exists a constant $L_{\mathrm{lip}} > 0$ such that for any two models $\omega, \omega' \in \Omega$, parameterized by vectors $\varphi(\omega), \varphi(\omega') \in \mathbb{R}^{d_\omega}$, and any trajectory $\tau_L$:
\begin{equation}
\label{eq:lipschitz_bound}
\big|A_\omega(\tau_L) - A_{\omega'}(\tau_L)\big| \leq L_{\mathrm{lip}} \|\varphi(\omega) - \varphi(\omega')\|_\infty.
\end{equation}

\subsection{Optimistic bracketing and mass inflation}
\label{app:bracketing}

Using the Lipschitz property~\eqref{eq:lipschitz_bound}, we construct a finite optimistic $\epsilon$-bracketing cover for the trajectory densities. We build an $\epsilon$-grid over the compact parameter space $\Omega \subset \mathbb{R}^{d_\omega}$ with scale 
\begin{align}\label{eq:eps_grid}
\epsilon := \frac{1}{2 K O^L L_{\mathrm{lip}}}.
\end{align}
The number of points in this grid, and thus the number of pseudo-models $\overline{\Omega}$, is bounded by
\begin{equation}
\label{eq:cover_size}
    \log |\overline{\Omega}| \leq  \mathcal{O}\Big( d_\omega L\log\!\big(K O\big) \Big).
\end{equation}
where we just used that for a $d$-dimensional space, the number of grid points required to cover it scales inversely with the volume of an $\epsilon$-ball $|\overline{\Omega}| \approx \left(\frac{\text{Constant}}{\epsilon}\right)^d$.

Then for each $\omega' \in \Omega$, we can choose a grid point $\omega_c \in \overline{\Omega}$ such that $\|\varphi(\omega') - \varphi(\omega_c)\|_\infty \le \epsilon$. We define the corresponding optimistic upper bracket density $p_{\bar{\omega}'}^\pi$ (associated with $\bar{\omega}'$) uniformly for all $\tau_L$ as:
\begin{equation}
p_{\bar{\omega}'}^\pi(\tau_L) \coloneqq \pi(\tau_L) \big( A_{\omega_c}(\tau_L) + L_{\mathrm{lip}}\epsilon \big).
\end{equation}
By the Lipschitz bound, this guarantees the pointwise upper bound:
\begin{equation}
\label{eq:optimistic-bracketing}
p_{\omega'}^\pi(\tau_L) \leq  p_{\bar{\omega}'}^\pi(\tau_L), \qquad \forall \tau_L, \forall \pi.
\end{equation}
Integrating this pseudo-density with respect to the continuous reference measure $\nu_L$ (defined in Section~\ref{sec:continous_model}), we obtain the \emph{mass inflation} bound:
\begin{equation}
\begin{aligned}
    \big\|p_{\bar{\omega}'}^\pi\big\|_{\ell_1(\nu_L)} 
    & =  \sum_{o_{1:L}\in [O]^L} \int_{\mathcal{A}^L} p_{\bar{\omega}'}^\pi(\tau_L) \mu^{\otimes L}(\dd a_{1:L})  \\
    & \leq \sum_{o_{1:L}\in [O]^L} \int_{\mathcal{A}^L} \pi(\tau_L) \big( A_{\omega'}(\tau_L) + 2 L_{\mathrm{lip}}\epsilon \big)  \mu^{\otimes L}(\dd a_{1:L})  \\
    & = \int_{\mathcal{T}_L} p_{\omega'}^\pi(\tau_L) \nu_L(\dd\tau_L) + 2 L_{\mathrm{lip}}\epsilon \sum_{o_{1:L}\in [O]^L} \int_{\mathcal{A}^L} \pi(\tau_L) \mu^{\otimes L}(\dd a_{1:L}).
\end{aligned}
\end{equation}
The first term is exactly 1 as $p_{\omega'}^\pi$ is a valid probability density. For the second term, observe that 
\begin{align}
    \int_{\mathcal{A}^L} \pi(a_{1:L}) \mu^{\otimes L}(\dd a_{1:L}) = 1, 
\end{align}
for any fixed observation sequence. There there are exactly $O^L$ such sequences. Thus:
\begin{equation}
    \big\|p_{\bar{\omega}'}^\pi\big\|_{\ell_1(\nu_L)} \leq 1 + 2 L_{\mathrm{lip}}\epsilon  O^L  = 1 + \frac{1}{K}, 
    \label{eq:mass-inflation-assump}
\end{equation}
where our choice of $\epsilon$ in~\eqref{eq:eps_grid} ensures the mass inflation strictly satisfies the required scale.

\subsection{Proof of Proposition~\ref{prop:lr-cont} (Likelihood ratio concentration)}

\begin{proof}[Proof of Proposition~\ref{prop:lr-cont}]
Let $\omega \in \Omega$ denote the true underlying model generating the trajectories. For any pseudo-model $\bar{\omega} \in \overline{\Omega}$, define the cumulative log-likelihood ratio
\begin{equation}
    L_t(\bar{\omega}) \coloneqq  \sum_{i=1}^t \log \left( \frac{p_{\bar{\omega}}^{\pi^i}(\tau^i)}{p_{\omega}^{\pi^i}(\tau^i)} \right).
\end{equation}
Let $\mathcal{F}_{t-1}$ be the filtration containing the history up to the end of episode $t-1$, which deterministically specifies the policy $\pi^t$ to be executed at episode $t$. Taking the conditional expectation of a single term yields:
\begin{align}
    \mathbb{E} \left[ \exp\left( \log \frac{p_{\bar{\omega}}^{\pi^t}(\tau^t)}{p_{\omega}^{\pi^t}(\tau^t)} \right) \middle|  \mathcal{F}_{t-1} \right]
    =  \mathbb{E} \left[ \frac{p_{\bar{\omega}}^{\pi^t}(\tau^t)}{p_{\omega}^{\pi^t}(\tau^t)} \middle| \mathcal{F}_{t-1} \right]  = \int_{\mathcal{T}_L} p_{\bar{\omega}}^{\pi^t}(\tau) \nu_L(\dd\tau)  = \big\| p_{\bar{\omega}}^{\pi^t} \big\|_{\ell_1(\nu_L)} \leq  1 + \frac{1}{K}, 
    \label{eq:step-mgf}
\end{align}
where the inequality follows from the mass inflation bound~\eqref{eq:mass-inflation-assump}. 
Applying the tower rule and telescoping over the $t$ episodes, the moment generating function satisfies:
\begin{equation}
    \mathbb{E}\left[ \exp(L_t(\bar{\omega})) \right] \leq \left( 1 + \frac{1}{K} \right)^t \leq  \left( 1 + \frac{1}{K} \right)^K \leq  e.
\end{equation}
By Markov's inequality, for any $u > 0$:
\begin{equation}
    \Pr\left[ L_t(\bar{\omega}) > u \right] = \Pr \left[ \exp(L_t(\bar{\omega})) > e^u \right] \leq e^{-u} \mathbb{E}[\exp(L_t(\bar{\omega}))] \leq e^{1-u}.
\end{equation}
Setting $u = \log(K |\overline{\Omega}| / \delta) + 1$ and applying a union bound over all episodes $t \in [K]$ and all pseudo-models $\bar{\omega} \in \overline{\Omega}$, we ensure that with probability at least $1-\delta$:
\begin{equation}
    L_t(\bar{\omega}) \leq \log |\overline{\Omega}| + \log(K/\delta) + 1,  \qquad \forall t \in [K],   \forall \bar{\omega} \in \overline{\Omega}.
\end{equation}
Finally, for any continuous candidate model $\omega' \in \Omega$, we map it to its optimistic upper bracket $\bar{\omega}' \in \overline{\Omega}$ satisfying $p_{\omega'}^\pi \le p_{\bar{\omega}'}^\pi$ pointwise. Thus:
\begin{equation}
    \sum_{i=1}^t \log \left(\frac{p_{\omega'}^{\pi^i}(\tau^i)}{p_{\omega}^{\pi^i}(\tau^i)}\right) \leq \sum_{i=1}^t \log \left(\frac{p_{\bar{\omega}'}^{\pi^i}(\tau^i)}{p_{\omega}^{\pi^i}(\tau^i)}\right) = L_t(\bar{\omega}') \leq \log |\overline{\Omega}| + \log(K/\delta) + 1.
\end{equation}
Substituting the bracketing entropy bound $\log |\overline{\Omega}| \le \mathcal{O}(d_\omega L \log(K O))$ from~\eqref{eq:cover_size} and absorbing constants yields the stated result.
\end{proof}

\subsection{Proof of Proposition~\ref{prop:l1-cont} (Trajectory  \texorpdfstring{$\ell_1$}{l1} concentration)}

\begin{proof}[Proof of Proposition~\ref{prop:l1-cont}]
Fix $t \in [K]$. We adapt a tangent-sequence argument over the continuous trajectory space. Let $D_t = \{(\pi^i, \tau^i)\}_{i=1}^t$ be the true observed dataset, where $\tau^i \sim \mathbb{P}_{\omega}^{\pi^i}$ and $\omega$ is the true model. Let $\widetilde{D}_t = \{(\pi^i, \tilde{\tau}^i)\}_{i=1}^t$ be a ``tangent'' sequence, where conditional on $D_t$, each $\tilde{\tau}^i \sim \mathbb{P}_{\omega}^{\pi^i}$ is drawn independently.

For any measurable function $\ell(\pi, \tau) \in \mathbb{R}$, define
\begin{align}
    L(D_t) = \sum_{i=1}^t \ell(\pi^i, \tau^i), \qquad L(\widetilde{D}_t) = \sum_{i=1}^t \ell(\pi^i, \tilde{\tau}^i). 
\end{align}
By a tower-rule induction (integrability is ensured since we work with densities w.r.t.\ $\nu_L$), we have the identity~\cite{liu2023optimistic}:
\begin{equation}
\label{eq:tangent-identity}
    \mathbb{E} \left[ \exp\left( L(D_t) - \log \mathbb{E}[\exp(L(\widetilde{D}_t)) | D_t] \right) \right] = 1.
\end{equation}

For each pseudo-model $\bar{\omega} \in \overline{\Omega}$ from our optimistic bracketing cover, define the strictly normalized density:
\begin{equation}
    \tilde{p}_{\bar{\omega}}^\pi(\tau_L) \coloneqq  \frac{p_{\bar{\omega}}^\pi(\tau_L)}{\| p_{\bar{\omega}}^\pi \|_{\ell_1(\nu_L)}}.
\end{equation}
Since $\tilde{p}_{\bar{\omega}}^\pi$ is a true probability density w.r.t.\ $\nu_L$, its Hellinger affinity with the true density is bounded by 1:
\begin{equation}
    \int_{\mathcal{T}_L} \sqrt{\tilde{p}_{\bar{\omega}}^\pi(\tau_L)  p_\omega^\pi(\tau_L)}  \nu_L(\dd \tau_L) \leq  1.
\end{equation}

Define the function $\ell_{\bar{\omega}}(\pi, \tau) := \frac{1}{2} \log \left( \frac{\tilde{p}_{\bar{\omega}}^\pi(\tau)}{p_\omega^\pi(\tau)} \right)$ if $p_\omega^\pi(\tau) > 0$, and $0$ otherwise. 
Let $L_{\bar{\omega}}(D_t) = \sum_{i=1}^t \ell_{\bar{\omega}}(\pi^i, \tau^i)$. By the conditional independence of the tangent sequence given $D_t$, we have:
\begin{equation}
\label{eq:hellinger-affinity}
    \mathbb{E}\exp[(L_{\bar{\omega}}(\widetilde{D}_t)) | D_t] = \prod_{i=1}^t \int_{\mathcal{T}_L} \sqrt{\tilde{p}_{\bar{\omega}}^{\pi^i}(\tau) p_\omega^{\pi^i}(\tau)}  \nu_L(\dd\tau).
\end{equation}
Similarly to~\eqref{eq:tangent-identity}, 
\begin{align}
    \mathbb{E}\left[\exp\left( L_{\bar{\omega}}(D_t) - \log \mathbb{E}[\exp(L_{\bar{\omega}}(\widetilde{D}_t)) | D_t] \right)\right]=1. 
\end{align}
By Markov's inequality and union bound over all $t \in [K]$ and $\bar{\omega} \in \overline{\Omega}$ with probability at least $1-\delta$:
\begin{equation}
\label{eq:chernoff-step}
    - \log \mathbb{E}[\exp(L_{\bar{\omega}}(\widetilde{D}_t)) | D_t] \leq  - L_{\bar{\omega}}(D_t) + \log(K |\overline{\Omega}|/\delta ).
\end{equation}
Using the scalar inequality $-\log x \ge 1 - x$ for $x \in (0, 1]$, and~\eqref{eq:hellinger-affinity} yields:
\begin{equation}
-\log \mathbb{E}[\exp(L_{\bar{\omega}}(\widetilde{D}_t)) | D_t] \;\ge\; \sum_{i=1}^t \left(1 - \int_{\mathcal{T}_L} \sqrt{\tilde{p}_{\bar{\omega}}^{\pi^i}(\tau) \, p_\omega^{\pi^i}(\tau)} \, \nu_L(d\tau) \right).
\label{eq:minuslog-lower}
\end{equation}
By Cauchy-Schwarz inequality, the $\ell_1$ (total variation) distance is strictly bounded by the Hellinger affinity gap: 
\begin{align}
    \left(\int |p - q| \dd \nu_L\right)^2 \leq 8\left(1 - \int \sqrt{pq} \dd\nu_L\right). 
\end{align} 
Combining this with~\eqref{eq:chernoff-step} and~\eqref{eq:minuslog-lower}, and letting $\tilde{\mathbb{P}}_{\bar{\omega}}^{\pi}$ denote the probability measure corresponding to the density $\tilde{p}_{\bar{\omega}}^{\pi}$, gives:
\begin{equation}
\label{eq:baromega-bound}
    \sum_{i=1}^t \big\| \tilde{\mathbb{P}}_{\bar{\omega}}^{\pi^i} - \mathbb{P}_\omega^{\pi^i} \big\|_1^2 \leq 8 \left( - L_{\bar{\omega}}(D_t) + \log (K |\overline{\Omega}|/\delta ) \right).
\end{equation}

We now map this pseudo-model bound to an arbitrary target model $\omega' \in \Omega$. Select its associated optimistic bracket $\bar{\omega}' \in \overline{\Omega}$ from~\eqref{eq:optimistic-bracketing}. 
Because $p_{\omega'}^\pi \le p_{\bar{\omega}'}^\pi = \| p_{\bar{\omega}'}^{\pi^i} \|_{\ell_1(\nu_L)}  \tilde{p}_{\bar{\omega}'}^\pi$ point-wise:
\begin{equation}
\label{eq:log-compare}
    - L_{\bar{\omega}'}(D_t) = \frac{1}{2} \sum_{i=1}^t \log \left( \frac{p_\omega^{\pi^i}(\tau^i)}{\tilde{p}_{\bar{\omega}'}^{\pi^i}(\tau^i)} \right) \leq \frac{1}{2} \sum_{i=1}^t \log \left( \frac{p_\omega^{\pi^i}(\tau^i)}{p_{\overline{\omega}'}^{\pi^i}(\tau^i)} \right) + \frac{1}{2} \sum_{i=1}^t \log \| p_{\bar{\omega}'}^{\pi^i} \|_{\ell_1(\nu_L)}.
\end{equation}
Using~\eqref{eq:mass-inflation-assump}, $\log \| p_{\bar{\omega}'}^{\pi^i} \|_{\ell_1(\nu_L)} \leq 1/K$ and $t \le K$, the sum $\sum_{i=1}^t \log \| p_{\bar{\omega}'}^{\pi^i} \|_{\ell_1(\nu_L)} \le 1$. 

By the triangle inequality and the mass inflation bound:
\begin{equation}
    \big\| \mathbb{P}_{\omega'}^{\pi^i} - \mathbb{P}_\omega^{\pi^i} \big\|_1 \leq  \big\| \tilde{\mathbb{P}}_{\bar{\omega}'}^{\pi^i} - \mathbb{P}_\omega^{\pi^i} \big\|_1 + \big\| \tilde{p}_{\bar{\omega}'}^{\pi^i} - p_{\omega'}^{\pi^i} \big\|_{\ell_1(\nu_L)}.
\end{equation}
The residual term evaluates to:
\begin{equation}
    \big\| \tilde{p}_{\bar{\omega}'}^{\pi^i} - p_{\omega'}^{\pi^i} \big\|_{\ell_1(\nu_L)} \leq \big\| \tilde{p}_{\bar{\omega}'}^{\pi^i} - p_{\bar{\omega}'}^{\pi^i} \big\|_{\ell_1(\nu_L)} + \big\| p_{\bar{\omega}'}^{\pi^i} - p_{\omega'}^{\pi^i} \big\|_{\ell_1(\nu_L)} \leq  (\| p_{\bar{\omega}'}^{\pi^i} \|_{\ell_1(\nu_L)} - 1) + \frac{1}{K} \leq \frac{2}{K}.
\end{equation}
Using the algebraic inequality $(x + 2/K)^2 \le 2x^2 + 8/K^2$, we obtain:
\begin{equation}
    \sum_{i=1}^t \big\| \mathbb{P}_{\omega'}^{\pi^i} - \mathbb{P}_\omega^{\pi^i} \big\|_1^2 \;\le\; 2 \sum_{i=1}^t \big\| \tilde{\mathbb{P}}_{\bar{\omega}'}^{\pi^i} - \mathbb{P}_\omega^{\pi^i} \big\|_1^2 + \frac{8}{K}.
\end{equation}
Combining this with~\eqref{eq:baromega-bound},~\eqref{eq:log-compare}, and observing that $\sum_{i=1}^t \log \| p_{\bar{\omega}'}^{\pi^i} \|_{\ell_1(\nu_L)}  \le 1$, we absorb all additive constants into a universal constant $c > 0$, concluding that:
\begin{equation}
    \sum_{i=1}^t \big\| \mathbb{P}_{\omega'}^{\pi^i} - \mathbb{P}_\omega^{\pi^i} \big\|_1^2 \;\le\; c \left( \sum_{i=1}^t \log \left( \frac{p_\omega^{\pi^i}(\tau^i)}{p_{\omega'}^{\pi^i}(\tau^i)} \right) + \log |\overline{\Omega}| + \log(K/\delta) \right).
\end{equation}
Substituting the bracketing entropy bound $\log |\overline{\Omega}| \le \mathcal{O}\big(d_\omega L \log(K O)\big)$ finishes the proof.
\end{proof}

\section{Detailed derivation for work distribution}
\label{sec:detailed_work_distribution}

In this section, we provide a detailed derivation of the work distribution resulting from the $\rho^*$-ideal work extraction protocol (Alg.~\ref{alg:sc_work_extraction}) and analyze the expected work extracted for different input states. We will utilize the Lagrange mean value theorem and the first mean value theorem for definite integrals.

\begin{theorem}
\label{thm:work_convergence}
Let $\{\phi_i\}_{i=0,1}$ and $\{p_i\}_{i=0,1}$ be the eigenvectors and eigenvalues of the target state $\rho^*$. Let $\Delta W$ be the extracted work (which is a continuous random variable) and $M$ be the number of repetitions as in Alg.~\ref{alg:sc_work_extraction}. If the extraction protocol is operated on a state $\rho$ that is an eigenstate of $\rho^*$ (i.e., $\rho = \phi_i$), then the expected extracted work $\mathbb{E}[\Delta W]$ converges to a fixed value $w_i$, and the extracted work $\Delta W$ converges in probability to its expectation $\mathbb{E}[\Delta W]$ as $M \rightarrow \infty$.

To be precise, it means:
\begin{equation}
    \lim_{M\rightarrow\infty}\mathbb{E}[\Delta W]=w_{i}
\end{equation}
where
\begin{equation}
    w_{i} := \beta^{-1}(\D(\phi_{i}\|\mathbb{I}/2)+\ln p_{i}),
\end{equation}
and for any $\epsilon>0$,
\begin{equation}
    \lim_{M\rightarrow\infty}\Pr[|\Delta W-\mathbb{E}[\Delta W]|\ge\epsilon]=0.
\end{equation}
\end{theorem}

\begin{proof}
We consider the case where the input state is $\rho = \phi_i$. We assume that $p_0 > \frac{1}{2}$ without loss of generality.

According to the state evolution described in Section 9.2, the state after the first repetition can be viewed as a classical state: it is $\phi_{x_1}$ where $x_1$ is a random bit sampled from $\{0,1\}$ according to the probability distribution $(p_{0,1}, p_{1,1})$. The extracted work after this first repetition, conditioned on $x_1$, is $(x_1 - i)\nu(1)$.

The evolution in subsequent repetitions $l$ (where $l \ge 2$) can be viewed as a classical process: the state after repetition $l$ is $\phi_{x_l}$, where $x_l$ is a random bit sampled from $\{0,1\}$ according to $(p_{0,l}, p_{1,l})$. The extracted work in repetition $l$, conditioned on $x_{l-1}$ and $x_l$, is $(x_l - x_{l-1})\nu(l)$.

Suppose the sequence of random bits sampled over $M$ repetitions is $x_1 \dots x_M$. The total extracted work after $M$ repetitions, conditioned on this sequence, is:
\begin{equation}
    \Delta W = (x_{1}-i)\nu(1)+\sum_{l=2}^{M}(x_{l}-x_{l-1})\nu(l) = -i\nu(1)+\sum_{l=1}^{M-1}x_{l}(\nu(l)-\nu(l+1))+x_{M}\nu(M)
\end{equation}
Recall that $\nu(l)=\beta^{-1}\ln\frac{p_{0}-l\delta p}{p_{1}+l\delta p}$, where $\delta p=\frac{1}{M}(p_{0}-\frac{1}{2})$.

The expected extracted work is:
\begin{equation}
    \mathbb{E}[\Delta W] = -i\nu(1)+\sum_{l=1}^{M-1}\mathbb{E}[x_{l}](\nu(l)-\nu(l+1))+\mathbb{E}[x_{M}]\nu(M)
\end{equation}
Since $\mathbb{E}[x_l] = p_{1,l} = p_1 + l\delta p$, we obtain:
\begin{equation}
    \mathbb{E}[\Delta W] = -i\beta^{-1}\ln\frac{p_{0}-\delta p}{p_{1}+\delta p}+\beta^{-1}\sum_{l=1}^{M-1}\left(\ln\frac{p_{0}-l\delta p}{p_{1}+l\delta p}-\ln\frac{p_{0}-(l+1)\delta p}{p_{1}+(l+1)\delta p}\right)(p_{1}+l\delta p).
\end{equation}

Using the Lagrange mean value theorem on $\nu(l)$, we can express the difference as:
\begin{equation}
    \beta(\nu(l)-\nu(l+1)) = \frac{1}{\xi_{l}(1-\xi_{l})}\delta p,
\end{equation}
for some $\xi_{l}\in[p_{0}-(l+1)\delta p,p_{0}-l\delta p]$.

Therefore, the expectation simplifies to:
\begin{align*}
    \mathbb{E}[\Delta W] &= -i\beta^{-1}\left(\ln\frac{p_{0}}{p_{1}}-\frac{\delta p}{\xi_{0}(1-\xi_{0})}\right)+\beta^{-1}\sum_{l=1}^{M-1}\frac{p_{1}+l\delta p}{\xi_{l}(1-\xi_{l})}\delta p \\
    &= -i\beta^{-1}\ln\frac{p_{0}}{p_{1}}+\beta^{-1}\sum_{l=1}^{M}\frac{p_{1}+l\delta p}{\xi_{l}(1-\xi_{l})}\delta p+i\beta^{-1}\frac{\delta p}{\xi_{0}(1-\xi_{0})}-\beta^{-1}\frac{\delta p}{2\xi_{M}(1-\xi_{M})}
\end{align*}

We approximate the sum with a definite integral. Using the first mean value theorem for definite integrals, the remainder $R_1(\delta p)$ can be shown to be bounded such that $R_1(\delta p) < \frac{1}{2} \delta p \frac{1}{p_1 (1-p_0)}$. Thus, the summation is finite while the remainder terms are of order $\mathcal{O}(\delta p)$:
\begin{equation}
    \beta^{-1}\sum_{l=1}^{M}\frac{p_{1}+l\delta p}{\xi_{l}(1-\xi_{l})}\delta p = \beta^{-1}\int_{1/2}^{p_{0}}\frac{dp}{p}+\beta^{-1}R_{1}(\delta p) = \beta^{-1}\ln p_{0}+\beta^{-1}\ln(2)+\mathcal{O}(\delta p).
\end{equation}

This yields:
\begin{align*}
    \mathbb{E}[\Delta W] &= -i\beta^{-1}\ln\frac{p_{0}}{p_{1}}+\beta^{-1}\ln p_{0}+\beta^{-1}\ln(2)+\mathcal{O}(\delta p) \\
    &= -\beta^{-1}\Tr(\phi_{i}\ln \mathbb{I}//2)-i\beta^{-1}\ln\frac{p_{0}}{p_{1}}+\beta^{-1}\ln p_{0}+\mathcal{O}(\delta p) \\
    &= \beta^{-1}[\D(\phi_{i}\|\mathbb{I}/2)+\ln p_{i}]+\mathcal{O}(\delta p).
\end{align*}
Taking the limit $M \rightarrow \infty$ (which implies $\delta p \rightarrow 0$), we arrive at:
\begin{equation}
    \lim_{M\rightarrow\infty}\mathbb{E}[\Delta W] = \beta^{-1}[\D(\phi_{i}\|\mathbb{I}/2)+\ln p_{i}] = w_i.
\end{equation}

To demonstrate convergence in probability, we note that $x_{l}(\nu(l)-\nu(l+1))\in[0,\frac{2}{p_{1}}\delta p]$ and apply Hoeffding's inequality:
\begin{equation}
    \Pr[|\Delta W-\mathbb{E}[\Delta W]|\ge\zeta] \le 2\exp\left(-\frac{\zeta^{2}}{\sum_{l=1}^{M-1}(\frac{2}{p_{1}}\delta p)^{2}}\right) \le 2\exp\left(-\frac{p_{1}^{2}\zeta^{2}M}{(2p_{0}-1)^{2}}\right).
\end{equation}
As $M \rightarrow \infty$, this probability vanishes, proving the theorem.
\end{proof}
\end{document}